%% file: msr-alex-main-arxiv.tex
\documentclass[sigconf, 10pt]{acmart}
\renewcommand\footnotetextcopyrightpermission[1]{} % removes footnote with conference information in first column
\usepackage[capitalize]{cleveref}
\usepackage{booktabs} % For formal tables
\usepackage{xspace}
\usepackage{algorithm,algorithmicx,algpseudocode}
\usepackage{subfig}
\usepackage{graphicx}
\usepackage[normalem]{ulem}
\usepackage[subtle,title=tight]{savetrees}
\usepackage{makecell}
\usepackage{microtype}
\usepackage{balance}
\usepackage{etoolbox}

\algnewcommand{\LineComment}[1]{\State \(\triangleright\) #1}

\newcommand{\atlex}{ALEX\xspace}

\newcommand{\bptree}{B+Tree\xspace}
\newcommand{\bptrees}{B+Trees\xspace}
\newcommand{\modelbtree}{Model B+Tree\xspace}
\newcommand{\kraskali}{Learned Index\xspace}
\newcommand{\art}{ART\xspace}
\newcommand{\expfactor}{{c}}

\newcommand{\firstrev}[1]{\textcolor{black}{#1}}
\newcommand{\secondrev}[1]{\textcolor{black}{#1}}
\newcommand{\thirdrev}[1]{\textcolor{black}{#1}}

\newtoggle{sigmod}
\togglefalse{sigmod}

\iftoggle{sigmod}{
	\newcommand{\sigmod}[1]{{#1}}
}{
	\newcommand{\sigmod}[1]{}
}

\begin{document}
\fancyhead{}
\title{ALEX: An Updatable Adaptive Learned Index}

\author{Jialin Ding{$^\dagger$}  \; Umar Farooq Minhas{$\ddagger$}  \; Jia Yu{$^\S$} \; Chi Wang{$\ddagger$}  \; Jaeyoung Do{$\ddagger$}  \; Yinan Li{$\ddagger$}  \; Hantian Zhang{$^\mp$}   \; Badrish Chandramouli{$\ddagger$}  \; Johannes Gehrke{$\ddagger$}  \; Donald Kossmann{$\ddagger$}  \; David Lomet{$\ddagger$}  \; Tim Kraska{$^\dagger$}
}
\affiliation{
	\institution{
		{$^\dag$}Massachusetts Institute of Technology \; {$^\ddag$}Microsoft Research  \; {$^\S$}Arizona State University\\ {$^\mp$}Georgia Institute of Technology \; 
	}
}
\thanks{* Work performed while at Microsoft Research.}

% The default list of authors is too long for headers.
\renewcommand{\shortauthors}{J. Ding et al.}
\renewcommand{\authors}{Jialin Ding, Umar Farooq Minhas, Jia Yu, Chi Wang, Jaeyoung Do, Yinan Li, Hantian Zhang, Badrish Chandramouli, Johannes Gehrke, Donald Kossmann, David Lomet, Tim Kraska}

\begin{abstract}
\input{sections/0-abstract}

\end{abstract}

%
% The code below should be generated by the tool at
% http://dl.acm.org/ccs.cfm
% Please copy and paste the code instead of the example below.
%
% \begin{CCSXML}
% <ccs2012>
% <concept>
% <concept_id>10002951.10002952.10002971.10003450</concept_id>
% <concept_desc>Information systems~Data access methods</concept_desc>
% <concept_significance>500</concept_significance>
% </concept>
% </ccs2012>
% \end{CCSXML}

% \ccsdesc[500]{Information systems~Data access methods}

% \keywords{indexing, machine learning, \bptree}

\maketitle

\input{sections/1-intro}

\input{sections/2-background}
\input{sections/3-atlex}

\input{sections/analysis-new}

\input{sections/4-eval}
\input{sections/5-related}
\input{sections/7-conclusion}

\vspace{4pt}\noindent{\bf Acknowledgements.}
This research is supported by Google, Intel, and Microsoft as part of the MIT Data Systems and AI Lab (DSAIL) at MIT, NSF IIS 1900933, DARPA Award 16-43-D3M-FP040, and the MIT Air Force Artificial Intelligence Innovation Accelerator (AIIA). 

\bibliographystyle{ACM-Reference-Format}
\balance
\bibliography{msr-alex-main-arxiv}

\clearpage
\nobalance
\iftoggle{sigmod}{
}{
	\input{sections/appendix-sigmod}
}

\end{document}

%% file: sections/0-abstract.tex
Recent work on ``learned indexes'' has changed the way we look at the
decades-old field of DBMS indexing. The key idea is that indexes can
be thought of as ``models'' that predict the position of a key in a
dataset. Indexes can, thus, be learned.  The original work by Kraska
et al.\ shows that a learned index beats a \bptree by a factor of up
to three in search time and by an order of magnitude in memory
footprint. However, it is limited to static, read-only workloads.

\begin{sloppypar}
In this paper, we present a new learned index called \atlex which addresses
practical issues that arise when implementing
learned indexes \secondrev{for workloads that contain a mix of point lookups, short range queries, inserts, updates, and deletes}. \atlex effectively combines the core insights from
learned indexes with proven storage and indexing techniques to achieve high
performance and low memory footprint. On read-only workloads, \atlex
beats the learned index from Kraska
et al.\ by up to 2.2$\times$ on performance with up to
15$\times$ smaller index size. Across the spectrum of read-write
workloads, \atlex beats \bptrees by up to 4.1$\times$ while never
performing worse, with up to 2000$\times$ smaller index size.  We
believe \atlex presents a key step towards making learned indexes
practical for a broader class of database workloads with dynamic
updates.
\end{sloppypar}

%% file: sections/1-intro.tex
\section{Introduction}

%% We are currently living through the age of Software 2.0, which refers
%% to the concept of replacing and augmenting human-written algorithms
%% with machine learning (ML) models. Such software has already shown
%% great results for web search, recommendation, speech understanding and
%% generation, video and image processing, robot control, and
%% self-driving vehicles.
Recent work by Kraska et al.~\cite{kraska2018case}, which we will
refer to as the \kraskali, proposes to replace a standard database
index with a hierarchy of machine learning (ML) models. Given a key,
an intermediate node in the hierarchy is a model to predict the child
model to use, and a leaf node in this hierarchy is a model to predict
the location of the key in a densely packed array
(\cref{fig:li_architecture}).  The models for this \kraskali are
trained from the data. Their key insight is that using (even
simple) models that 
%exploit knowledge of
adapt to the data distribution to
make a ``good enough'' guess of a key's actual location significantly
improves performance.
%For example, the results show that compared to
%the classic \bptree, the learned index has a factor of up to three
%faster lookups and over an order of magnitude less memory
%usage.
However, their solution can only handle lookups on
read-only data, with no support for update operations. This critical
drawback makes the \kraskali unusable for dynamic, read-write
workloads, common in practice.

In this work, we start by asking ourselves the following research
question: \emph{Can we design a new high performance index for dynamic workloads that effectively combines the core insights from the
  \kraskali with proven storage \& indexing techniques \secondrev{to deliver great performance in both time and space}?}  Our answer
is
%the design and
%implementation of 
a new in-memory index structure called \atlex, a 
fully dynamic data structure that \secondrev{simultaneously provides efficient support for point lookups, short range
queries, inserts, updates, deletes, and bulk loading.
This mix of operations is commonplace in online transaction processing (OLTP) workloads~\cite{cooper2010benchmarking,tpcc,rocksdb} and is also supported by \bptrees~\cite{dbms}.}
% In practice, the \bptree is the most popular index for efficiently supporting \secondrev{OLTP workloads}, and thus it serves as a good
% baseline. Therefore, \atlex aims to provide update times that are
% competitive with a \bptree, but lookup times that are comparable to the
% \kraskali. In addition, \atlex also aims to maintain the index size to
% be competitive to the \kraskali and, thus, much smaller than a
% \bptree. 
%If \atlex achieves these goals, it would be the first
%learned index to efficiently support dynamic workloads, thus
%effectively exploiting the advantages of \kraskali in a practical
%setting.

% Achieving these goals is not an easy task.
Implementing writes with
high performance requires a careful design of the underlying data
structure that stores records. \cite{kraska2018case}
uses a sorted, densely packed array which works well for static
datasets but can result in high costs for shifting records if new
records are inserted. Furthermore, the prediction accuracy of
the models can deteriorate as the data distribution changes over time,
requiring
%careful and selective 
repeated retraining.  To address these challenges, we make the following technical contributions in this paper:\

\begin{itemize}
\item \textbf{Storage layout optimized for models:} Similar to a
  \bptree, \atlex 
  %uses a \emph{node per leaf} layout to allow
  builds a tree, but allows
  different nodes to grow and shrink at different rates. To store
  records in a data node, \atlex uses an array with gaps, a
  \emph{Gapped Array}, which (1) amortizes the cost of shifting
  the keys for each insertion because gaps can absorb inserts, and (2)
  allows more accurate placement of data using \emph{model-based
    inserts} to ensure that records are located closely to the
  predicted position when possible.
  For efficient search, gaps are actually filled with adjacent keys.
  
\item \textbf{Search strategy optimized for models:} \atlex exploits
  model-based inserts combined with \emph{exponential search} starting
  from the predicted position. This always
  beats binary search when models are accurate.
  
\item \textbf{Keeping models accurate with dynamic data distributions
  and workloads:} \atlex provides robust performance
%lookup and insert performance,
  even when the data distribution is skewed or
  dynamically changes after index initialization. \atlex
  achieves this by exploiting adaptive expansion, and node splitting
  mechanisms, paired with selective model retraining, which is
  triggered by intelligent policies based on simple cost models. Our
  cost models take the actual workload into account and thus can
  effectively respond to dynamic changes in the workload. \atlex
  achieves all the above benefits without needing to hand-tune parameters for each dataset or workload.
  
\item \textbf{Detailed evaluation:} We present the results of an
  extensive experimental analysis with real-life datasets and varying
  read-write workloads and compare against state of the art indexes
  that support range queries.
\end{itemize}

% \todo{UPDATE RESULTS. 
% On read-only workloads, \atlex beats the
% \kraskali by up to 2.7$\times$ on performance with up to 3000$\times$
% smaller index size. \atlex achieves up to 3.5$\times$ higher
% performance than the \bptree while having up to 5 orders of magnitude
% smaller index size. On certain read-write workloads, \atlex achieves
% up to 3.3$\times$ higher throughput than the \bptree while having up
% to 2000$\times$ smaller index size. On other read-write workloads,
% \atlex is competitive with, or slightly worse than, a \bptree.}

On read-only workloads, \atlex beats the \kraskali by up to 2.2$\times$ on performance with up to 15$\times$ smaller index size. Across the spectrum of read-write workloads, \atlex beats \bptree by up to 4.1$\times$ while never performing worse, with up to 2000$\times$ smaller index size. \atlex also beats an ML-enhanced \bptree and the memory-optimized Adaptive Radix Tree, scales to large data sizes, and is robust to data distribution shift.

In the remainder of this paper, we give background (\cref{sec:background}), present the architecture of \atlex (\cref{sec:atlex}), describe the operations on \atlex (\cref{sec:algorithms}), present an analysis of \atlex performance (\cref{sec:analysis}), present experimental results (\cref{sec:exp}), review related work (\cref{sec:related}), and conclude (\cref{sec:conclusion}).
\sigmod{\thirdrev{Supplemental details in appendices are available in an extended technical report~\cite{techreport}.}}

%The remainder of this paper is organized as
%follows. \cref{sec:background} gives some background on \kraskali. \cref{sec:atlex} presents the 
%architecture of \atlex. \cref{sec:algorithms} describes 
%the operations on \atlex in detail.
%\cref{sec:analysis} presents an analysis of \atlex's performance. \cref{sec:exp}
%presents experimental results. \cref{sec:related} reviews
%related work, and we conclude the paper in \cref{sec:conclusion}.

%% file: sections/2-background.tex
\section{Background}
\label{sec:background}

%Set the stage for ``indexes are models''
\subsection{Traditional \bptree Indexes}
\label{subsec:bptree_indexes}

\emph{\bptree} is a classic range index structure.
%, which is a crucial component of database systems.
It is a height-balanced tree which
stores either the data (primary index) or pointers to the data
(secondary index) at the leaf level, in a sorted order to facilitate
range queries.

A \bptree lookup operation can be broken down into two
steps: (1) traverse to leaf, and (2) search within the leaf. Starting
at the root, traverse to leaf performs comparisons with the
keys stored in each node, and branches via stored pointers to the next
level.
When the tree is deep, the number of 
comparisons and branches can
be large, leading to many cache misses. Once traverse to leaf identifies the correct leaf page, typically a binary search is performed
to find the position of the key within the node, which might
incur additional cache misses.

The \bptree is a dynamic data structure that supports inserts,
updates, and deletes; is robust to data sizes and distributions; and is applicable in
many different scenarios, including in-memory and on-disk.
However, the generality of \bptree comes at a cost. In some cases knowledge of the data helps improve performance. As an extreme example, if the keys are consecutive integers, we can store the data in an array and perform lookup in ${\rm O}(1)$ time. A \bptree does not exploit such knowledge. Here, ``learning'' from the input data has an edge. 

\setlength{\textfloatsep}{1em}
\begin{figure}
	\includegraphics[width=0.8\columnwidth]{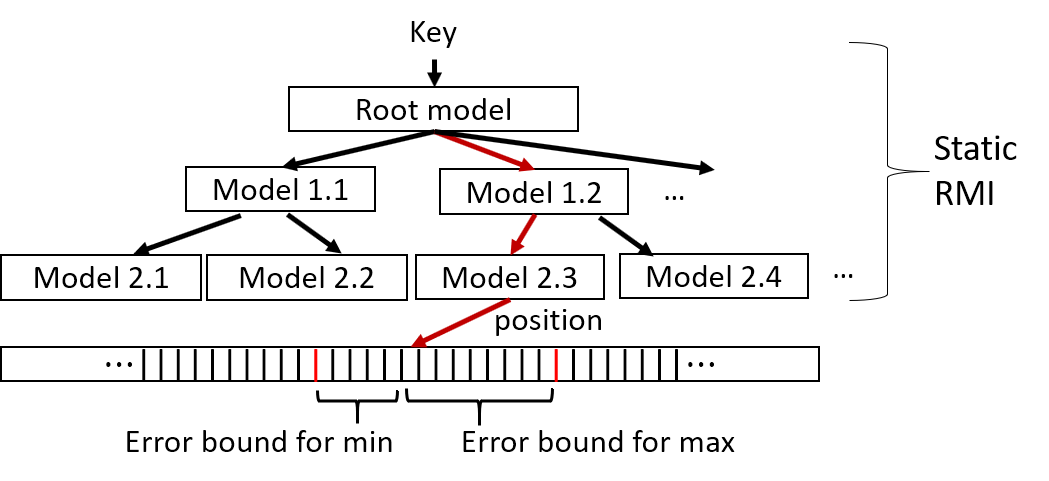}
	\vspace{-1.5em}
	\caption{Learned Index by Kraska et al.}
	\vspace{-0.5em}
	\label{fig:li_architecture}
\end{figure}

\subsection{The Case for Learned Indexes}
\label{subsec:learned_indexes}
Kraska et al.~\cite{kraska2018case} observed that \bptree indexes can
be thought of as models. Given a key, they predict the location of the
key within a sorted array (logically) at the leaf level. If indexes
are models, they can be learned using traditional ML
techniques
% More specifically, the task is to 
by learning the cumulative distribution function (CDF) of the input
data. The resulting \emph{\kraskali} is optimized for the
specific data distribution.

Another insight from Kraska et al.~is that a single ML model learned over
the entire data is not accurate enough because of the complexity of
the CDF. To overcome this, they introduce the \emph{recursive model
  index} (\emph{RMI})~\cite{kraska2018case}. RMI is a hierarchy
of models, with a \emph{static} depth of two or three, where a higher-level model picks the model at the next level, and so on, with the
leaf-level model making the final prediction for the position of the
key in the data structure (\cref{fig:li_architecture}). Logically, the RMI replaces the
internal \bptree nodes with models. The effect is that comparisons
and branches in internal \bptree nodes during traverse to leaf are
replaced by model inferences in a
\kraskali.

In~\cite{kraska2018case}, the keys are stored in an in-memory sorted
array. Given a key, the leaf-level model predicts the position
(array index) of the key. Since the model is not perfect, it could
make a wrong prediction. The insight is that if the leaf model
is accurate, a local search surrounding the predicted location
is faster than a binary search on the entire array. To
support local search, ~\cite{kraska2018case} keeps
\emph{min} and \emph{max} error bounds for each model in RMI and performs binary search within these bounds.

Last, each model in RMI can be a different type of model. Both
linear regression and neural network based models are considered
in~\cite{kraska2018case}. There is a trade-off between model accuracy
and model complexity.
%To get better accuracy, neural network models are used at the root,
The root of the RMI is tuned to be either a neural network or a linear regression, depending on which provides better performance,
while the simplicity and the speed of computation
for linear regression model is beneficial at the non-root
levels. A linear regression model can be represented as $y = \lfloor a*x + b \rfloor$,
where $x$ is the key and $y$ is the predicted position. A linear
regression model needs to store just two parameters $a$ and $b$, so storage overhead is low. The inference with a single
linear regression model requires only one 
multiplication, one addition and one rounding, which are fast to
execute on modern processors.

Unlike \bptree, which could have many internal levels, RMI uses two or
three levels. Also, the storage space required for models (two or
four 8-byte double values per model)
is much smaller than the storage space for internal nodes
in \bptree (which store keys and pointers). A \kraskali can
be an order of magnitude smaller in main memory storage (vs. internal \bptree nodes), while outperforming a \bptree in
lookup performance by a factor of up to three~\cite{kraska2018case}.

% \subsection{Current Limitations of Learned Indexes}
% \label{subsec:learned_indexes_limitations}
The main drawback of the \kraskali
%~\cite{kraska2018case} 
is that it does not support any modifications, including inserts,
updates, or deletes. Let us demonstrate a na\"ive insertion strategy for such an index.
\iffalse
Deletes are relatively easy to handle. Once the records to delete are
located in the array, we can simply wipe them off and flag the space
they occupied as free. Similarly, updates on non-key data are easy,
which do not change the index. The difficult cases are updates of keys
and insertions. To update a key, the data associated with the old key need to be first deleted, and
then inserted with the new key. So the challenge boils down to how
to perform insertion efficiently. 
\fi
Given a key $k$ to insert, we first use the model to find the insertion position for $k$. Then we create a new array whose length is one plus the length of the
old array. Next, we copy the data from the old array to the new array, where the
elements on the right of the insertion position are
shifted to the right by one position. We insert $k$ at the insertion position of the new array. Finally, we update the models to reflect the change in the data distribution.
Such a strategy has a linear time complexity with respect to the data
size, which is unacceptable in
practice. Kraska et al. suggest building delta-indexes to handle inserts~\cite{kraska2018case}, which is complementary to our strategy.
In this paper, we describe an alternative data structure to make modifications in a learned index more efficient.

\iffalse
A na\"ive insertion strategy for a
given key $k$ is as follows. 
\begin{itemize}
    \item Use the model to predict the position of $k$ in the sorted array.
    \item Use bounded binary search to find the actual position the key
      should be inserted at.
    \item Create a new array whose length is one plus the length of the
      old array.
    \item Copy the data from the old array to the new array, where the
      elements on the right of the insertion position for $k$ are
      shifted to the right by one position.
    \item Insert $k$ at the insertion position of the new array.
    \item Update the models to reflect the change in the data distribution.
\end{itemize}
\fi

%  and ways to retrain models incrementally.

% Further, as data are inserted, the RMI models get less accurate
% over time, which requires model retraining, further adding to the cost
% of inserts. 

%% file: sections/3-atlex.tex
\section{\atlex Overview}
\label{sec:atlex}

The \atlex design (\cref{fig:atlex_design}) takes advantage of two key insights.  First, we
propose a careful space-time trade-off that not only leads to an
updatable data structure, but is also faster for lookups. To explore
this trade-off, \atlex supports a \emph{Gapped Array (GA)} layout for
the leaf nodes, which we present in
\cref{subsec:node_layout}. Second, the \kraskali supports
static RMI (SRMI) only, where the number of levels and the number of
models in each level is fixed at initialization. SRMI performs poorly on inserts if the data distribution is difficult to model. \atlex can be
updated dynamically and efficiently at runtime and uses linear cost models that
predict the latency of lookup and insert operations based on simple
statistics measured from an RMI. \atlex uses these cost models to
initialize the RMI structure and to
dynamically adapt the RMI structure based on the workload.

%\subsection{Design Goals}
%\label{subsec:design_goals}
\atlex aims to
achieve the following goals w.r.t. the \bptree and
\kraskali. (1) Insert time should be competitive with \bptree, (2)
lookup time should be faster than \bptree and \kraskali, (3) index storage space should be smaller than
\bptree and \kraskali (4) data storage space (leaf level) should be comparable
to dynamic \bptree.  In general, data storage space will overshadow index
storage space, but the space benefit from smaller index storage
space is still important because it allows more
indexes to fit into the same memory budget.
% The \atlex design assumes that there will be no distribution shifts
% in the data. We will revisit this assumption in later sections.
% Achieving the above goals is quite challenging.
The rest of this
section describes how our \atlex design achieves these goals.
%and in
%doing so combines the best of \bptree and \kraskali in one data
%structure.

\setlength{\textfloatsep}{1em}
\begin{figure}
    \includegraphics[width=0.8\columnwidth]{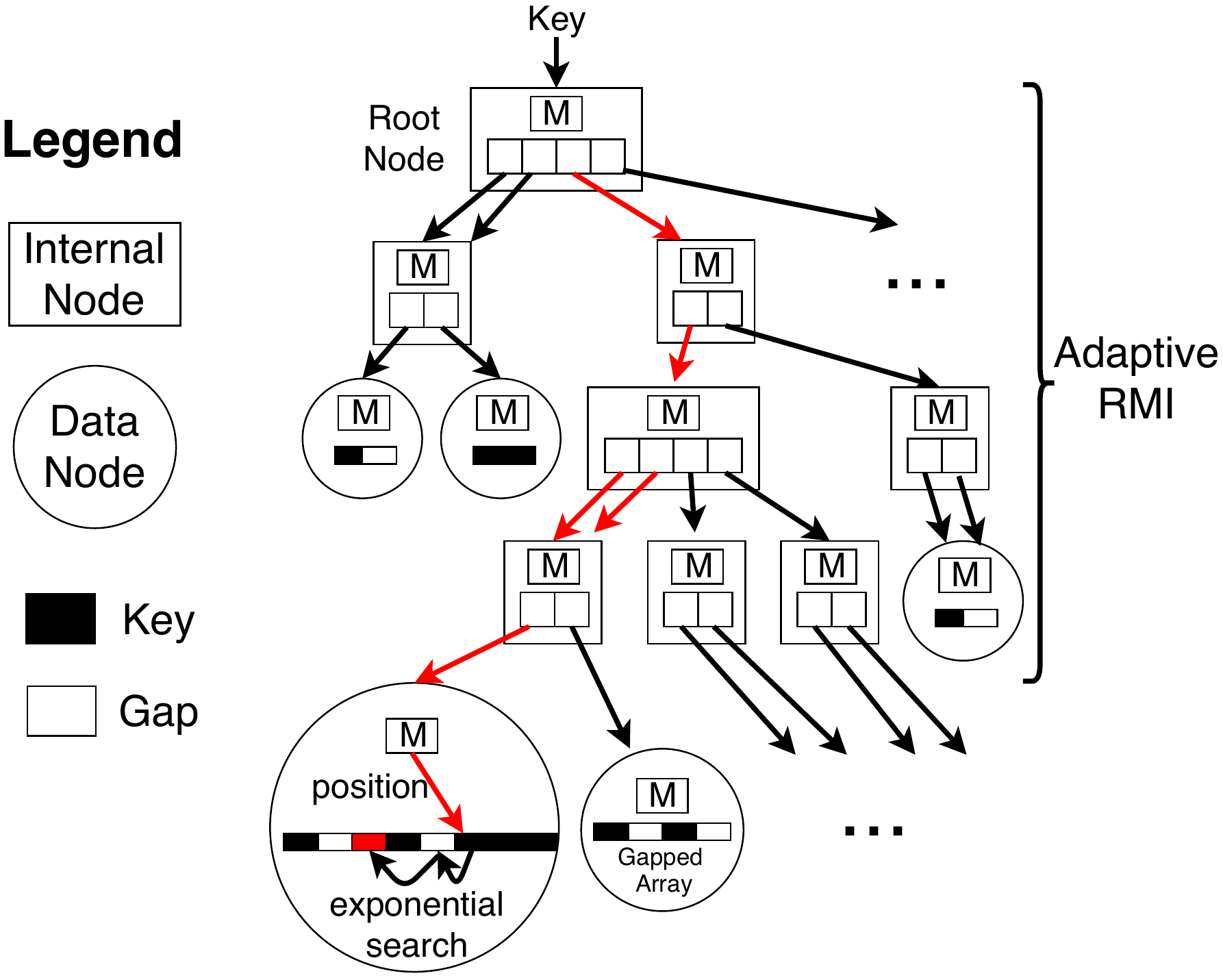}
    \vspace{-1.5em}
    \caption{\atlex Design}
    \vspace{-0.5em}
    \label{fig:atlex_design}
\end{figure}

\subsection{Design Overview}
\label{subsec:design_overview}
\atlex is an in-memory, updatable learned index.
\atlex has a number of differences from the \kraskali~\cite{kraska2018case}.
% For
% read-only workloads, \atlex can be instantiated with a two-layer RMI
% with a root model and some pre-determined number of leaf models, where
% each model, including the root model, is a linear regression
% model. This is similar to the design of the learned index presented
% in~\cite{kraska2018case}. However, this is where the similarities
% between the two end.

The first difference lies in the data structure
used to store the data at the leaf level. Like
\bptree, \atlex uses a \emph{node per leaf}. This allows the
individual nodes to expand and split more flexibly and also limits the number of shifts required during an
insert. In a typical \bptree, every leaf node stores
an array of keys and payloads and has ``free
space'' at the end of the array to absorb inserts. \atlex uses a
similar design but more carefully chooses how to use the free space. The insight is that by introducing gaps that are
strategically placed between elements of the array, we can achieve
faster insert and lookup times. As shown in
\cref{fig:atlex_design}, \atlex uses a Gapped Array (GA) layout
for each data node, which we describe in \cref{subsec:node_layout}.

The second difference is that \atlex uses exponential search to
find keys at the leaf level to correct mispredictions of the
RMI, as shown in \cref{fig:atlex_design}. In contrast,
\cite{kraska2018case} uses binary search within the error
bounds provided by the models. We experimentally verified that
exponential search without bounds is faster than binary search with
bounds (\cref{subsec:search_comparison}). This is because if the models are good, their prediction is
close enough to the correct position. Exponential search also removes the need to store error
bounds in the models of the RMI.

The third difference is that \atlex inserts keys into data nodes at the position where the
models predict that the key should be.  We call this {\em
  model-based insertion}.  In
contrast, the \kraskali produces an RMI on
an array of records without changing the position of records in the
array. Model-based insertion has better
search performance because it reduces model misprediction errors.

The fourth difference is that \atlex dynamically adjusts the
shape and height of the RMI depending on the workload. We
describe the design of initializing and dynamically growing the RMI structure in
\cref{sec:algorithms}.

The final difference is that \atlex has no parameters that
need to be re-tuned for each dataset or workload, unlike the \kraskali, in which the number of models must be tuned. \atlex automatically bulk loads and
adjusts the structure of RMI to achieve high performance by using a cost model.

% The fourth difference is that \atlex does not have any parameters that
% need to be hand-tuned, unlike the \kraskali, in which the number of models must be tuned.  Instead, \atlex automatically
% creates the RMI for a data distribution by using a cost
% model that predicts the time to perform lookup and insert operations
% based on the data, RMI structure, and optionally information about the
% workload (\cref{sec:algorithms}).  Using the cost model, \atlex can bulk load and dynamically
% adjust the structure of RMI to optimize performance.

\subsection{Node Layout}
\label{subsec:node_layout}

\subsubsection{Data Nodes}
Like a \bptree, the leaf nodes of \atlex store the data records and
thus are referred to as \emph{data nodes}, shown as circles in
\cref{fig:atlex_design}. A data node stores a linear regression
model (two double values for slope and intercept), which maps a key to
a position, and two Gapped Arrays (described below), one for
\emph{keys} and one for \emph{payloads}. We show only the
keys array in \cref{fig:atlex_design}. By default, both keys and
payloads are fixed-size. (Note that payloads could be records or pointers
to \emph{variable-sized} records, stored in separately allocated
spaces in memory). We also impose a \emph{max node size} for
practical reasons (see details in \cref{sec:algorithms}).

\atlex uses a \emph{Gapped Array} layout which uses model-based inserts to distribute extra space between the elements of the
array, thereby achieving faster inserts and lookups.
In contrast, \bptree places all the gaps at the end of the array.
%Gapped Arrays not only make \atlex
%updatable, they also make search faster by keeping the keys close
%to the predicted position.
Gapped Arrays fill
the gaps with the closest key to the right
of the gap, which helps maintain exponential search
performance. 
In order to efficiently skip gaps when scanning, each data node maintains a bitmap which tracks whether each location in the node is occupied by a key or is a gap.
The bitmap is fast to query and has low space overhead compared to the Gapped Array.
\firstrev{We compare Gapped Array to an existing gapped data structure called Packed Memory Array~\cite{bender2006adaptive} in \iftoggle{sigmod}{Appendix E}{\cref{sec:pma}}\sigmod{ in~\cite{techreport}}.}

\subsubsection{Internal Nodes}
We refer to all the nodes which are part of
the RMI structure as \emph{internal nodes}, shown as rectangles in \cref{fig:atlex_design}. Internal
nodes store a linear regression model and an array containing pointers
to children nodes. Like a \bptree, internal nodes direct traversals down
the tree, but unlike \bptree, internal nodes in \atlex use models to ``compute'' the location, in the pointers array, of
the next child pointer to follow.  Similar to data nodes, we impose a
\emph{max node size}.

%\subsubsection{Role of Internal Nodes in \atlex RMI}
\setlength{\textfloatsep}{1em}
\begin{figure}
    \subfloat{
        \includegraphics[width=0.6\columnwidth]{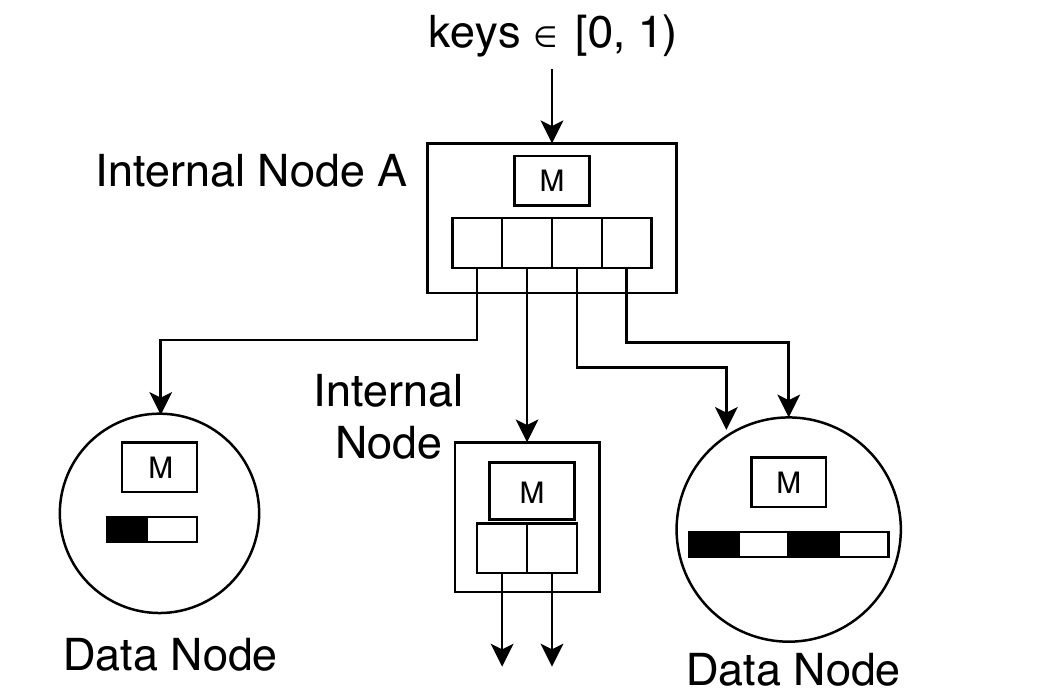}
    }
    ~
    \subfloat{
        \includegraphics[width=0.35\columnwidth,trim={10 5 5 5},clip]{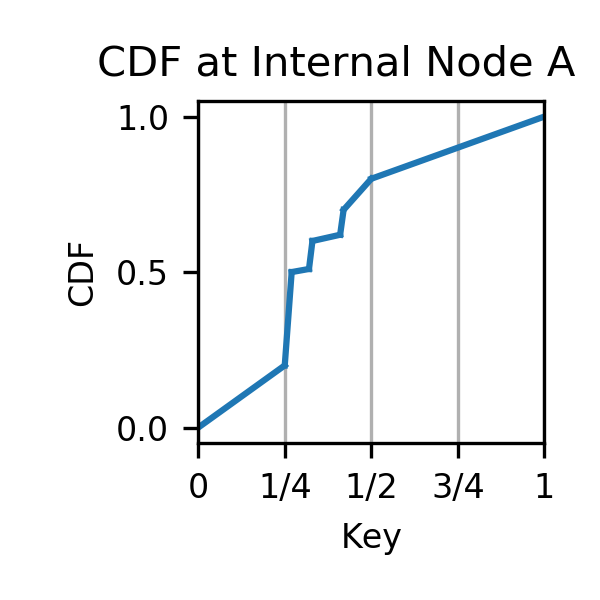}
        }
    \vspace{-1em}
    \caption{Internal nodes allow different resolutions in different parts of the key space $[0, 1)$.}
    \vspace{-0.5em}
    \label{fig:internal_nodes}
\end{figure}
The internal
nodes of \atlex serve a conceptually different purpose than those of
the \kraskali. \kraskali's internal nodes have models that
are fit to the data; an internal node with a
perfect model partitions keys equally to
its children, and an RMI with perfect internal nodes
results in an equal number of keys in each data node.  However, the goal
of the RMI structure is not to produce equally sized data nodes, but rather data nodes whose key distributions are
roughly linear, so that a linear model can be accurately
fit to its keys.

Therefore, the role of the internal nodes in \atlex 
is to provide a flexible way to partition the key space.
Suppose internal node A in \cref{fig:internal_nodes} covers the key space $[0, 1)$ and has four child pointers. A \kraskali would assign a node to each of these pointers, either all internal nodes or all data nodes.
However, \atlex more flexibly partitions the space.
Internal node A assigns the key spaces $[0, 1/4)$ and $[1/2,1)$ to data nodes (because the CDF in those spaces are linear), and assigns $[1/4,1/2)$ to another internal node (because the CDF is non-linear and the RMI requires more resolution into this key space).
% \cref{fig:internal_nodes} shows an example of how one internal node divide up the key space between 0 and 16.
% The CDF in the figure shows that different parts of the key space are locally linear.
% Instead of splitting the key space amongst 16 data nodes, which would result in adjacent data nodes with redundant models, the hierarchy of internal nodes is able to resolve different parts of the key space to different granularities, and is able to perfectly model the CDF with only 5 data nodes.
As shown in the figure, multiple pointers can point to the same child node; this is useful for handling inserts (\cref{subsubsec:node_splits}). We restrict the number of pointers in every internal node to always be a power of 2. This allows nodes to split without retraining its subtree (\cref{subsubsec:node_splits}).

\section{\atlex Algorithms}
\label{sec:algorithms}
In this section, we describe the algorithms for lookups, inserts
(including how to dynamically grow the RMI and the data nodes),
deletes, out of bounds inserts, and bulk load.

\subsection{\thirdrev{Lookups and Range Queries}}
To look up a key,
starting at the root node of the RMI, we iteratively use the model to ``compute''
a location in the pointers array, and we follow the pointer to a child node at the next level, until we reach a data node. By construction, the internal node models have perfect accuracy, so there is no search
involved in the internal nodes. We use the model in the data node to predict the position of the
search key in the \emph{keys} array, doing exponential search if
needed to find the actual position of the key. If a key is found, we
read the corresponding value at the same position from the
\emph{payloads} array and return the record. Else, we return a null
record.
% We describe the lookup logic in \cref{alg:node_layout}.
We visually show (using red arrows)
a lookup in \cref{fig:atlex_design}.
\thirdrev{A range query first performs a lookup to find the position and data node of the first key whose value is not less than the range's start value, then scans forward until reaching the range's end value, using the node's bitmap to skip over gaps and if necessary using pointers stored in the node to jump to the next data node.}

\subsection{Insert in non-full Data Node}
For the insert algorithm, the logic to reach the correct data node
(i.e., TraverseToLeaf) is the same as in the lookup algorithm
described above. In a non-full data node, to find the insertion
position for a new element, we use the model in the data node to
predict the insertion position. If the
predicted position is not correct (if inserting there would not maintain sorted order), we do exponential search to find
the correct insertion position.
If the insertion position is
a gap, then we insert the element into the gap and are done.
Else, we make a gap at the insertion position by
shifting the elements by one position in the direction of the closest
gap. We then insert the element into the newly created gap. The Gapped
Array achieves $O(\log{n})$ insertion time with high probability
\cite{bender2006insertion}.
% For the insert algorithm, the logic to reach the correct data node
% (i.e., TraverseToLeaf) is the same as in the lookup algorithm
% described above. In a non-full data node, to find the insertion
% position for a new element, we use the model in the data node to
% predict the insertion position. If this is the correct position, i.e.,
% inserting the key at this position maintains the sorted order, and is
% a gap, then we insert the element into the gap and are done. This is the best case for model-based inserts, as we place the key
% exactly where the model predicts and thus a later model-based lookup
% will result in a direct hit, thus we can do a lookup in $O(1)$. If the
% predicted position is not correct, we do exponential search to find
% the actual insertion position. Again, if the insertion position is a
% gap, then we insert the element there and are done.  If the insertion
% position is not a gap, we make a gap at the insertion position by
% shifting the elements by one position in the direction of the closest
% gap. We then insert the element into the newly created gap. The gapped
% array achieves $O(\log{n})$ insertion time with high probability
% \cite{bender2006insertion}.

\setlength{\textfloatsep}{1em}
\begin{figure}
    \includegraphics[width=0.45\columnwidth]{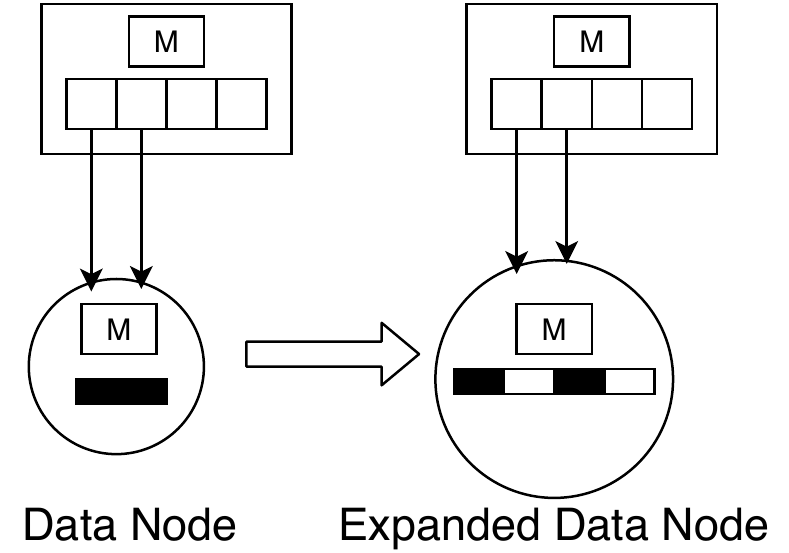}
    \vspace{-1em}
    \caption{Node Expansion}
    \vspace{-0.5em}
    \label{fig:expansion}
\end{figure}

\begin{figure*}[h!]
    \subfloat{
        \includegraphics[width=0.22\textwidth]{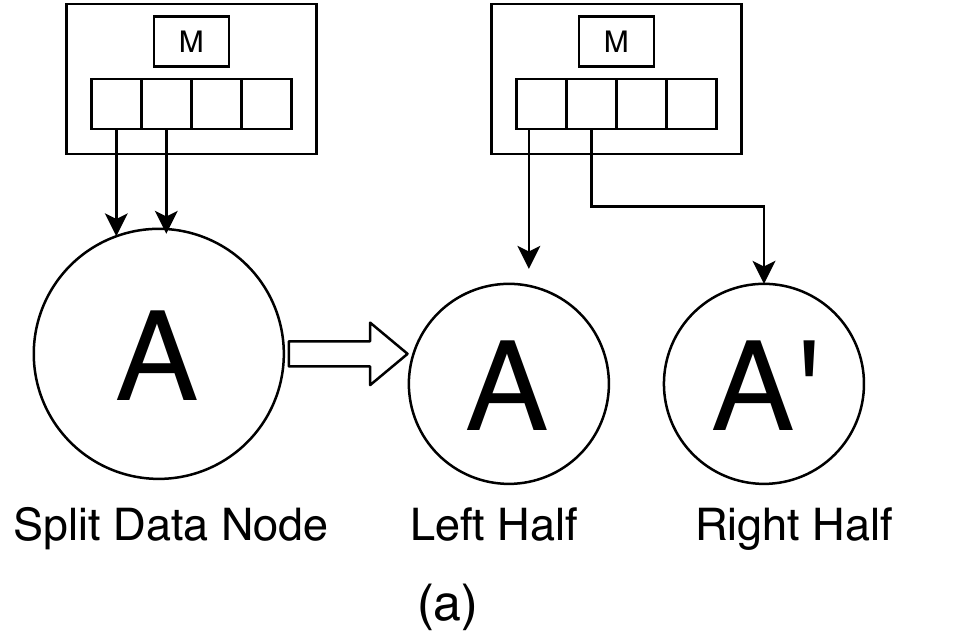}
        \label{fig:split_sideways_case1}
        }
    ~
    \subfloat{
        \includegraphics[width=0.45\textwidth]{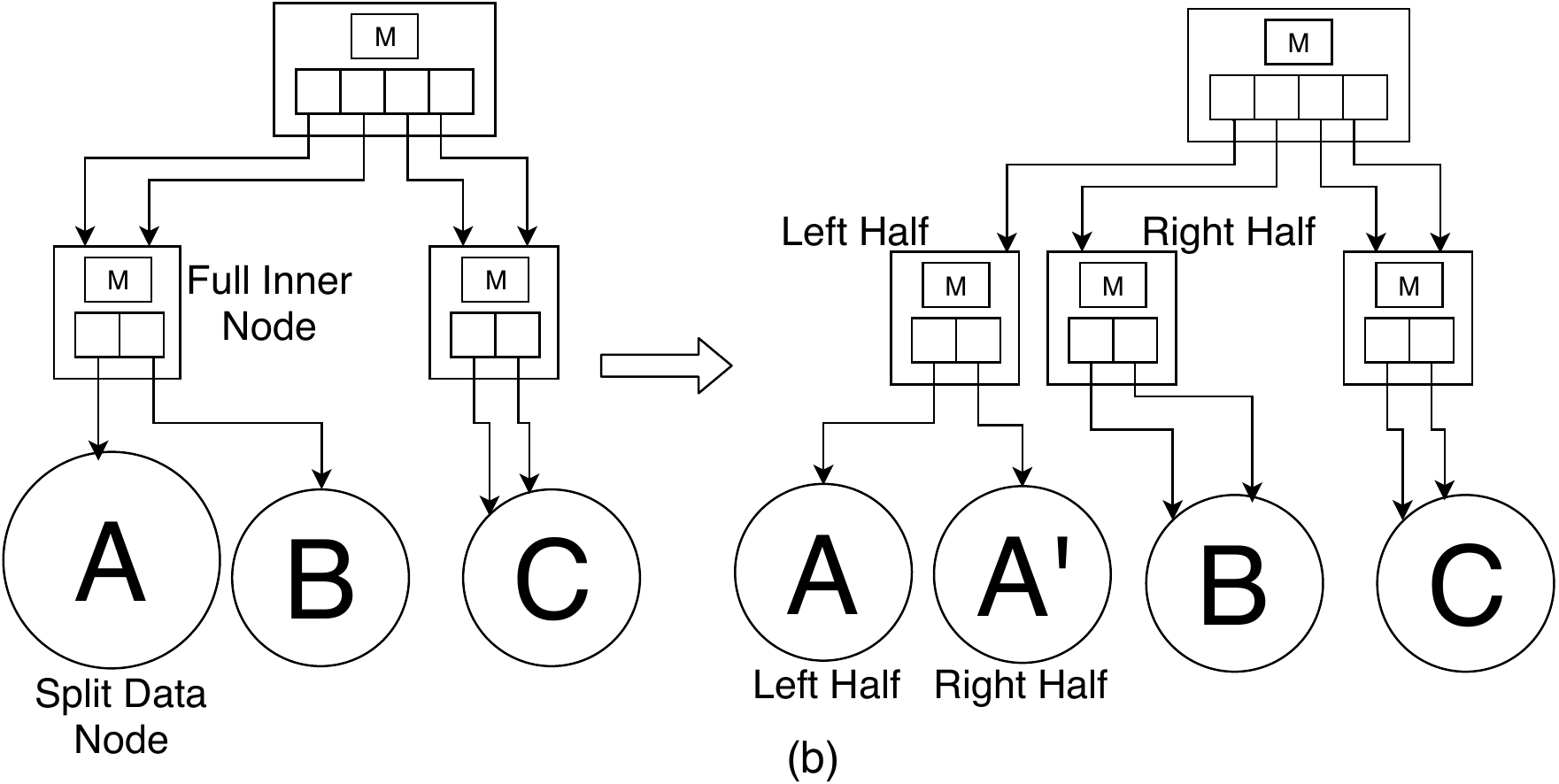}
        \label{fig:split_sideways_case3}
    }
    ~
    \subfloat{
        \includegraphics[width=0.22\textwidth]{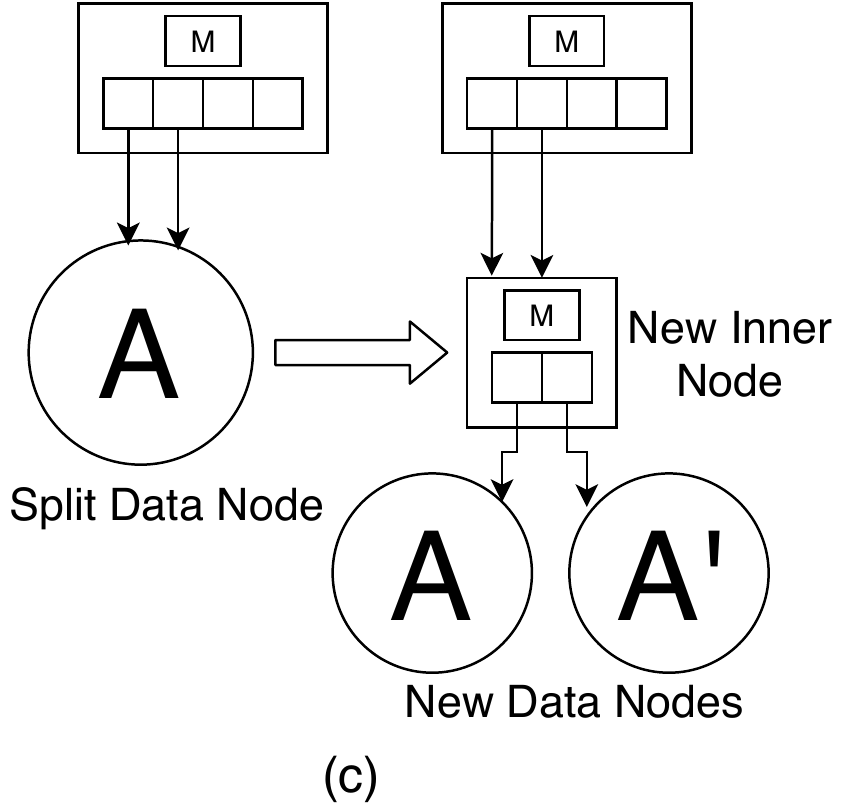}
        \label{fig:split_down}
    }
    \vspace{-1em}
    \caption{Node Splits}
    \vspace{-1em}
\end{figure*}

%% \begin{figure*}[h!]
%%   \centering
%%   \begin{subfigure}[b]{0.25\textwidth}
%%     \centering
%%     \includegraphics[width=0.25\textwidth]{figures/new/split_sideways_case1.pdf}{a}
%%     \label{fig:split_sideways_case1}
%%     \caption{Split Sideways}
%%   \end{subfigure}
%%     ~
%%   \begin{subfigure}[b]{0.25\textwidth}
%%     \centering
%%     \includegraphics[width=0.25\textwidth]{figures/new/split_sideways_case3.pdf}{b}
%%     \label{fig:split_sideways_case3}
%%     \caption{Inner Node Split}
%%   \end{subfigure}
%%     ~
%%   \begin{subfigure}[b]{0.25\textwidth}
%%     \centering
%%     \includegraphics[width=0.25\textwidth]{figures/new/split_down.pdf}{c}
%%     \label{fig:split_down}
%%     \caption{Split Down}
%%   \end{subfigure}
%%   \caption{Node Splits}
%%   \vspace{-1em}
%% \end{figure*}

\subsection{Insert in full Data Node}
\label{subsec:insert_full}
When a data node becomes full, \atlex uses two mechanisms to create
more space: expansions and splits. \atlex relies on simple cost models to pick between different
mechanisms. Below, we first define the notion of ``fullness,'' then describe the expansion and split
mechanisms, and the cost models. We then present the insertion
algorithm that combines the mechanisms with the cost models.
\cref{alg:ga}
summarizes the procedure for inserting into a data node.

\subsubsection{Criteria for Node Fullness}
\label{subsubsec:node_fullness}
\thirdrev{\atlex does not wait for a data node to become 100\% full, because insert performance on a Gapped Array will deteriorate as the number of gaps decreases.}
We introduce lower and upper density limits on the Gapped Array: $d_l, d_u \in (0, 1]$, with the constraint that $d_l < d_u$.
Density is defined as the
fraction of positions that are filled by elements.  A node is full if the next insert results in
exceeding $d_u$.
%We introduce an upper limit on the density of the Gapped Array $d_u \in (0, 1]$, defined as the
%fraction of positions that are filled by elements.  A node is full if the next insert results in
%exceeding $d_u$.
%We similarly define a lower density limit $d_l \in [0, 1)$.
% Users can set $d_u$ and $d_l$ so that the average density is the desired data storage utilization.
By default we set $d_l=0.6$ and $d_u=0.8$ to achieve average data storage utilization of $0.7$, similar to \bptree~\cite{graefe2011modern}, which in our experience always
produces good results and did not need to be tuned.
\thirdrev{In contrast, \bptree nodes typically have $d_l=0.5$ and $d_u=1$.}
\cref{sec:analysis} presents a
theoretical analysis of how the density of the Gapped Array provides a way to trade off between the space and the lookup
performance for \atlex.

% In \atlex, a node may become full because of two reasons: (1) there is no more free space left to insert new keys, (2) there may be free space in the node, but the model has
% become inaccurate.

%% As we insert more elements into the gapped array, insertion
%% performance degrades over time because the number of gaps decreases
%% which increases the expected number of elements that need to be
%% shifted for each insert.
% (1) \atlex does not wait for a node to become 100\% full. We introduce
% an upper limit on the density of the gapped array, defined as the
% fraction of the positions that are filled by elements.  A node is full
% (has no free space) if an additional insertion results in
% crossing the upper density limit $d \in (0, 1]$.

% (2) A node is full due to \emph{model capacity} if performance of
%   lookups and inserts for the data node has deviated too far from the
%   expected performance, as computed using the intra-node cost
%   models. This often happens when the distribution of keys that are
%   inserted does not follow the distribution of existing keys, which
%   results in the model becoming inaccurate. An inaccurate model
%   may result in a long contiguous region without any gaps, which we
%   call a \emph{fully-packed regions}. Inserting into a fully-packed
%   region requires shifting up to half of the elements within it to
%   create a gap, which in the worst case takes $O(n)$ time. Performance
%   may also degrade simply due to random noise as the node grows larger
%   or due to changing access patterns for lookups.

\subsubsection{Node Expansion Mechanism}
%Expansions are conceptually simple.
To expand a data node that contains $n$ keys, we allocate
a new larger Gapped Array with $n/d_l$ slots. We then either
scale or retrain the linear regression model, and then do model-based
inserts of all the elements in this new larger node using the scaled
or retrained model. After creation, the new data node is at the lower density limit $d_l$. \cref{fig:expansion} shows an example data node expansion where the Gapped Array inside the data node is expanded from two slots on the left to four slots on the right.
%Internal nodes can be expanded using the same mechanism.

\subsubsection{Node Split Mechanism}
\label{subsubsec:node_splits}
To split a data node in two, we allocate the keys to two new data nodes,
such that each new node is responsible for half of the key space
of the original node. \atlex supports two ways to split a node:

(1) \emph{Splitting sideways} is conceptually similar to how a \bptree
uses splits. There are two cases: (a) If the parent internal node of the
split data node is not yet at the \emph{max node size}, we replace the parent node's
pointers to the split data node with pointers to
the two new data nodes. The parent internal node's pointers array might have redundant pointers to the
split data node (\cref{fig:internal_nodes}). If so, we give half of the redundant pointers to each
of the two new nodes. Else, we create a second pointer to the split data node by
doubling the size of the parent node's pointers array and making a redundant copy for every pointer, and then give one of the redundant pointers to each of
the two new nodes. \cref{fig:split_sideways_case1} shows an
example of a sideways split that does not require an expansion of the
parent internal node. (b) If the parent internal node has reached \emph{max
  node size}, then we can choose to split the parent internal node, as we
show in \cref{fig:split_sideways_case3}. Note that by
restricting all the internal node sizes to be powers of 2, we can always
split a node in a ``boundary preserving'' way, and thus require no
retraining of any models below the split internal node. Note
that the split can propagate all the way to the root node, just like
in a \bptree.

(2) \emph{Splitting down} converts a data node into an internal node
with two child data nodes, as we show in
\cref{fig:split_down}. The models in the two child
data nodes are trained on their respective keys. \bptree does not have an analogous splitting down mechanism.

\subsubsection{Cost Models}
\label{subsubsec:cost_models}
To make decisions about which mechanism to apply (expansion or various types of splits), \atlex relies
on simple linear cost models that predict average lookup time and insert
time based on two simple statistics tracked at each data node: (a)
average number of exponential search iterations, and (b) average
number of shifts for inserts. Lookup performance is directly
correlated with (a) while insert performance is directly correlated
with (a) and (b) (since an insert first needs to do a lookup to find
the correct insertion position). These \emph{intra-node} cost models predict the time to perform operations within a data node.

These two statistics are not known when creating a data node.  To
find the \emph{expected cost} of a new data node, we compute the expected value of these statistics under the
assumption that lookups are done uniformly on the existing keys, and
inserts are done according to the existing key
distribution. Specifically, (a) is computed as the average base-2 logarithm of model prediction error for all
keys; (b) is computed as the average distance
to the closest gap in the Gapped Array for all existing keys.  These
expected values can be computed without creating the data node.  If the data
node is created using a subset of keys from an existing data node, we
can use the empirical ratio of lookups
vs. inserts to weight the relative importance of the two statistics
for computing the expected cost.

In addition to the intra-node cost model, \atlex uses a \emph{TraverseToLeaf} cost model to
predict the time for traversing from the root node to a data node. The TraverseToLeaf cost model uses two statistics: (1) the depth of the data node being traversed
to, and (2) the total
size (in bytes) of all inner nodes and data node metadata (i.e.,
everything except for the keys and payloads). These statistics
capture the cost of traversal: deeper data
nodes require more pointer chases to find, and larger size will
decrease CPU cache locality, which slows down the traversal to a
data node.
\thirdrev{We provide more details about the cost models and show their low usage overhead in \iftoggle{sigmod}{Appendix D}{\cref{sec:cost_drilldown}}\sigmod{ in~\cite{techreport}}.
}
  
\subsubsection{Insertion Algorithm}
\label{subsubsec:insertion}

%First, describe expansions to create more ``free space'', assume
%model is accurate

As lookups and inserts are done on the data node, we count the
number of exponential search iterations and shifts per insert. From
these statistics, we compute the \emph{empirical cost} of the data node using the intra-node cost model. Once the data node is full, we compare
the expected cost (computed at node creation time) to the empirical
cost. If they do not deviate significantly, then we conclude that the model is still accurate, and we perform node expansion (if the size after expansion is less
than the \emph{max node size}), scaling the model instead of retraining.  The
models in the internal nodes of the RMI are not retrained or rescaled.
\thirdrev{We define significant \emph{cost deviation} as occurring when the empirical cost is more than 50\% higher than the expected cost. In our experience, this cost deviation threshold of 50\% always produces good results and did not need to be tuned.}

Otherwise, if the empirical cost has deviated from the expected cost, we must either (i) expand the data
node and retrain the model, (ii) split the data
node sideways, or (iii) split the data node downwards.
We select the action that results in lowest expected cost, according to our intra-node cost model.
For simplicity, \atlex always splits a data node in
two. The data node could conceptually split into any power of 2, but deciding
the optimal fanout can be time-consuming, and we experimentally verified that a fanout of 2 is best according to the cost model in most cases.

%If we decide to split the data node, we allocate the keys to two new data nodes, such
%that each new node is responsible for half of the key value space of
%the original node.
%As described above, \atlex supports two ways to
%split: (i) if an ancestor internal node of the split data node (any node on the path up to the root) is not yet
%at the \emph{max node size}, we can split sideways and propagate up to that internal node, as shown in
%\cref{fig:split_sideways_case1} for the parent internal node and in \cref{fig:split_sideways_case3} for the grandparent internal node. (ii) We
%can \emph{split downwards} as shown in
%\cref{fig:split_down}. We use our cost models to decide on the best splitting mechanism. When splitting, all existing data nodes apart from the split data node remain unchanged, so intra-node cost is constant. Therefore, this decision only depends on the change in TraverseToLeaf cost.

\subsubsection{Why would empirical cost deviate from expected cost?} This often happens when the distribution of keys that are
inserted does not follow the distribution of existing keys, which
results in the model becoming inaccurate. An inaccurate model
may lead to long contiguous regions without any gaps. Inserting into these \emph{fully-packed regions} requires shifting up to half of the elements within it to
create a gap, which in the worst case takes $O(n)$ time. Performance
may also degrade simply due to random noise as the node grows larger
or due to changing access patterns for lookups.

% As described above, \atlex supports two ways to
% split: (i) if the parent inner node of the split data node is not yet
% at the \emph{max node size}, we split sideways, as shown in
% \cref{fig:split_sideways_case1}, (ii) if the parent inner node
% has reached the \emph{max node size} and there are no redundant
% pointers, we further have two choices: we can split the parent inner
% node, and post a new model entry into its parent node (data node's
% grandparent), as shown in \cref{fig:split_sideways_case3}. Or we
% can \emph{split downwards} as shown in
% \cref{fig:split_down}. We use our cost models to make
%   this decision. ~\todo{Update as per the findings from new
%   experiments.}

%%  If an additional element will push the
%% gapped array over its density bound $d$, then the gapped array
%% expands. If an additional element does not violate the density bound,
%% then insertion proceeds in the manner described previously.

\setlength{\textfloatsep}{1em}
\begin{algorithm}[t]
\caption{{\em Gapped Array} Insertion}
\small
\begin{algorithmic}[1]
\State {\textbf{struct} Node $\{$ keys[] (Gapped Array); num\_keys; $d_u, d_l$;
  \hskip\algorithmicindent model: key$\rightarrow [0,$ keys.size); $\}$}

\Procedure{Insert}{$key$}
   \If {num\_keys / keys.size >= $d_u$}
   \If {expected cost $\approx$ empirical cost}
	\State {Expand(retrain=False)}
%  \ElsIf {Cost(expanded node) < Cost(split nodes)}
%  \State {Expand(retrain=True)}
  \Else
  \State {Action with lowest cost /* described in Sec. \ref{subsubsec:insertion} */}
%  \State {(choices: Expand(retrain=True) or split)}
%  \State {Split in best way according to cost model}
  \EndIf
   \EndIf
   \State {predicted\_pos = model.predict(key)}
   \State {/* check for sorted order */}
   \State {insert\_pos = CorrectInsertPosition(predicted\_pos)}
   \If { keys[insert\_pos] is occupied }
	\State {MakeGap(insert\_pos) /* described in text */}
   \EndIf
   \State {keys[insert\_pos] = key}
   \State {num\_keys++}
\EndProcedure
\Procedure{Expand($retrain$)}{}
   \State {expanded\_size = num\_keys * 1/$d_l$}
   \State {/* allocate a new expanded array */}   
   \State {expanded\_keys = array(size=expanded\_size)}
   \If {retrain == True} 
   \State {model = /* train linear model on keys */}
   \Else
   \State {/* scale existing model to expanded array */}   
   \State {model *= expanded\_size / keys.size}%~~~~~~~~~~~~~~~{/*scale model*/}
   \EndIf
   \For{key : keys} \State{ModelBasedInsert(key)} \EndFor
   \State {keys = expanded\_keys}
\EndProcedure

\Procedure{ModelBasedInsert}{$key$}
    \State {insert\_pos = model.predict(key)}
    \If {keys[insert\_pos] is occupied}
       \State {insert\_pos = first gap to right of predicted\_pos}
    \EndIf
    \State {keys[insert\_pos] = key}
\EndProcedure

\end{algorithmic}
\label{alg:ga}
\end{algorithm}

%% \setlength{\textfloatsep}{1em}
%% \begin{figure}
%%     \includegraphics[width=0.8\columnwidth]{figures/reinsert.png}
%%     \vspace{-0.5em}
%%     \caption{Even though the array expands by 50\%, model-based insertion results
%%     in a fully-packed region on the left.~\todo{Candidate to be cut for space.}}
%%     \label{fig:fully-packed}
%% \end{figure}

%% In \cref{alg:node_layout}, we describe how to perform lookups
%% and expansions.  Lookups are performed by using the model to predict
%% the position of an element, then using exponential search from the
%% predicted position to find the actual position.  Expansions are
%% performed by creating an expanded data array; retraining the node's
%% linear model on the existing keys; rescaling the model to predict
%% positions in the expanded array; and using the model to perform
%% model-based inserts into the expanded array.  One detail is that the
%% node performs model-based inserts of elements in sorted order.  If the
%% model tries to insert multiple elements into the same position, every
%% element after the first will instead be inserted into the first gap to
%% the right.

%% \jialin{In the case of a ``cold start'' in which a node is initialized with no keys or very few keys, we do not build a model and instead perform lookups by doing binary search over the node, which is equivalent to what \bptree does.
%% Once the node contains a sufficient number of keys, we build and maintain a model in order to accelerate lookups.
%% }

\subsection{Delete, update, and other operations}
To delete a key, we do a lookup to find the location of the key, and
then remove it and its payload. Deletes do not shift any existing keys, so deletion is a strictly simpler operation than inserts and does not cause model accuracy
to degrade.
If a data node hits the lower density limit $d_l$
due to deletions, then we contract the
data node (i.e., the opposite of expanding the data node) in order to avoid low space utilization. Additionally, we can use intra-node cost models to determine that two data
nodes should merge together and potentially grow upwards, locally decreasing the RMI depth by 1. However,
for simplicity we do not implement these merging operations.

Updates that modify the key are implemented by
combining an insert and a delete. Updates that
only modify the payload will look up the key
and write the new value into the payload.
%\thirdrev{\bptree performs merge and diff by iterating over the sorted keys and bulk loading, and \atlex can do the same.} 
\thirdrev{Like \bptrees, we can merge two \atlex indexes or find the difference between two \atlex indexes by iterating over their sorted keys in tandem and bulk loading a new \atlex index.} 

\subsection{Handling out of bounds inserts}
A key that is lower or higher than the existing key
space would be inserted into the the left-most or
right-most data node, respectively. A series of out-of-bounds inserts, such
as an append-only insert workload, would result in poor
performance because that data node has no mechanism to split the
out-of-bounds key space. Therefore, \atlex has two ways to
smoothly handle out-of-bounds inserts. Assume that the
out-of-bounds inserts are to the right (e.g., inserted keys are
increasing); we apply analogous strategies when inserts are to the
left.

\setlength{\textfloatsep}{1em}
\begin{figure}
    \includegraphics[width=0.82\columnwidth]{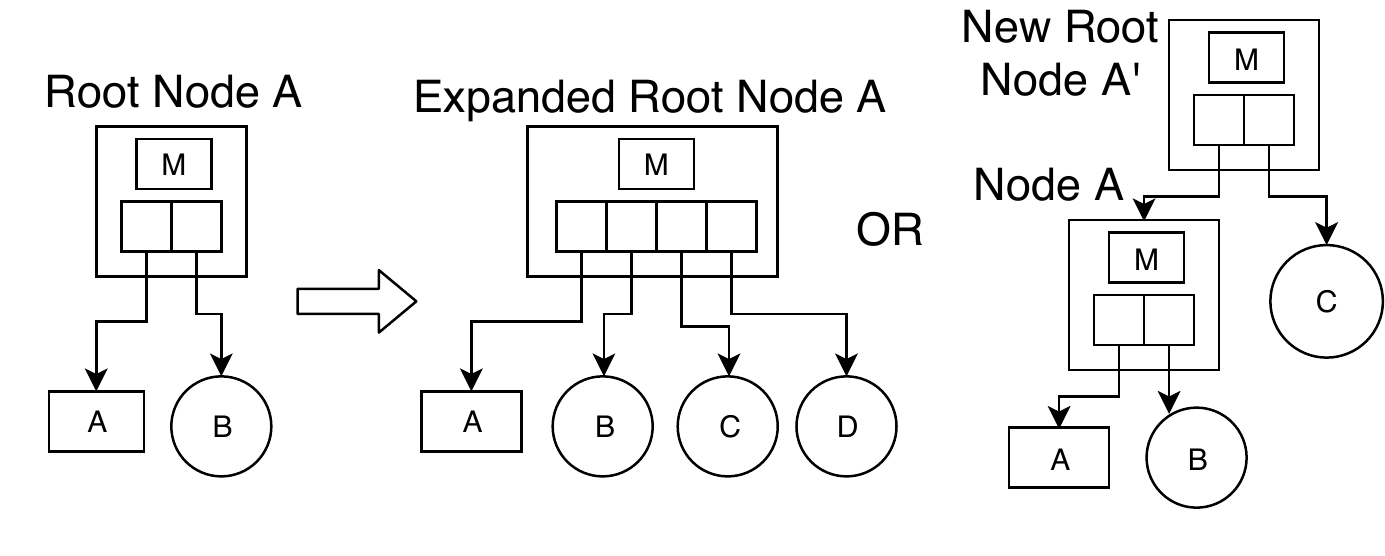}
    \vspace{-1.5em}
    \caption{Splitting the root}
    \label{fig:split_root}
\end{figure}

First, when an insert that is outside the existing key space is
detected, \atlex will \emph{expand the root node}, thereby expanding the key space, shown in \cref{fig:split_root}. We
expand the size of the child pointers array to the right. Existing
pointers to existing children are not modified. A new data node is
created for every new slot in the expanded pointers array. In case
this expansion would result in the root node exceeding the max node
size, \atlex will create a new root node. The first child pointer of
the new root node will point to the old root node, and a new data node
is created for every other pointer slot of the new root node. At
the end of this process, the out-of-bounds key will fall into one of
the newly created data nodes.

Second, the right-most data node of \atlex detects append-only
insertion behavior by maintaining the value of the maximum key in the
node and keeping a counter for how many times an insert exceeds that
maximum value. If most inserts exceed the maximum value, that implies append-only behavior, so the data node expands to the right without doing model-based
re-insertion; the expanded space is kept initially empty in anticipation of more
append-like inserts.
% Similarly, when splitting the data node,
% model-based inserts will only re-insert keys to the left portion of
% the data node, so that the right portion is left empty.

\subsection{Bulk Load}
\label{subsec:bulk_load}
\atlex supports a bulk load operation,
which is used in practice to index large amounts of data at
initialization or for rebuilding an index. 
Our goal is to find an RMI structure with minimum cost, defined as the expected average time to do an operation (i.e., lookup or insert) on this RMI.
Any \atlex operation is composed of TraverseToLeaf to the data node followed by an intra-node operation, so RMI cost is modeled by combining the TraverseToLeaf and intra-node cost models.

% \subsubsection{RMI cost model}
% \atlex uses an additional cost model, namely \emph{RMI cost model}, to
% find the optimal RMI structure when bulk loading. The RMI cost model
% predicts the time for \emph{TraverseToLeaf}. To build this cost model,
% we use two statistics: (1) the depth of the data node being traversed
% to, and (2) the total
% size (in bytes) of all inner nodes and data node metadata (i.e.,
% everything except for the keys and payloads). These two statistics
% capture the cost of \emph{TraverseToLeaf} quite well: deeper data
% nodes will require more pointer chases to find, and larger models will
% decrease CPU cache locality, which also slows down the traversal to a
% data node.

% % ~\todo{Name the intra-node cost models, use name here to cut text.}
% We additionally use the intra-node cost models from \cref{subsubsec:cost_models}. Any \atlex operation will be composed of TraverseToLeaf followed by an operation within the data node, so bulk loading will combine the costs from TraverseToLeaf and intra-node operations to form an overall cost for an RMI, which is the expected average time to do an operation on this RMI.

\subsubsection{Bulk Load Algorithm}
Using the cost models, we grow an RMI downwards greedily, starting
from the root node. At each node, we
independently make a decision about whether the node should be a data
node or an internal node, and in the latter case, what the fanout should
be. The fanout must be a power of 2, and child nodes will equally
divide the key space of the current node. Note that we can make
this decision locally for each node because we
use linear cost models, so decisions will
have a purely additive effect on the overall cost of the
RMI. If
we decide the node should be an internal node, we recurse on each of its
child nodes. This continues until all the data is loaded in \atlex.

\setlength{\textfloatsep}{1em}
\begin{figure}
    \includegraphics[width=2.5in]{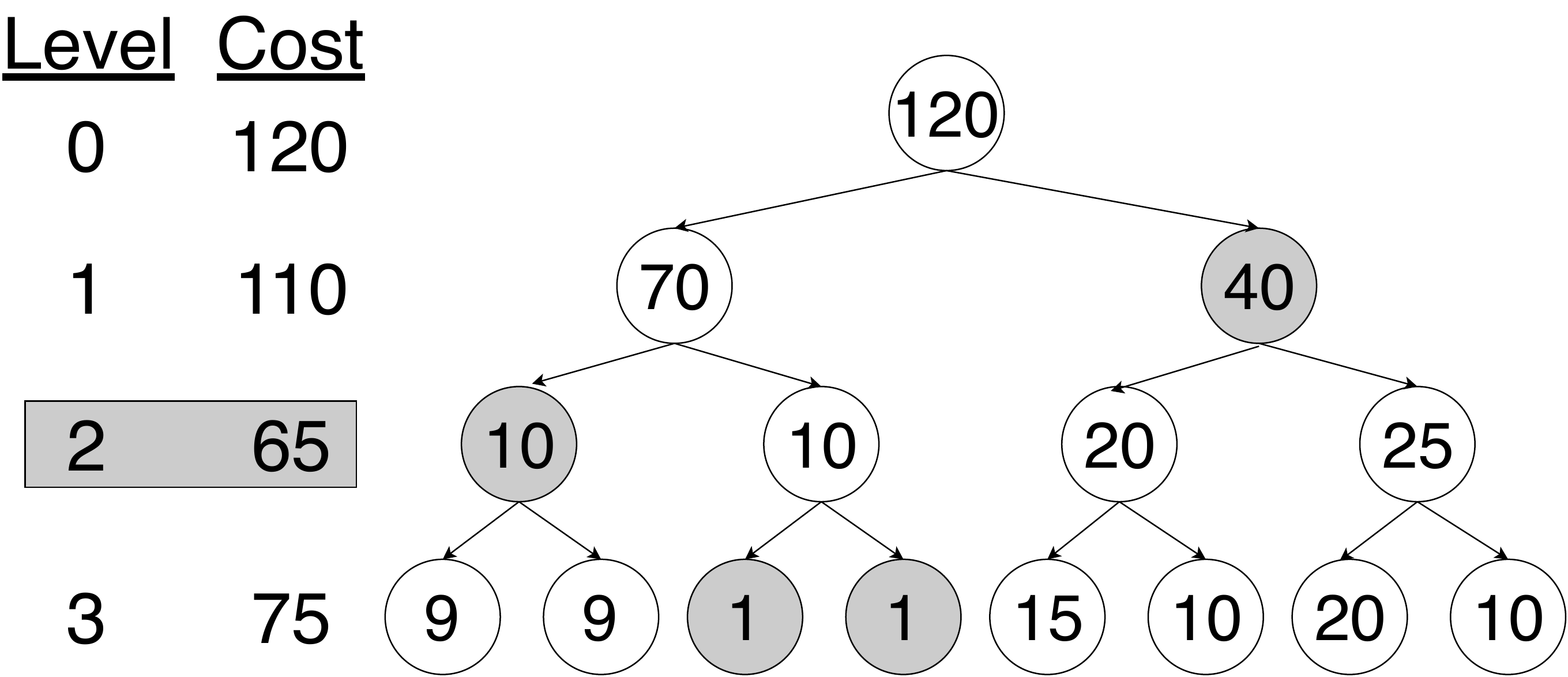}
    \vspace{-1em}
    \caption{Fanout Tree}
    \label{fig:fanout_tree}
\end{figure}

\subsubsection{The Fanout Tree}
As we grow the RMI, the main challenge is to determine the best
fanout at each node. We introduce the concept of
a \emph{fanout tree} (FT), which is a complete binary tree.
% ~\todo{the
%   previous sentence is not clear. Do you mean, the total number of
%   pointers it can store do not exceed the max size?}
An FT will help decide the fanout for a single RMI node;
in our bulk loading algorithm, we
construct an FT each time we want to
decide the best fanout for an RMI node. A fanout of 1 means that the RMI node should be a data node.

\cref{fig:fanout_tree} shows an example FT. Each FT node represents a possible child of the RMI node. If the key space of the RMI node is $[0, 1)$, then the $i$-th FT node on a level with $n$ children represents a child RMI node with key space $[i/n, (i+1)/n)$. Each FT node is associated with the expected cost of constructing a data node over its key space, as predicted by the intra-node cost models. Our goal is to find a set of FT nodes that cover the entire key space of the RMI node with minimum overall cost. The overall cost of a covering set is the sum of the costs of its FT nodes, as well as the TraverseToLeaf cost due to model size (e.g., going a level deeper in the FT means the RMI node must have twice as many pointers). This covering set determines the optimal fanout of the RMI node (i.e., the number of child pointers) as well as the optimal way to allocate child pointers.

We use the following method to find a low-cost covering set: (1) Starting from the FT root, grow
entire levels of the FT at a time, and compute the cost of
using each level as the covering set. Continue doing so until the
costs of each successive level start to increase. In \cref{fig:fanout_tree}, we find that level 2 has the lowest combined cost, and we do not keep growing after level 3. In concept, a deeper level might have lower cost, but computing the cost for each FT node is expensive. (2) Starting from
the level of the FT with lowest combined cost, we start merging or
splitting FT nodes locally. If the cost of two adjacent
FT nodes is higher than the cost of its parent, then we merge
(e.g., the nodes with cost 20 and 25 are merged to one with cost 40); this might
happen when the two nodes have very few keys, or when their
distributions are similar. In the other direction, if the cost of a
FT node is higher than the cost of its two children, we split
the FT node (e.g., the node with cost 10 is split into two nodes each with cost 1); this might happen when the two halves of the key
space have different distributions. We continue with this
process of merging and splitting adjacent nodes locally until it is no
longer possible. We return the resulting covering set
of FT nodes.

% ~\todo{Polish the text in last two sections. Add supporting figures.}

%% file: sections/analysis-new.tex
% the analysis for the tradeoff between space and search time for model-guided array and model-guided PMA
\section{Analysis of \atlex}
\label{sec:analysis}
%We provide bounds on the RMI depth, complexity analysis, and bounds on the performance of model-based search.
In this section, we provide bounds on the RMI depth and complexity analysis. Bounds on the performance of model-based search are found in \iftoggle{sigmod}{Appendix F}{\cref{sec:model_based_search_analysis}}\sigmod{ in~\cite{techreport}}.

\subsection{Bound on RMI depth}
% This sub-section describes an RMI construction and insertion policy that is different from the actual policies in \cref{subsec:insert_full,subsec:bulk_load}. However, this analysis is useful for gaining intuition about RMI depth.
% In this analysis, we define data node fullness only in terms of the maximum density bound $d$ (\cref{subsubsec:node_fullness}); we do not consider fullness due to model capacity.
In this section we present a worst-case bound on maximum RMI depth and describe how to achieve it.
Note that the goal of \atlex is to maximize performance, not to minimize tree depth; though the two are correlated, the latter is simply a proxy for the former (e.g., depth is one input to our cost models).
Therefore, this analysis is useful for gaining intuition about RMI depth, but does not reflect worst-case guarantees in practice.

Let $m$ be the maximum node size, defined in number of slots (in the pointers array for internal nodes, in the Gapped Array for data nodes). We constrain node size to be a power of 2: $m=2^k$. Internal nodes can have up to $m$ child pointers, and data nodes must contain no more than $md_u$ keys.
Let all keys to be indexed fall within the key space $s$.
Let $p$ be the minimum number of partitions such that when the key space $s$ is divided into $p$ partitions of equal width, every partition contains no more than $md_u$ keys.
Define the root node depth as 0.
\begin{theorem}\label{th:depth}
    We can construct an RMI that satisfies the max node size and upper density limit constraints whose depth is no larger than $\lceil\log_m{p}\rceil$---we call this the maximal depth. Furthermore, we can maintain maximal depth under inserts. (Note that $p$ might change under inserts.)
\end{theorem}
In other words, the depth of the RMI is bounded by the density of the densest subregion of $s$. In contrast, \bptrees bound depth as a function of the number of keys. \cref{th:depth} can also be applied to a subspace within $s$, which would correspond to some subtree within the RMI.

\begin{proof}
Constructing an RMI with maximal depth is straightforward. The densest subregion, which spans a key space of size $|s|/p$, is allocated to a data node. The traversal path from the root to this densest region is composed of internal nodes, each with $m$ child pointers.
It takes $\lceil\log_m{p}\rceil$ internal nodes to narrow the key space size from $|s|$ to $|s|/p$. To minimize depth in other subtrees of the RMI, we apply this construction mechanism recursively to the remaining parts of the space $s$.

Starting from an RMI that satisfies maximal depth, we maintain maximal depth using the mechanisms in \cref{subsec:insert_full} under the following policy:
(1) Data nodes expand until they reach max node size.
(2) When a data node must split due to max node size, it splits sideways to maintain current depth (potentially propagating the split up to some ancestor internal node).
(3) When splitting sideways is no longer possible (all ancestor nodes are at max node size), split downwards.
By following this policy, RMI only splits downward when $p$ grows by a factor of $m$, thereby maintaining maximal depth.
\end{proof}

\subsection{Complexity analysis}
Here we provide complexity of lookups and inserts, as well as the mechanisms from \cref{subsec:insert_full}.
Both lookups and inserts do TraverseToLeaf in $\lceil\log_m{p}\rceil$ time. Within the data node, exponential search for lookups is bounded in the worst case by $O(\log{m})$. In the best case, the data node model predicts the key's position perfectly, and lookup takes $O(1)$ time. We show in the next sub-section that we can reduce exponential search time according to a space-time trade-off.

Inserts into a non-full node are composed of a lookup, potentially followed by shifts to introduce a gap for the new key. This is bounded in the worst case by $O(m)$, but since Gapped Array achieves $O(\log{m})$ shifts per insert with high probability~\cite{bender2006insertion}, we expect $O(\log{m})$ complexity in most cases. In the best case, the predicted insertion position is correct and is a gap, and we place the key exactly where the model predicts for insert complexity of $O(1)$; furthermore, a later model-based lookup
will result in a direct hit in $O(1)$.

There are three important mechanisms in \cref{subsec:insert_full}, whose costs are defined by how many elements must be copied: (1) Expansion of a data node, whose cost is bounded by $O(m)$. (2) Splitting downwards into two nodes, whose cost is bounded by $O(m)$. (3) Splitting sideways into two nodes and propagating upwards in the path to some ancestor node, whose cost is bounded by $O(m\lceil\log_m{p}\rceil)$ because every internal node on this path must also split.
%Since splits are ``boundary-preserving,'' nodes not on this path do not need to change.
As a result, the worst-case performance for insert into a full node is $O(m\lceil\log_m{p}\rceil)$.

%% file: sections/4-eval.tex
\section{Evaluation}\label{sec:exp}

% ~\todo{WARNING: ALL GRAPHS HAVE BEEN UPDATED BUT NOT THE TEXT SINCE WE ARE STILL RE-RUNNING SOME EXPERIMENTS. BUT PLEASE LOOK AT THE GRAPHS AND PROVIDE FEEDBACK.}

We compare \atlex with the \kraskali, \bptree, a model-enhanced \bptree, and Adaptive Radix Tree (\art), using a variety of datasets and workloads.
%We conclude this section with a drilldown into \atlex performance.
This evaluation demonstrates that:

\begin{itemize}
\item On read-only workloads, \atlex achieves up to 4.1$\times$, 2.2$\times$, 2.9$\times$, 3.0$\times$ higher throughput and 800$\times$, 15$\times$, 160$\times$, 8000$\times$ smaller index size than the \bptree, \kraskali, \modelbtree, and \art, respectively.

\item On read-write workloads, \atlex achieves up to 4.0$\times$, 2.7$\times$, 2.7$\times$ higher throughput and 2000$\times$, 475$\times$, 36000$\times$ smaller index size than the \bptree, \modelbtree, and \art, respectively.

% \item On read-only workloads, \atlex achieves up to 4.3$\times$ higher throughput than the \bptree while having up to 800$\times$ smaller index size. \atlex also achieves up to 1.8$\times$ higher throughput and up to 9$\times$ smaller index size than the \kraskali.

% \item On read-write workloads, \atlex achieves up to 3.4$\times$ higher throughput than the \bptree while having up to 1000$\times$ smaller index size.

\item \atlex has competitive bulk load times and maintains an advantage over other indexes when scaling to larger datasets and under distribution shift due to data skew.

\item Gapped Array and the adaptive RMI structure allow \atlex to adapt to different datasets and workloads.
\end{itemize}

% \noindent	$\bullet$ On read-only workloads, \atlex achieves up to 3.5$\times$ higher throughput than the \bptree while having up to 5 orders of magnitude smaller index size. \atlex also achieves up to 2.7$\times$ higher throughput and up to 3000$\times$ smaller index size than the \kraskali.

% \noindent	$\bullet$ On read-write workloads, \atlex achieves up to 3.3$\times$ higher throughput than the \bptree while having up to 2000$\times$ smaller index size.

% \noindent	$\bullet$ \atlex maintains an advantage over the \bptree when scaling to larger datasets and remains competitive with the \bptree under moderate dataset distribution shift.

% \noindent	$\bullet$ Flexible node layout and adaptive RMI allows \atlex to adapt to different datasets and workloads.
	% \item On adversarial datasets, \atlex can underperform \bptree on throughput by up to 40\%, and adversarial distribution shift can result in up to 26$\times$ lower throughput than \bptree.

\subsection{Experimental Setup}
\label{sec:exp_setup}
We implement \atlex in C++\footnote{\url{https://github.com/microsoft/ALEX}}.  We perform our evaluation via
single-threaded experiments on an Ubuntu Linux machine with \firstrev{Intel Core i9-9900K}
3.6GHz CPU and 64GB RAM.  We compare \atlex against four baselines.
(1) A standard \bptree, as implemented in the STX
\bptree~\cite{stx}. (2) Our best-effort
reimplementation of the \kraskali~\cite{kraska2018case}, using a two-level RMI with linear models at each node and binary search for lookups.\footnote{In private communication with the authors of~\cite{kraska2018case}, we learned that the added complexity of using a neural net for the root model usually is not justified by the resulting minor performance gains, which we also independently verified.}
(3) \modelbtree, which maintains a linear model in every node of the \bptree, stores each node as a Gapped Array, and uses model-based exponential search instead of binary search, implemented on top of~\cite{stx}; this shows the benefit of using models while keeping the fundamental \bptree structure.
(4) Adaptive Radix Tree (\art)~\cite{leis2013adaptive}, a trie that adapts to the data which is optimized for main memory indexing, implemented in C~\cite{artimpl}.
\secondrev{Since \atlex supports all operations common in OLTP workloads, we do not compare to hash tables and dynamic hashing techniques, which cannot efficiently support range queries.}

For each dataset and workload, we use grid search to tune the page size for \bptree and \modelbtree and the number of models for \kraskali
to achieve the best throughput.
%The STX \bptree does not have any obvious tunable parameters other than page size.
In contrast, no tuning is necessary for \atlex, unless users place additional constraints.
For example, users might want to bound the latency of a single operation. We set a max node size of 16MB to achieve tail latency (99.9th percentile) of around 2$\mu$s per operation, but max node size can be adjusted according to user's desired limits (\cref{fig:latency}).
%Most nodes have size much smaller than the max node size (\cref{tab:node_size}).

Index size of \atlex and \kraskali is the sum of the sizes of all models used in the index and metadata; index size for \atlex also includes internal node pointers. For \atlex, each linear model consists of two 64-bit doubles which represent the slope and intercept. \kraskali keeps two additional integers per model that represent the error bounds. The index size of \bptree and \modelbtree is the sum of the sizes of all inner nodes, which for \modelbtree includes the models in each node.
The index size of \art is the sum of inner node sizes minus the total size of keys, since keys are encoded into the inner nodes.
The data size of \atlex is the sum of the sizes of the arrays containing the keys and payloads, including gaps, as well as the bitmap in each data node. The data size of \bptree is the sum of the sizes of all leaf nodes.
At initialization, the Gapped Arrays in data nodes are set to have 70\% space utilization, comparable to \bptree leaf node space utilization~\cite{graefe2011modern}.

\subsubsection{Datasets}
\label{subsubsec:datasets}
We run all experiments using 8-byte keys from some dataset and randomly generated fixed-size payloads.
We evaluate \atlex on 4 datasets, whose characteristics and CDFs are
shown in \cref{tab:data_params} and \cref{fig:cdfs}.  The \textit{longitudes} dataset
consists of the longitudes of locations around the world from Open
Street Maps~\cite{openstreetmap}.  The \textit{longlat} dataset
consists of compound keys that combine longitudes and
latitudes from Open Street Maps by applying the transformation $k = 180\cdot\text{floor}(\text{longitude})+\text{latitude}$ to every pair of longitude and latitude.
The resulting distribution of keys $k$ is highly non-linear.
% , ordered so that locations within one degree of longitude have smaller value than locations in the next degree of longitude, and locations within a degree of longitude are ordered by latitude.
The \textit{lognormal}
dataset has values generated according to a lognormal distribution with $\mu=0$ and $\sigma=2$, multiplied by $10^9$ and rounded down to the nearest integer.
The \textit{YCSB} dataset has values representing user IDs generated according to the YCSB Benchmark~\cite{cooper2010benchmarking}, which \firstrev{are uniformly distributed across the full 64-bit domain}, and uses an 80-byte payload.  These datasets do not
contain duplicate values.  Unless otherwise stated, these datasets are
randomly shuffled to simulate a uniform dataset distribution over
time.
% A more detailed description of dataset creation and analysis of dataset characteristics can be found in Appendix~\ref{sec:appendix_dataset}.

\begin{figure}[]
	\vspace{-0.5em}
    \subfloat{
        \includegraphics[width=\columnwidth,trim={5 5 5 5},clip]{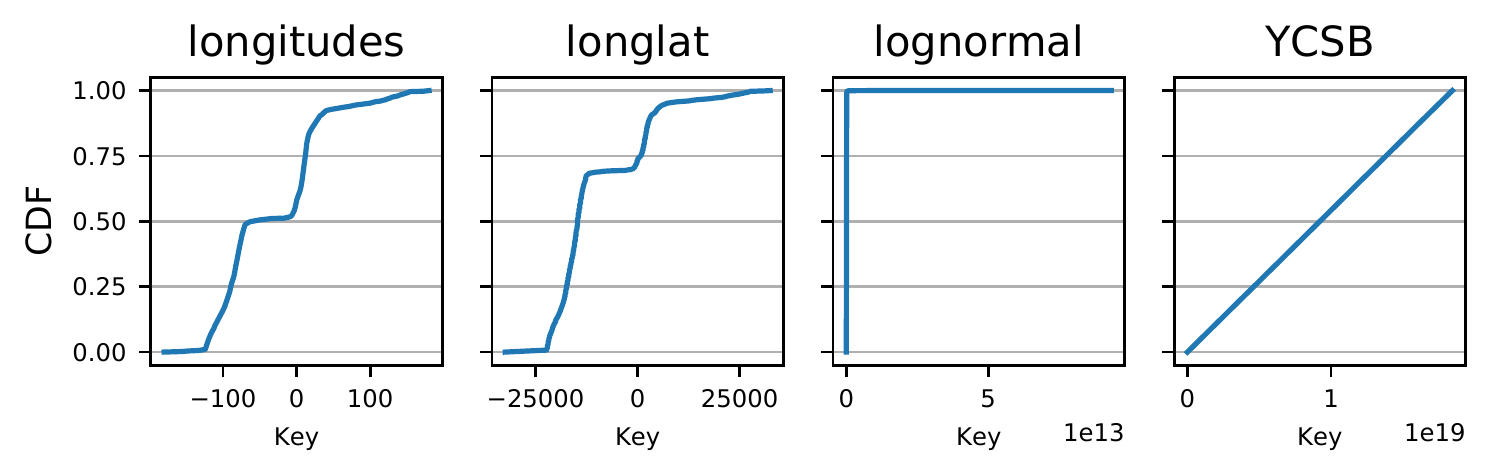}
        \label{fig:cdf_all}
        }
    \\
    \vspace{-0.5em}
    \subfloat{
        \includegraphics[width=0.48\columnwidth,trim={5 5 5 5},clip]{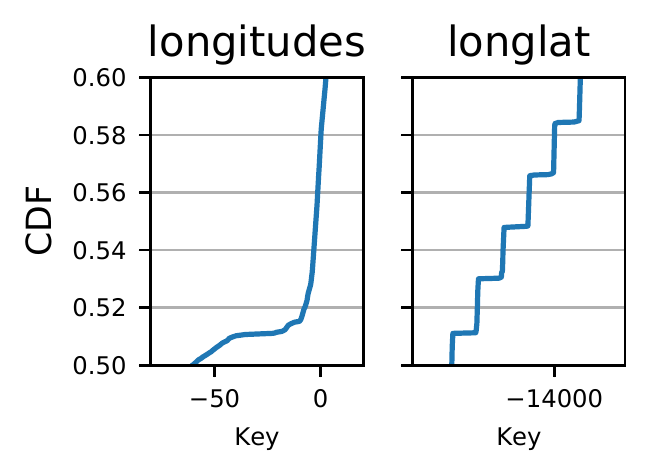}
        \label{fig:cdf_zoom1}
        }
    ~
    \subfloat{
        \includegraphics[width=0.48\columnwidth,trim={5 5 5 5},clip]{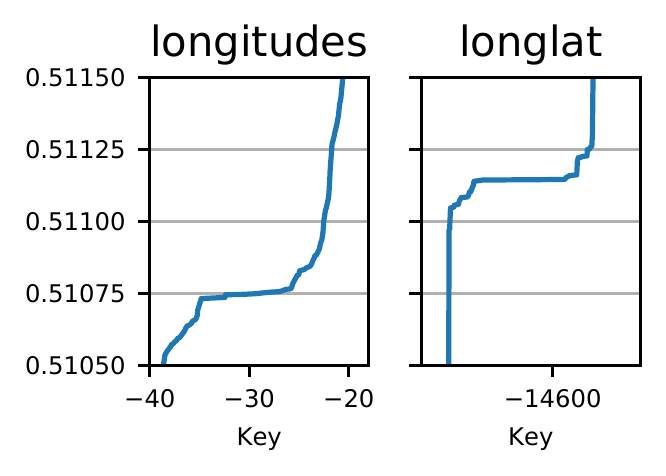}
        \label{fig:cdf_zoom2}
        }
        \vspace{-0.5em}
    \caption{
        Dataset CDFs, and zoomed-in CDFs.
    }
    \label{fig:cdfs}
    \vspace{-1em}
\end{figure}

\begin{table}[]
\centering
\caption{Dataset Characteristics}
\vspace{-1em}
\small
\label{tab:data_params}
\begin{tabular}{@{}lllll@{}}
\toprule
       & \textbf{longitudes}     & \textbf{longlat}  & \textbf{lognormal} & \textbf{YCSB} \\
\midrule
\textbf{Num keys}  & 1B & 200M & 190M & 200M \\
\textbf{Key type}  & double & double & 64-bit int & 64-bit int \\
\textbf{Payload size}  & 8B & 8B & 8B   & 80B    \\
\textbf{Total size}  & 16GB & 3.2GB & 3.04GB & 17.6GB \\
% \begin{tabular}{@{}c@{}}\textbf{Read-only} \\ \textbf{init size}\end{tabular}  & 200M     & 200M  & 190M    & 200M \\
% \begin{tabular}{@{}c@{}}\textbf{Read-write} \\ \textbf{init size}\end{tabular}  & 50M     & 50M  & 50M    & 50M \\
\bottomrule
\end{tabular}
\end{table}

\begin{figure*}[th!]
    \centering
    \includegraphics[width=0.6\textwidth]
            {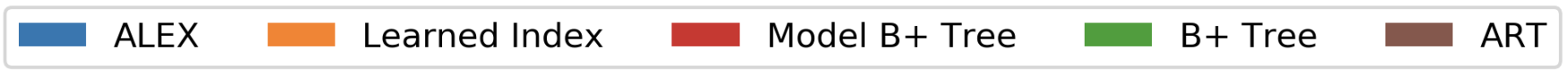}
    \\[-1ex] % negative space
    \subfloat{
        \includegraphics[width=0.19\textwidth,trim={10 10 10 8},clip]{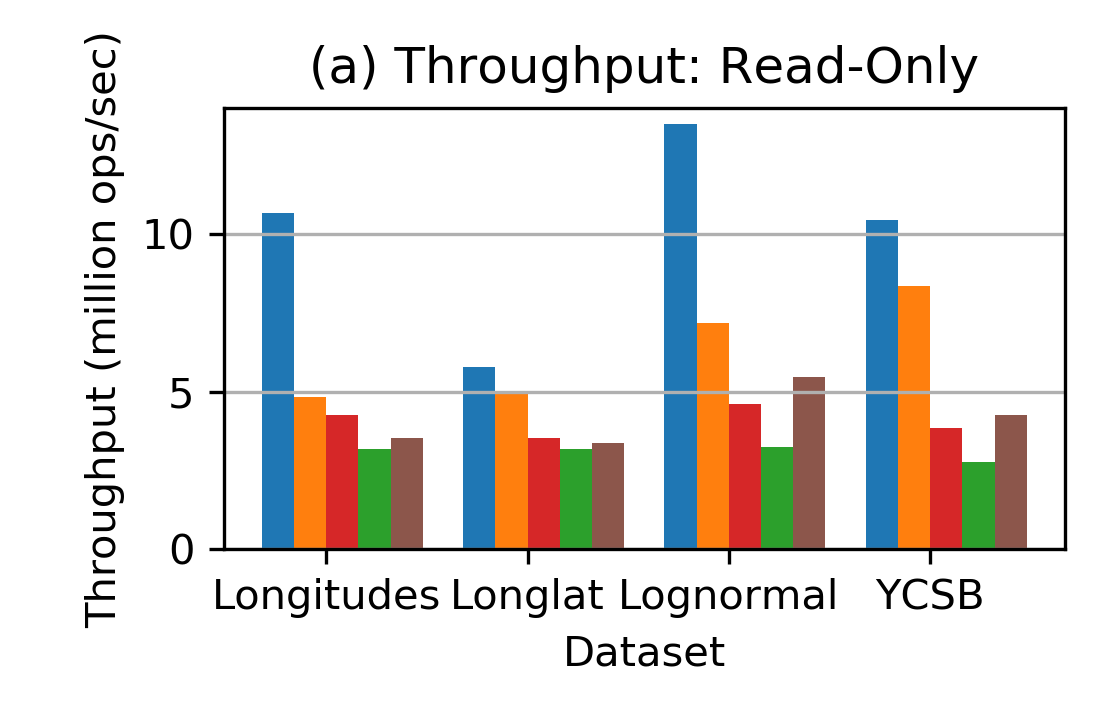}
        \label{fig:punchline_read}
        }
    ~
    \subfloat{
        \includegraphics[width=0.19\textwidth,trim={9 10 10 8},clip]{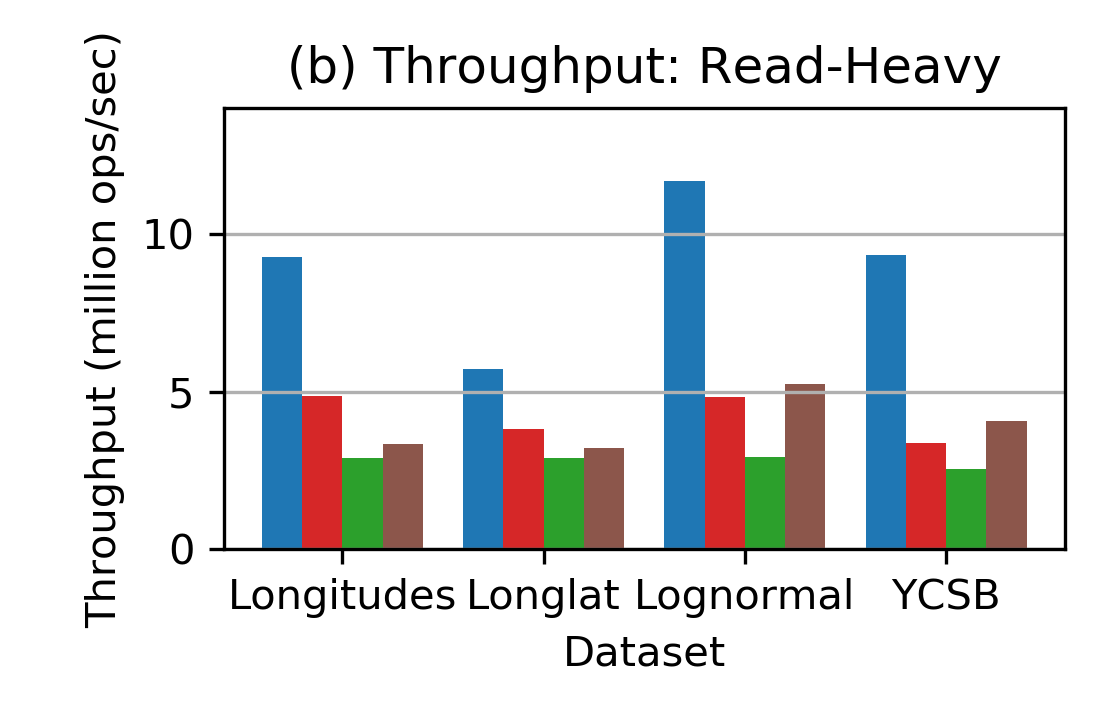}
        \label{fig:punchline_mixed}
        }
    ~
    \subfloat{
        \includegraphics[width=0.19\textwidth,trim={9 10 10 8},clip]{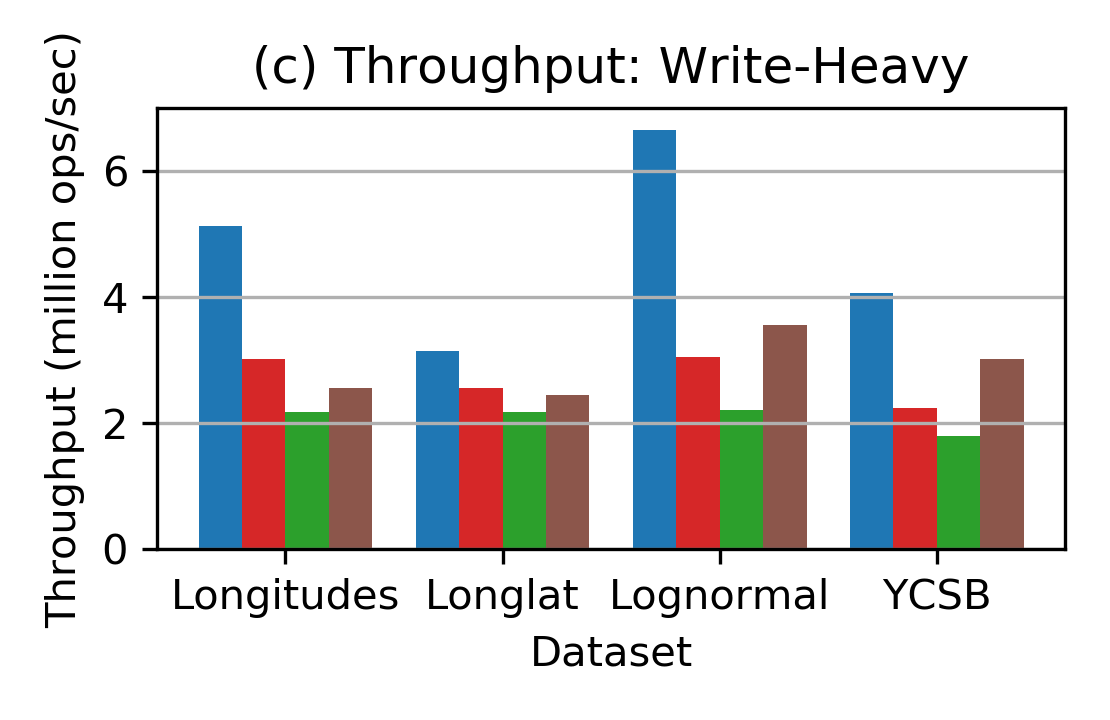}
        \label{fig:punchline_write}
        }
    ~
    \subfloat{
        \includegraphics[width=0.19\textwidth,trim={9 10 10 8},clip]{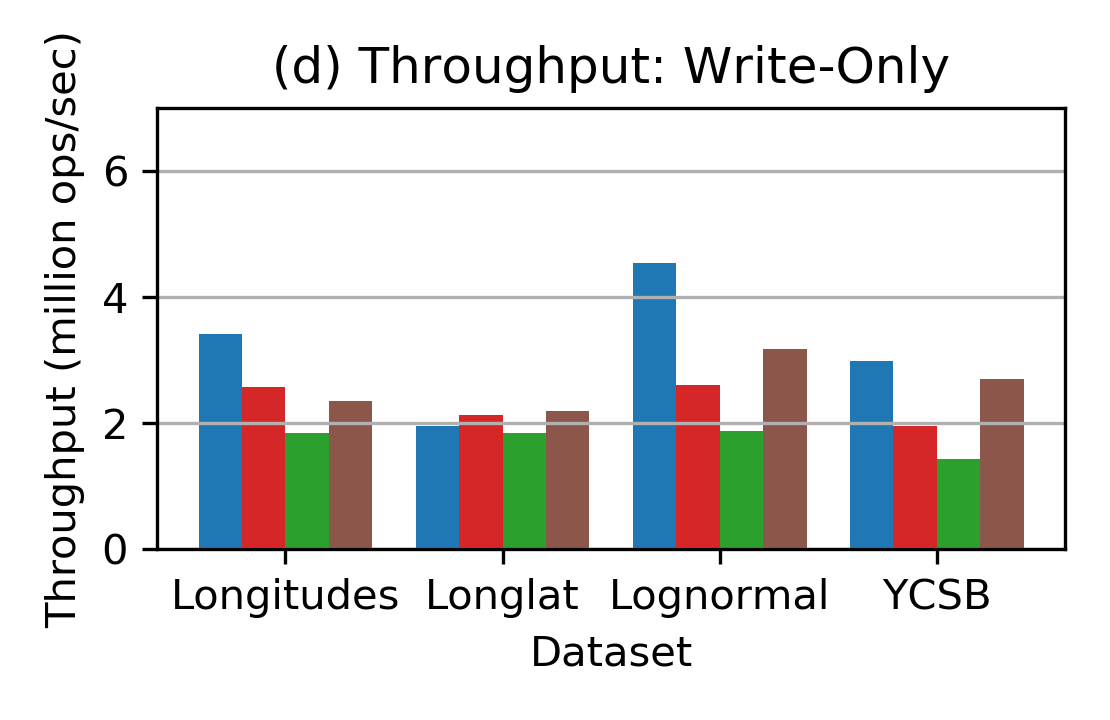}
        \label{fig:punchline_write_only}
        }
    ~
    \subfloat{
        \includegraphics[width=0.19\textwidth,trim={9 10 10 8},clip]{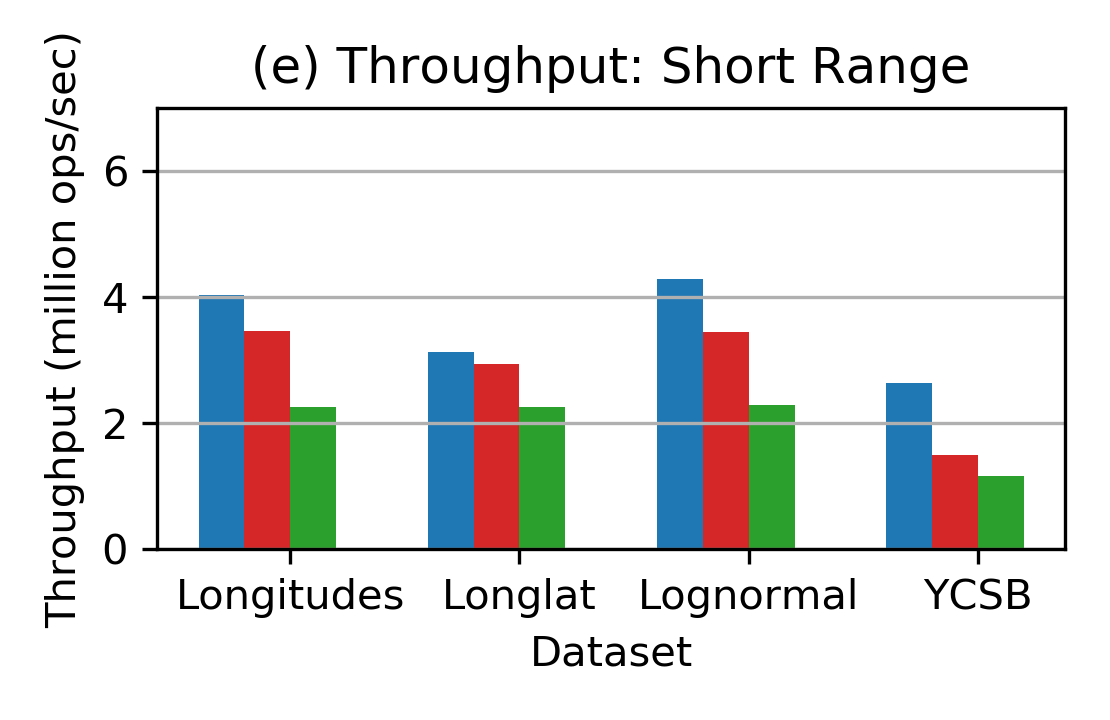}
        \label{fig:punchline_scan}
        }
    \\
    \vspace{-0.5em}
    \subfloat{
        \includegraphics[width=0.19\textwidth,trim={9 10 10 10},clip]{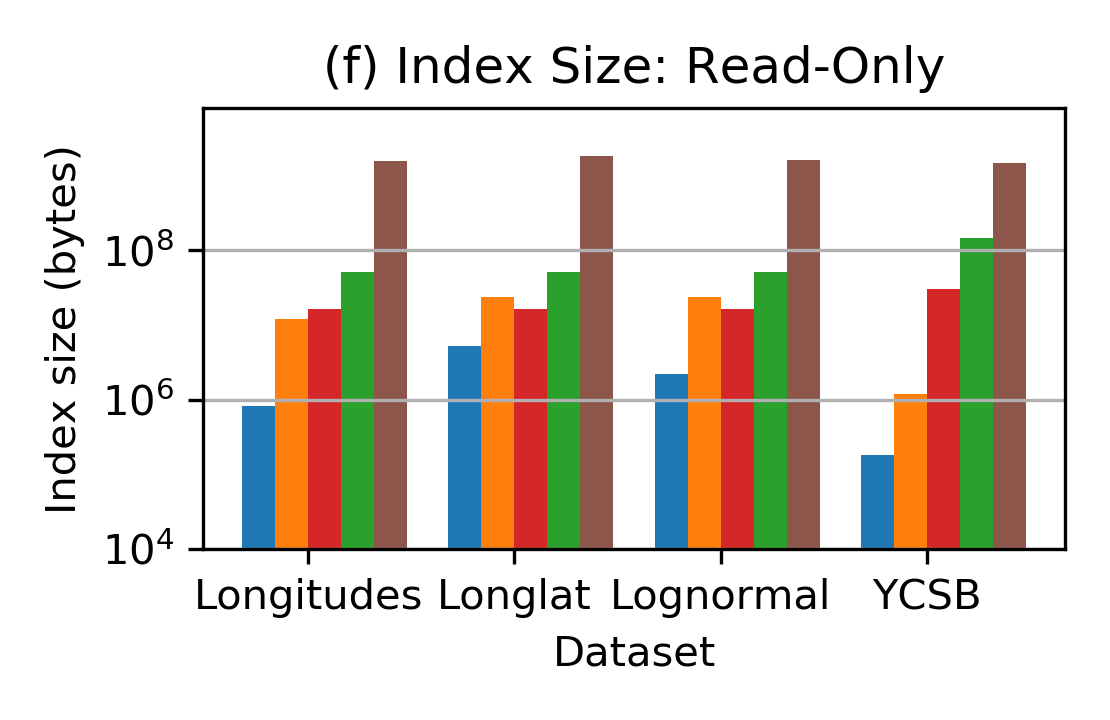}
        \label{fig:punchline_read_index_size}
        }
    ~
    \subfloat{
        \includegraphics[width=0.19\textwidth,trim={9 10 10 10},clip]{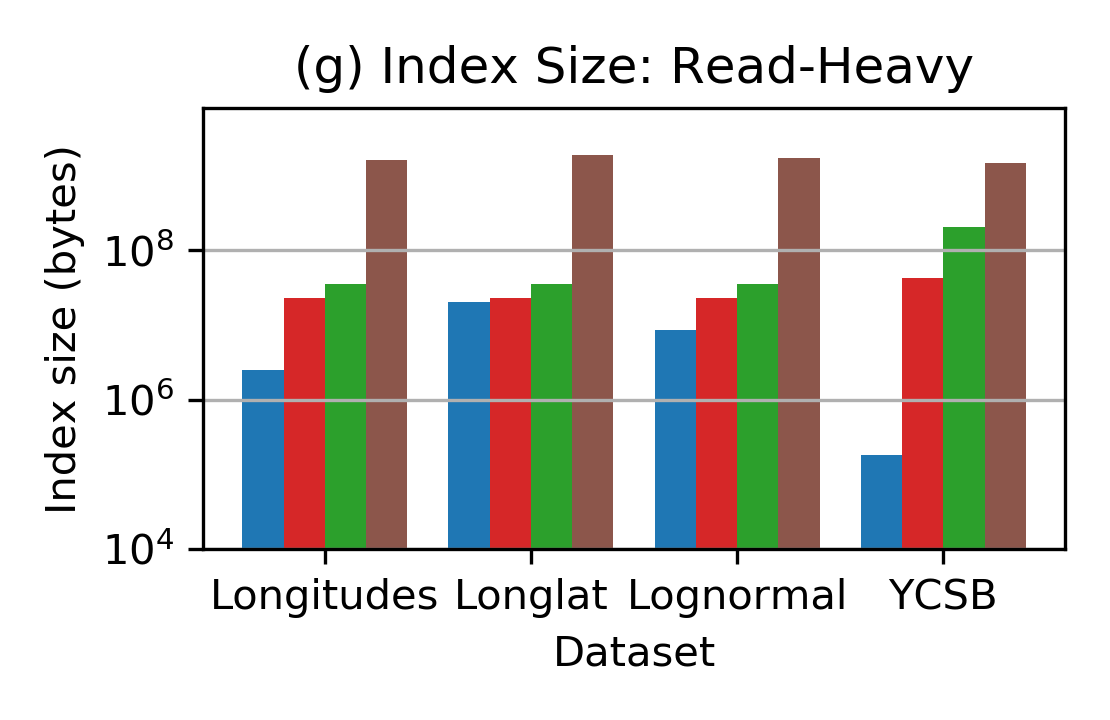}
        \label{fig:punchline_mixed_index_size}
        }
    ~
    \subfloat{
        \includegraphics[width=0.19\textwidth,trim={9 10 10 10},clip]{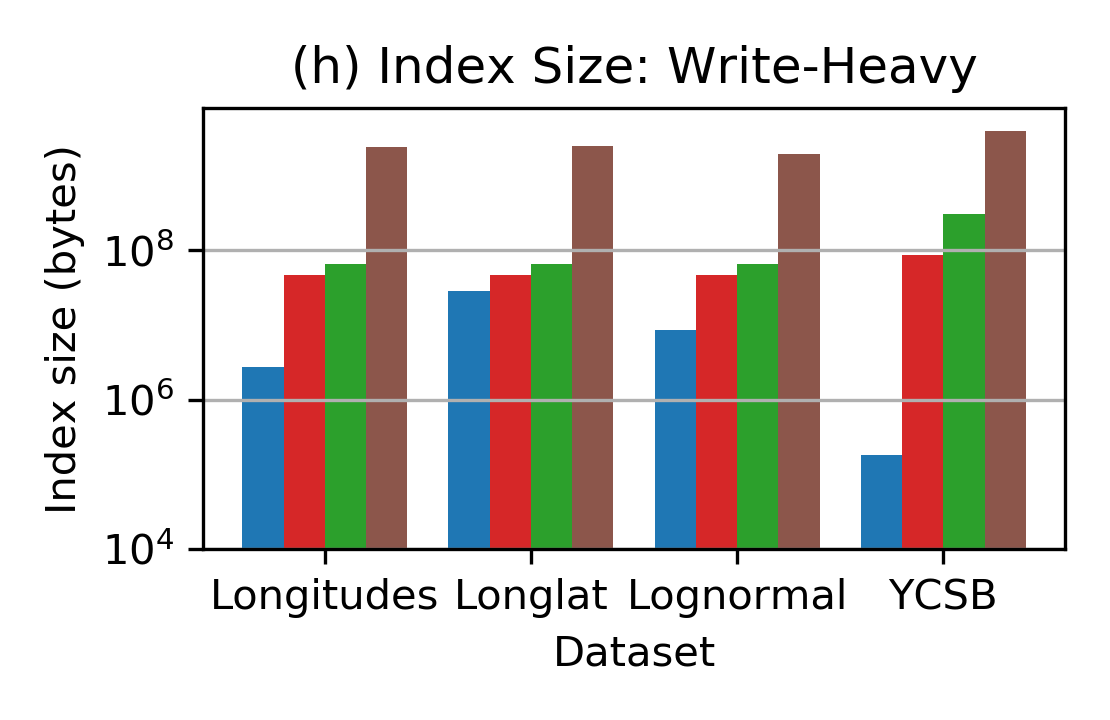}
        \label{fig:punchline_write_index_size}
        }
    ~
    \subfloat{
        \includegraphics[width=0.19\textwidth,trim={9 10 10 10},clip]{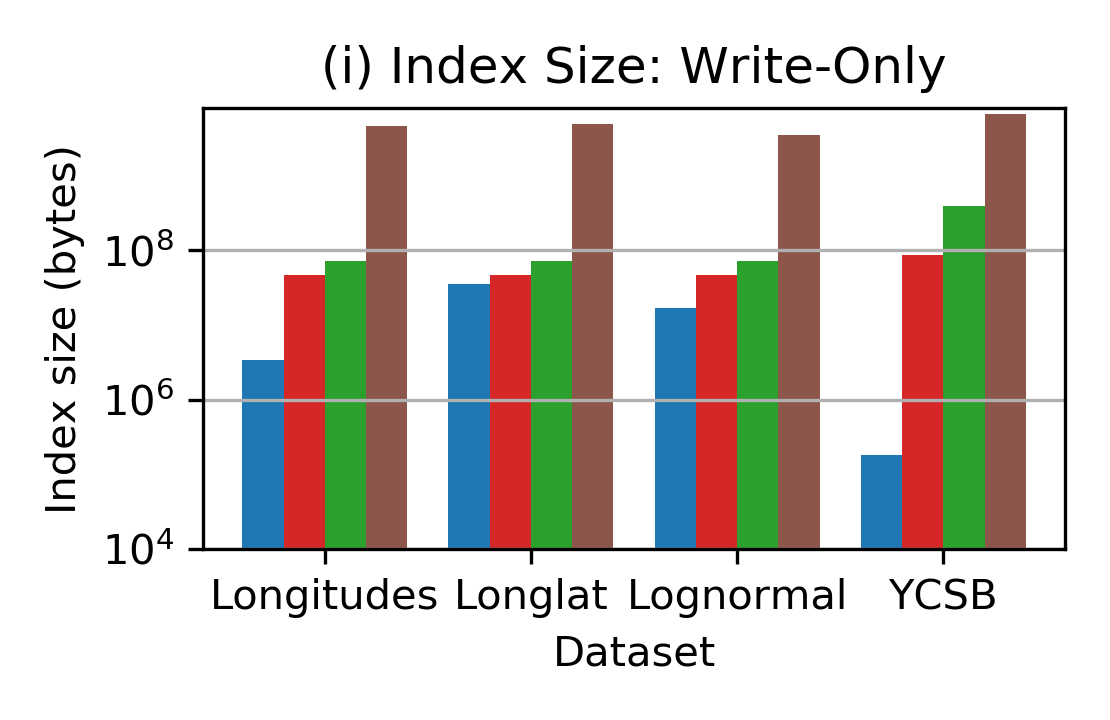}
        \label{fig:punchline_write_only_index_size}
        }
    ~
    \subfloat{
        \includegraphics[width=0.19\textwidth,trim={9 10 10 10},clip]{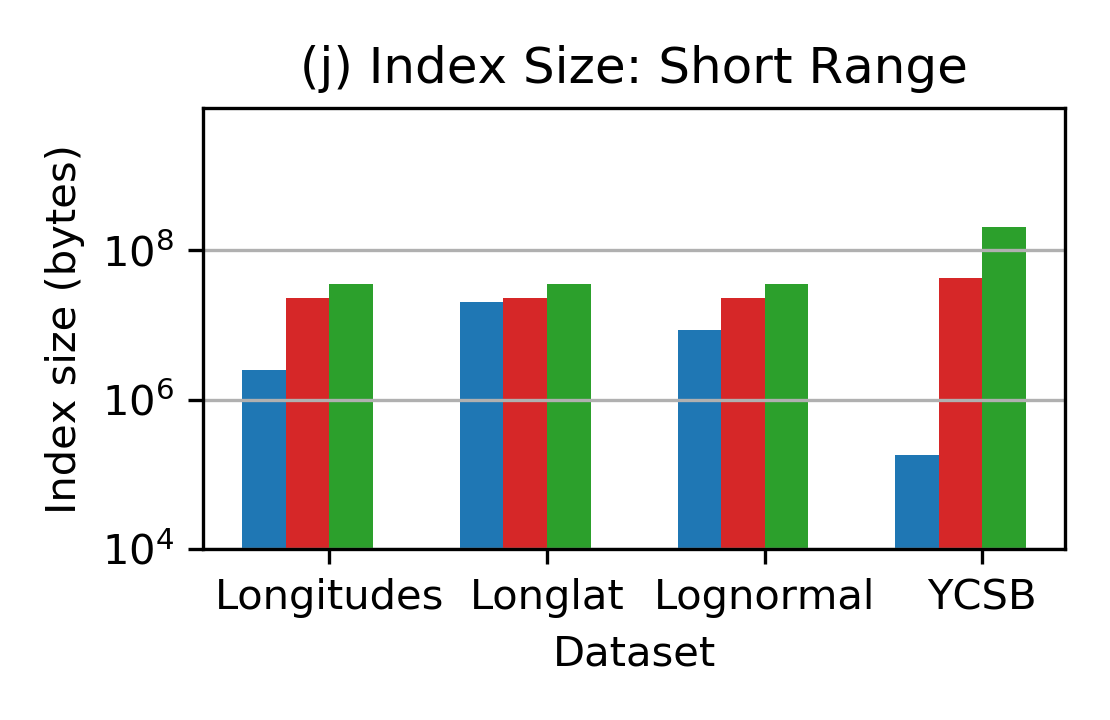}
        \label{fig:punchline_scan_index_size}
        }
    \caption{
        \atlex vs. Baselines: Throughput \& Index Size. Throughput includes model retraining time.
    }
    \label{fig:punchline}
    \vspace{-1em}
\end{figure*}

\subsubsection{Workloads}
Our primary metric for evaluating \atlex is average throughput.  We evaluate
throughput for five workloads:
(1) a read-only workload, (2) a read-heavy workload with 95\% reads and 5\% inserts, (3) a write-heavy workload with 50\% reads and 50\% inserts, (4) a short range query workload with 95\% reads and 5\% inserts, and (5) a write-only workload, to complete the read-write spectrum.
For the first three workloads, reads consist of a lookup of a single key.
For the short range workload, a read consists of a key lookup followed by a scan of the subsequent keys.
The number of keys to scan is selected randomly from a uniform distribution with a maximum scan length of 100.
For all workloads, keys to look up are selected randomly from the set of existing keys in the index according to a Zipfian distribution.
The first four workloads roughly correspond to Workloads C, B, A, and E from the YCSB benchmark~\cite{cooper2010benchmarking}, respectively.
For a given dataset, we
initialize an index with 100 million keys.  We then run the workload for
60 seconds, inserting the remaining keys.  We report the throughput of operations completed in
that time, where operations are either inserts or reads.  For the
read-write workloads, we interleave the operations: for the read-heavy workload and short range workload, we perform 19 reads/scans, then 1 insert, then repeat the cycle; for the write-heavy workload, we perform 1 read, then 1 insert, then repeat the cycle.
%For example, for the read-heavy workload, we
%perform 19 reads, then 1 insert, then repeat the cycle.

\begin{table}[]
\centering
\caption{\atlex Statistics after Bulk Load}
\vspace{-1em}
\small
\label{tab:node_size}
\begin{tabular}{@{}lllll@{}}
\toprule
       & \textbf{longitudes}     & \textbf{longlat}  & \textbf{lognormal} & \textbf{YCSB} \\
\midrule
\textbf{Avg depth} & 1.01 & 1.56 & 1.80 & 1 \\
\textbf{Max depth} & 2 & 4 & 3 & 1 \\
\textbf{Num inner nodes} & 55 & 1718 & 24 & 1 \\
\textbf{Num data nodes} & 4450 & 23257 & 757 & 1024 \\
\textbf{Min DN size} & 672B & 16B & 224B & 12.3MB \\
\textbf{Median DN size} & 161KB & 39.6KB & 2.99MB & 12.3MB \\
\textbf{Max DN size} & 5.78MB & 8.22MB & 14.1MB   & 12.3MB    \\
\bottomrule
\end{tabular}
\end{table}

\subsection{Overall Results}
\label{subsec:overall_results}
\subsubsection{Read-only Workloads}
\label{subsubsec:read-only_workloads}
% \todo{We need to add concrete numbers here. Saying something like
%   ``\atlex is XXx, and YYx faster than learned index, and \bptree,
%   respectively for ZZZ workload'' is always more powerful. Please
%   apply this comment throughput this section}

% For the read-only
% workload, all lookups will return a value in the index.

\begin{sloppypar}
For read-only workloads, \cref{fig:punchline_read,fig:punchline_read_index_size} show that \atlex achieves up to 4.1$\times$, 2.2$\times$, 2.9$\times$, 3.0$\times$ higher throughput and 800$\times$, 15$\times$, 160$\times$, 8000$\times$ smaller index size than the \bptree, \kraskali, \modelbtree, and \art, respectively.
\end{sloppypar}

On the longlat and YCSB datasets, \atlex performance is similar to \kraskali.
The longlat dataset is highly non-uniform, so \atlex is unable to achieve high performance, even with adaptive RMI.
The YCSB dataset is nearly uniform, so the optimal allocation of models is uniform; \atlex adaptively finds this optimal allocation, and \kraskali allocates this way by nature, so the resulting RMI structures are similar.
On the other two datasets, \atlex has more performance advantage over \kraskali, which we explain in \cref{subsec:drilldown}.

In general, \modelbtree outperforms \bptree while also having smaller index size, because the tuned page size of \modelbtree is always larger than those of \bptree. The benefit of models in \modelbtree is greatest when the key distribution within each node is more uniform, which is why \modelbtree has least benefit on non-uniform datasets like longlat.

The index size of \atlex is dependent on how well \atlex can model the data distribution.
On the YCSB dataset, \atlex does not require a large RMI to accurately model the distribution, so \atlex achieves small index size.
However, on datasets that are more challenging to model such as longlat, \atlex has a larger RMI with more nodes.
\atlex has smaller index size than the \kraskali, even when throughput is similar, for two reasons.
First, \atlex uses model-based inserts to obtain better predictive accuracy for each model, which we show in \cref{subsec:drilldown}, and therefore achieves high throughput while using relatively fewer models.
Second, \atlex adaptively allocates data nodes to different parts of the key space and does not use any more models than necessary (\cref{fig:internal_nodes}), whereas \kraskali fixes the number of models and ends up with many redundant models.
The index size of \art is higher than all other indexes. \cite{leis2013adaptive} claims that \art uses between 8 and 52 bytes to store each key, which is in agreement with the observed index sizes.

\cref{tab:node_size} shows \atlex statistics after bulk loading, including data node (DN) sizes. The root has depth 0. Average depth is averaged over keys. The max depth of the tuned \bptree is 4 on the YCSB dataset and 5 on the other datasets. Datasets that are easier to model result in fewer nodes. For uniform datasets like YCSB, the data node sizes are also uniform.

% For example, on the YCSB dataset, \kraskali achieves its best throughput with 50000 models, whereas \atlex achieves its best throughput with 25 models.
% This difference, combined with the fact that \kraskali uses additional space per model to store error bounds, accounts for \atlex having 3000$\times$ smaller index size.
% However, on datasets that are more challenging to model such as longlat, \atlex is unable to achieve high throughput with few models.
% Therefore,
% \atlex tries to use more models, but even then, does not achieve
% high throughput.
% This demonstrates a recurring trend: when \atlex achieves higher throughput, it does so with smaller index size.

% \begin{figure}
%     \includegraphics[width=\columnwidth]{figures/punchline_read.png}
%     \vspace{-2em}
%     \caption{Throughput of a read-only workload on different datasets.}
%     \label{fig:punchline_read}
% \end{figure}
% \begin{figure}
%     \includegraphics[width=\columnwidth]{figures/read_index_size.png}
%     \vspace{-2em}
%     \caption{Index size of indexes in \cref{fig:punchline_read}.}
%     \label{fig:punchline_read_index_size}
% \end{figure}

\subsubsection{Read-Write Workloads}
For read-write workloads, \cref{fig:punchline_mixed,fig:punchline_write,fig:punchline_write_only,fig:punchline_mixed_index_size,fig:punchline_write_index_size,fig:punchline_write_only_index_size}
show that \atlex achieves up to 4.0$\times$, 2.7$\times$, 2.7$\times$ higher throughput and 2000$\times$, 475$\times$, 36000$\times$ smaller index size than the \bptree, \modelbtree, and \art, respectively. The
\kraskali has insert time orders of magnitude slower than
\atlex and \bptree, so we do not include it in these benchmarks.

The relative performance advantage of \atlex over baselines decreases as the workload skews more towards writes, because all indexes must pay the cost of copying when splitting/expanding nodes. Copying has an especially big impact for YCSB, for which payloads are 80 bytes.
\art achieves comparable throughput to \atlex on the write-only workload for YCSB because \art does not keep payloads clustered, so it avoids the high cost of copying 80-byte payloads.
Note that \atlex could similarly avoid copying large payloads by storing unclustered payloads separately and keeping a pointer with every key; however, this would impact scan performance.
On datasets that are challenging to model such as longlat, \atlex only achieves comparable write-only throughput to \modelbtree and \art, but is still faster than \bptree.

% \begin{figure}
%     \includegraphics[width=\columnwidth]{figures/punchline_mixed.png}
%     \vspace{-2em}
%     \caption{Throughput of a mixed workload with 50\% reads and 50\% writes on different datasets.}
%     \label{fig:punchline_mixed}
% \end{figure}
% \begin{figure}
%     \includegraphics[width=\columnwidth]{figures/mixed_index_size.png}
%     \vspace{-2em}
%     \caption{Index size of indexes in \cref{fig:punchline_mixed}.}
%     \label{fig:punchline_mixed_index_size}
% \end{figure}
% \begin{figure}
%     \includegraphics[width=\columnwidth]{figures/punchline_write.png}
%     \vspace{-2em}
%     \caption{Throughput of a write-heavy workload with 90\% writes and 10\% reads on different datasets.}
%     \label{fig:punchline_write}
% \end{figure}

% Similar to the read-only workloads, higher throughput is generally achieved in conjunction with smaller index size.
% For example, for the lognormal dataset, \atlex
% has high throughput and small index size, but for the longlat dataset,
% \atlex has lower throughput and higher index size, comparable in size
% to \bptree.  

\subsubsection{Range Query Workloads}
\label{subsubsec:scan}
\cref{fig:punchline_scan,fig:punchline_scan_index_size} show that \atlex
maintains its advantage over \bptree on the short range workload, \secondrev{achieving up to 2.27$\times$, 1.77$\times$ higher throughput and 1000$\times$, 230$\times$ smaller index size than \bptree and \modelbtree, respectively}. However, the relative throughput benefit decreases, compared to \cref{fig:punchline_mixed}. This is because as scan time begins to dominate overall query time, the speedups that \atlex achieves on lookups become less apparent.
The \art implementation from~\cite{artimpl} does not support range queries; we suspect range queries on \art would be slower than for the other indexes because \art does not cluster payloads, leading to poor scan locality.
\secondrev{\iftoggle{sigmod}{Appendix C.2}{\cref{subsec:mixed_workload_eval}}\sigmod{ in~\cite{techreport}} shows that \atlex continues to outperform other indexes on a workload that mixes inserts, point lookups, and short range queries.}
% For this workload, we use a C++ \art implementation~\cite{artimpl2} because~\cite{artimpl} does not support range scans.

\begin{figure}
	\centering
	\includegraphics[width=\columnwidth]
	{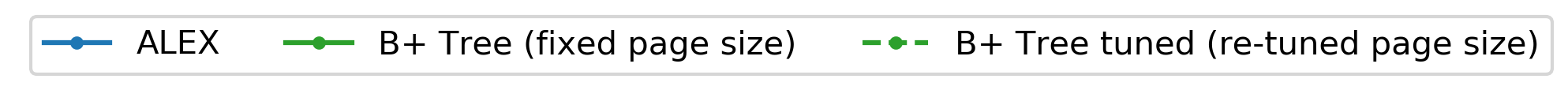}
	\vspace{-2em}
	\\
	\subfloat{
		\includegraphics[width=0.48\columnwidth,trim={8 8 7 8},clip]{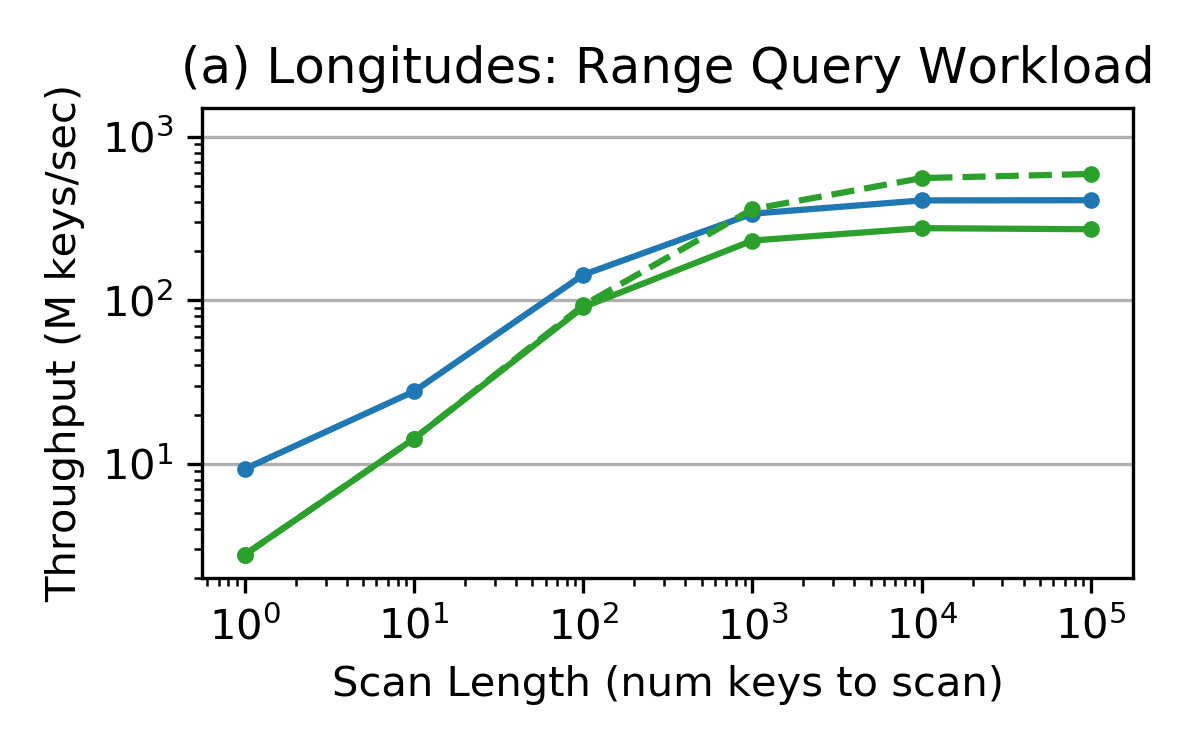}
		\label{fig:selectivity_throughput}
	}
	~
	\subfloat{
		\includegraphics[width=0.48\columnwidth,trim={8 8 7 8},clip]{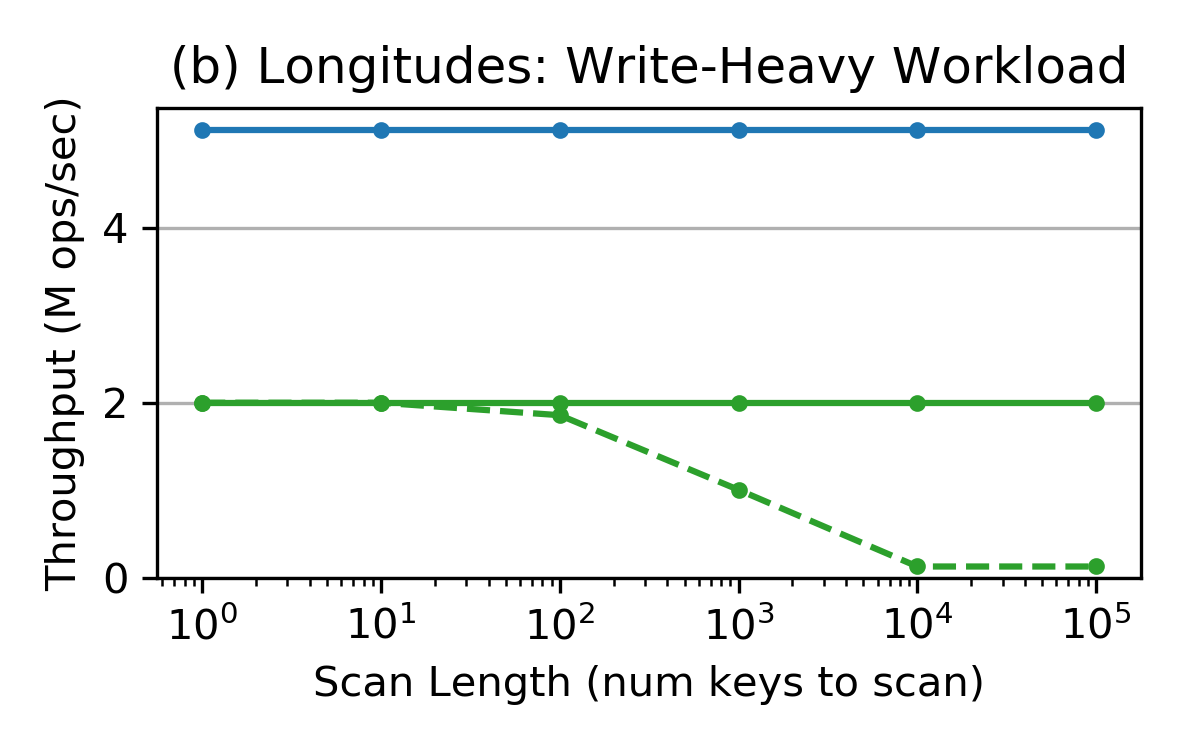}
		\label{fig:selectivity_counterpoint}
	}
	\caption{
		\secondrev{(a) When scan length exceeds 1000 keys, \atlex is slower on range queries than a \bptree whose page size is re-tuned for different scan lengths. (b) However, throughput of the re-tuned \bptree suffers for other operations, such as point lookups and inserts in the write-heavy workload.}
	}
	\label{fig:selectivity}
\end{figure}

\secondrev{To show how performance varies with range query selectivity, we compare \atlex against two \bptree configurations with increasingly larger range scan length over the longitudes dataset (\cref{fig:selectivity_throughput}).
In the first \bptree configuration, we use the optimal \bptree page size on the write-heavy workload (\cref{fig:punchline_write}), which is 1KB (solid green line).
In the second \bptree configuration, we tune the \bptree page size for each different scan length (dashed green line).}
	
\secondrev{Unsurprisingly, \cref{fig:selectivity_throughput} shows as scan length increases, the throughput in terms of keys scanned per second increases for all indexes due to better locality and a smaller fraction of time spent on the initial point lookup.
Furthermore, \atlex outperforms the 1KB-page \bptree for all scan lengths due to \atlex's larger nodes; median \atlex data node size is 161KB on the longitudes dataset (\cref{tab:node_size}), which benefits scan locality---scanning larger contiguous chunks of memory leads to better prefetching and fewer pointer chases.
This makes up for the Gapped Array's overhead.}
	
\secondrev{However, if we re-tune the \bptree page size for each scan length (dashed green line), the \bptree outperforms \atlex when scan length exceeds 1000 keys because past this point, the overhead of Gapped Array outpaces \atlex's scan locality advantage from having larger node sizes.
However, this comes at the cost of performance on other operations:
\cref{fig:selectivity_counterpoint} shows that if we run the re-tuned \bptree on the write-heavy workload, which includes both point lookups and inserts, its performance would begin to decline when scan length exceeds 100 keys.
In particular, larger \bptree pages lead to a higher number of search iterations for lookups and shifts for inserts; \atlex avoids both of these problems for large data nodes by using Gapped Arrays with model-based inserts.
We show in \iftoggle{sigmod}{Appendix C.1}{\cref{subsec:range_scan_throughput}}\sigmod{ in~\cite{techreport}} that this behavior also occurs on the other three datasets.}

\subsubsection{Bulk Loading}
\label{subsubsec:main_bulk_loading_eval}
\begin{figure}
	\centering
	\includegraphics[width=0.8\columnwidth]
	{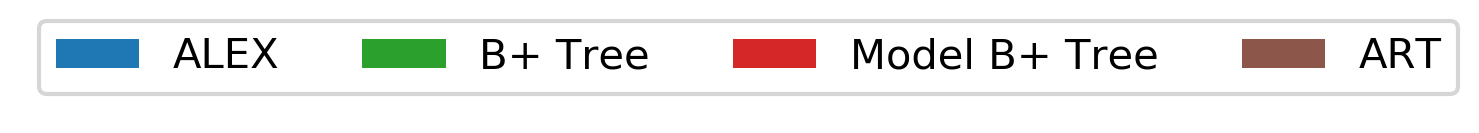}
	\vspace{-1em}
	\\
	\subfloat{
		\includegraphics[width=0.51\columnwidth,trim={8 8 7 8},clip]{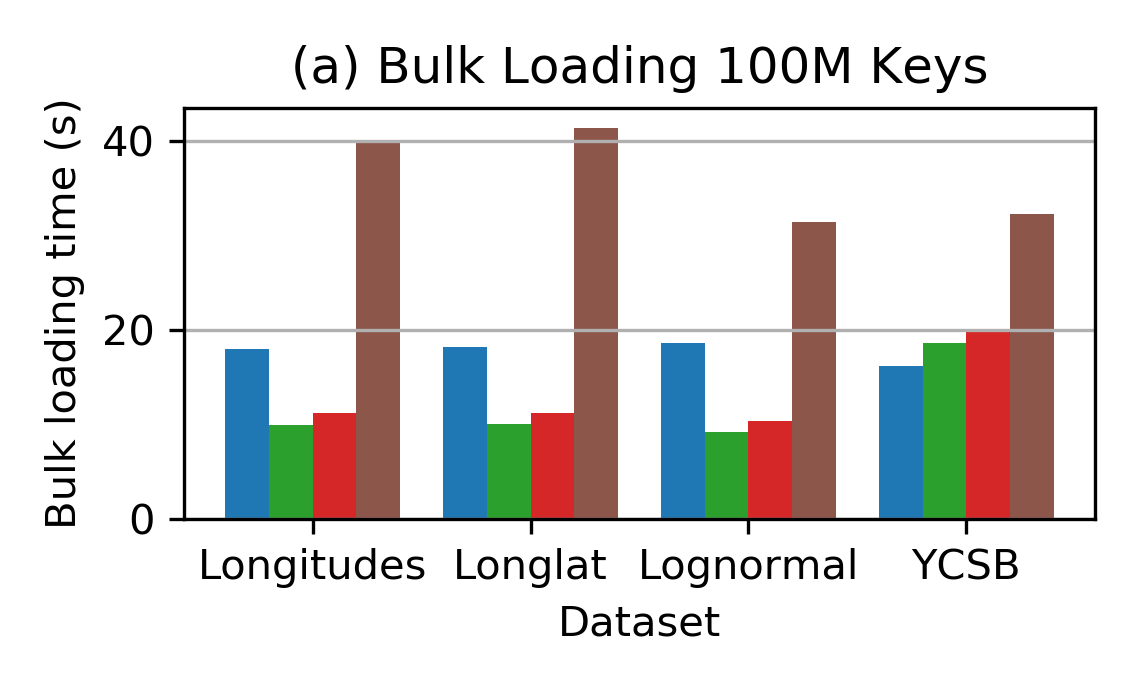}
		\label{fig:simple_bulk_loading_time}
	}
	~
	\subfloat{
		\includegraphics[width=0.47\columnwidth,trim={8 8 7 8},clip]{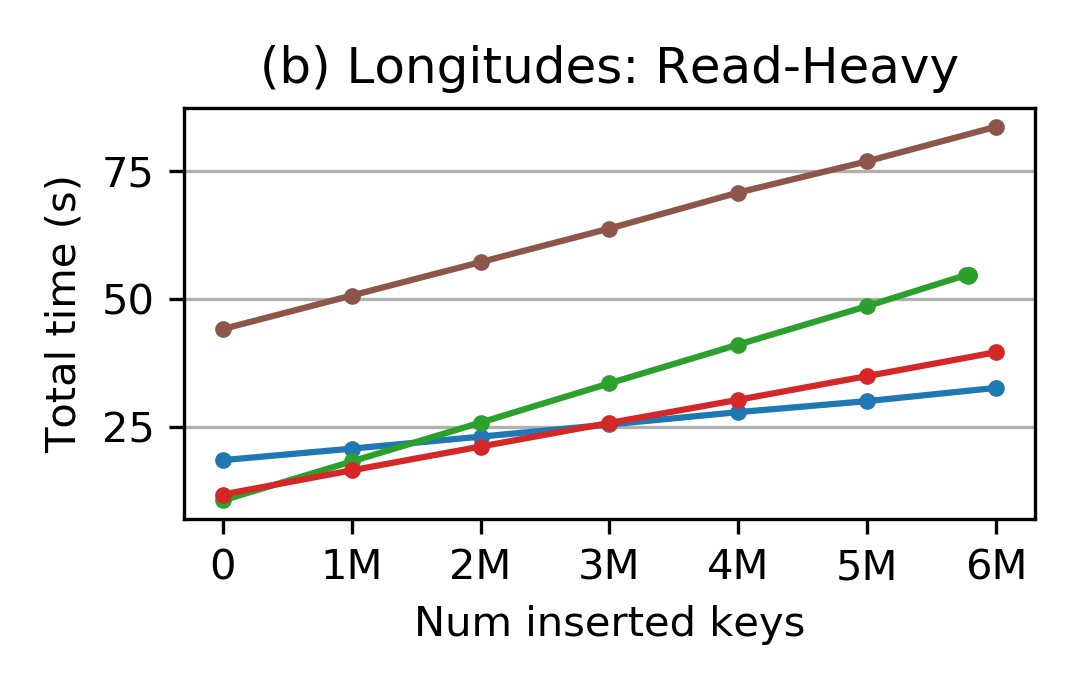}
		\label{fig:crossover}
	}
	\vspace{-0.5em}
	\caption{
		\thirdrev{\atlex takes 50\% more than time than \bptree to bulk load on average, but quickly makes up for this by having higher throughput.}
	}
	\label{fig:bulk_loading}
\end{figure}
\thirdrev{We compare the time to initialize each index with bulk loading, which includes the time to sort keys. \cref{fig:simple_bulk_loading_time} shows that on average, \atlex only takes 50\% more time to bulk load than \bptree, and in the worst case is only 2$\times$ slower than \bptree.
On the YCSB dataset, \bptree and \modelbtree take longer to bulk load due to the larger payload size, but bulk loading \atlex remains efficient due to its simple structure (\cref{tab:node_size}).
\modelbtree is slightly slower to bulk load than \bptree due to the overhead of training models for each node.
\art is slower to bulk load than \bptree, \modelbtree, and \atlex.}

\thirdrev{
\atlex can quickly make up for its slower bulk loading time than \bptree by having higher throughput performance.
\cref{fig:crossover} shows that when running the read-heavy workload on the longitudes dataset, \atlex's total time usage (bulk loading plus workload) drops below all other indexes after only 3 million inserts.
We provide a more detailed bulk loading evaluation in \iftoggle{sigmod}{Appendix A}{\cref{sec:bulk_loading_eval}}\sigmod{ in~\cite{techreport}}.}

\subsubsection{Scalability}
\atlex performance scales well to larger datasets.  We again run the read-heavy workload on the longitudes
dataset, but instead of initializing the index with
100 million keys, we vary the number of initialization keys.
\cref{fig:scalability_mixed} shows that as the number of indexed
keys increases, \atlex maintains higher throughput than \bptree and \modelbtree.  In
fact, as dataset size increases, \atlex throughput decreases at a
surprisingly slow rate.  This occurs because \atlex adapts its RMI structure in response to the incoming data.
% \atlex also maintains a
% constant proportion of gaps to keys in the data nodes, so even as the
% absolute number of keys increases, the time taken to insert a key does
% not increase by much.

% \begin{figure}
%     \includegraphics[width=0.8\columnwidth]
%     {figures/legend_distshift.png}
%     \begin{minipage}{0.38\linewidth}
%         \centering
%         % \includegraphics[width=\columnwidth]
%         %         {figures/legend_scalability.png}
%         % \\
%         \includegraphics[width=\columnwidth,height=2.1cm,trim={0 7 7 5},clip]{figures/scalability_mixed.png}
%         % \vspace{-2em}
%         \caption{\atlex maintains high throughput when scaling to larger datasets.}
%         \label{fig:scalability_mixed}        
%     \end{minipage}
%     \hfill
%     \begin{minipage}{0.6\linewidth}
%         \centering
%         % \\[-0.5ex] % negative space
%         \subfloat{
%             \includegraphics[width=0.48\columnwidth,trim={10 10 10 10},clip]{figures/distshift.png}
%             \label{fig:distshift}
%             }
%         ~
%         \subfloat{
%             \includegraphics[width=0.48\columnwidth,trim={10 10 10 10},clip]{figures/sorted.png}
%             \label{fig:sorted}
%             }
%         \caption{
%             \atlex maintains competitive throughput on mild distribution shift but has poor performance on sequential inserts.
%         }       
%     \end{minipage}
% \end{figure}

\begin{figure}
    \centering
    \includegraphics[width=0.6\columnwidth]
    {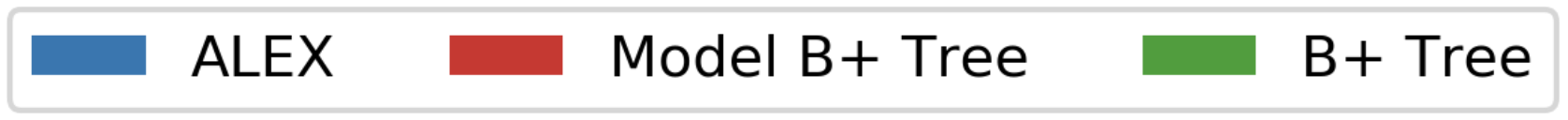}
    \vspace{-1em}
    \\
    \subfloat{
        \includegraphics[width=0.47\columnwidth,trim={0 10 10 5},clip]{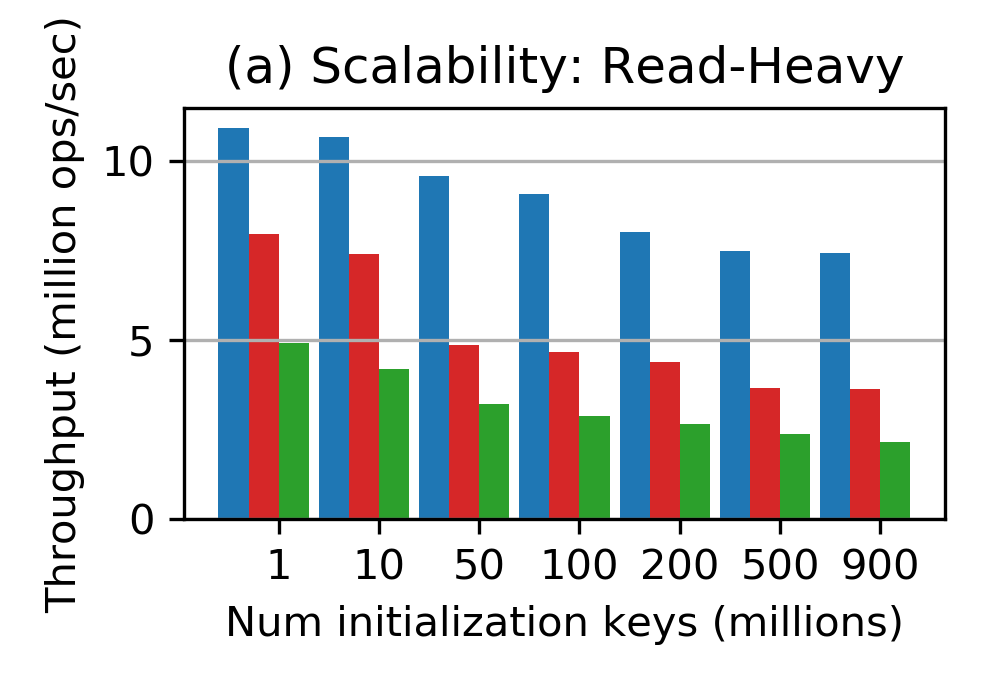}
        \label{fig:scalability_mixed}
        }
    \subfloat{
            \includegraphics[width=0.26\columnwidth,trim={10 10 9 10},clip]{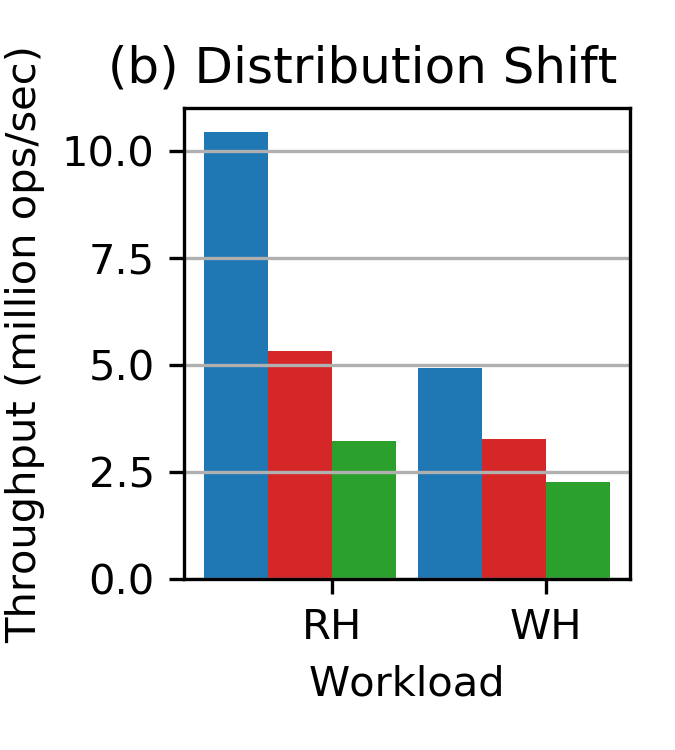}
            \label{fig:distshift}
            }
    \subfloat{
            \includegraphics[width=0.24\columnwidth,trim={10 10 10 10},clip]{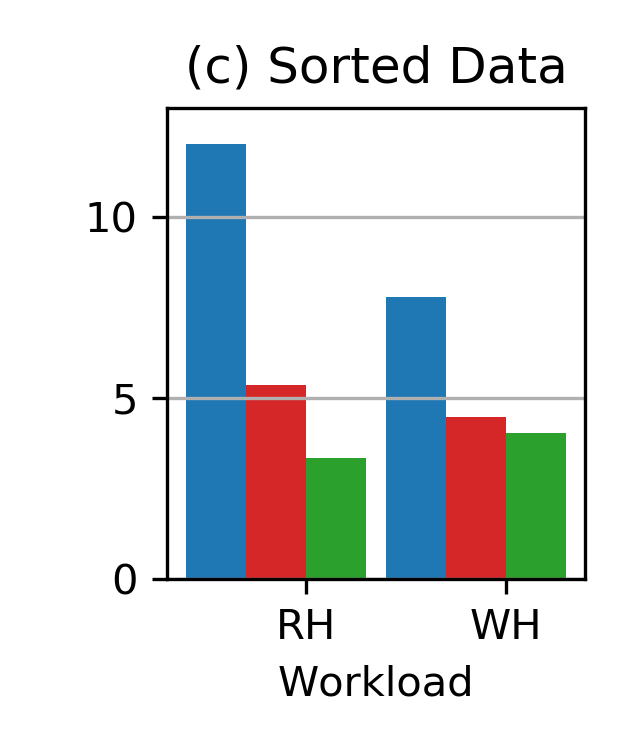}
            \label{fig:sorted}
            }
    \vspace{-0.5em}
    \caption{
        \atlex maintains high throughput when scaling to large datasets and under data distribution shifts. (RH = Read-Heavy, WH = Write-Heavy)
    }
\end{figure}

\subsubsection{Dataset Distribution Shift}
\label{ssub:dataset_distribution_shift}
\atlex is robust to dataset distribution shift.
We initialize the index with the 50 million smallest keys and run read-write workloads by inserting the remaining keys in random order. This simulates distribution
shift because the keys we initialize with come from a completely disjoint domain than the keys we subsequently insert with.
\cref{fig:distshift} shows that \atlex maintains up to 3.2$\times$ higher throughput than \bptree in this scenario.
\atlex is also robust to adversarial patterns such as
sequential inserts in sorted order, in which new keys are always larger than the maximum key currently indexed.  \cref{fig:sorted} shows that when we initialize with the 50 million smallest keys and insert the remaining keys in ascending sorted order, \atlex has
up to $3.6\times$ higher throughput than \bptree.
\firstrev{\iftoggle{sigmod}{Appendix B}{\cref{sec:extreme_distshift_eval}}\sigmod{ in~\cite{techreport}} further shows that \atlex is robust to radically changing key distributions.}

\begin{figure*}[h!]
	\begin{minipage}{0.29\linewidth} 
		\centering
		\includegraphics[width=\columnwidth]{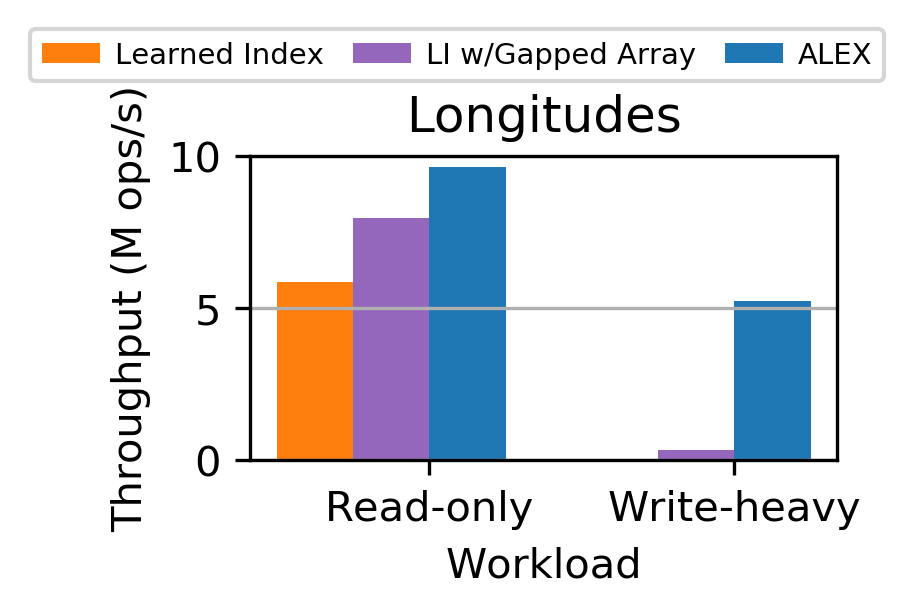}
		\vspace{-2em}
		\caption{
			Impact of Gapped Array and adaptive RMI.
		}
	\label{fig:ablation}
	\end{minipage}
    \begin{minipage}{0.66\linewidth} 
    \subfloat{
        \includegraphics[width=0.32\columnwidth]{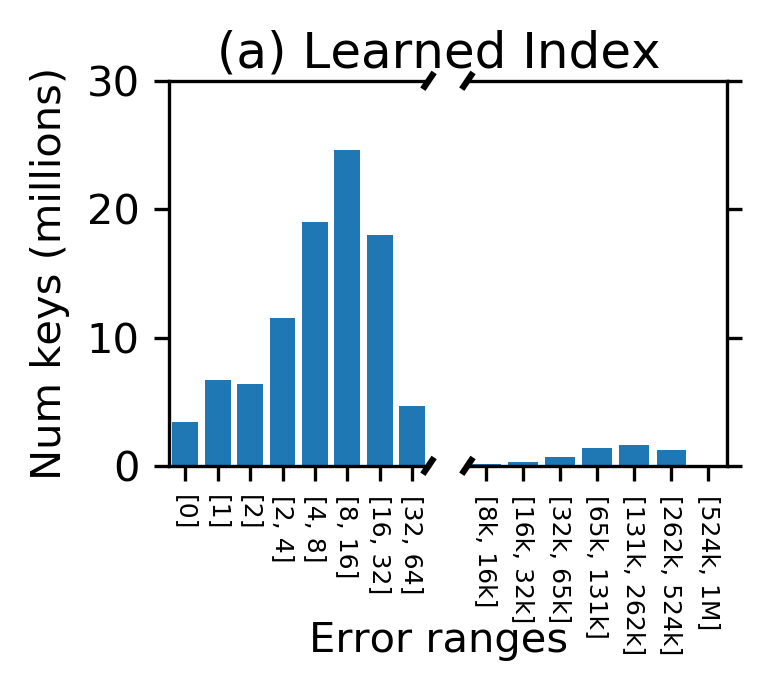}
        \label{fig:simple_errors}
        }
    ~
    \subfloat{
        \includegraphics[width=0.32\columnwidth]{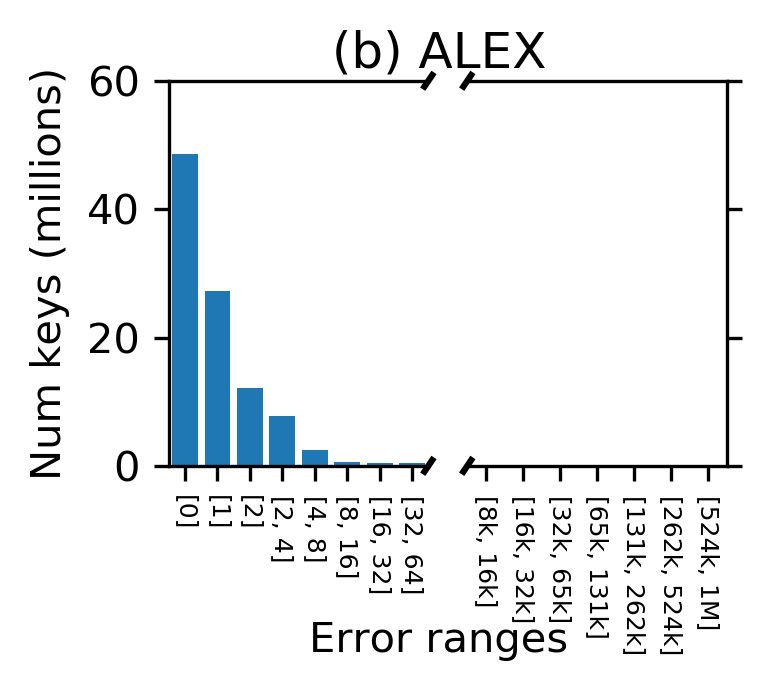}
        \label{fig:alex_errors}
        }
    ~
    \subfloat{
        \includegraphics[width=0.32\columnwidth]{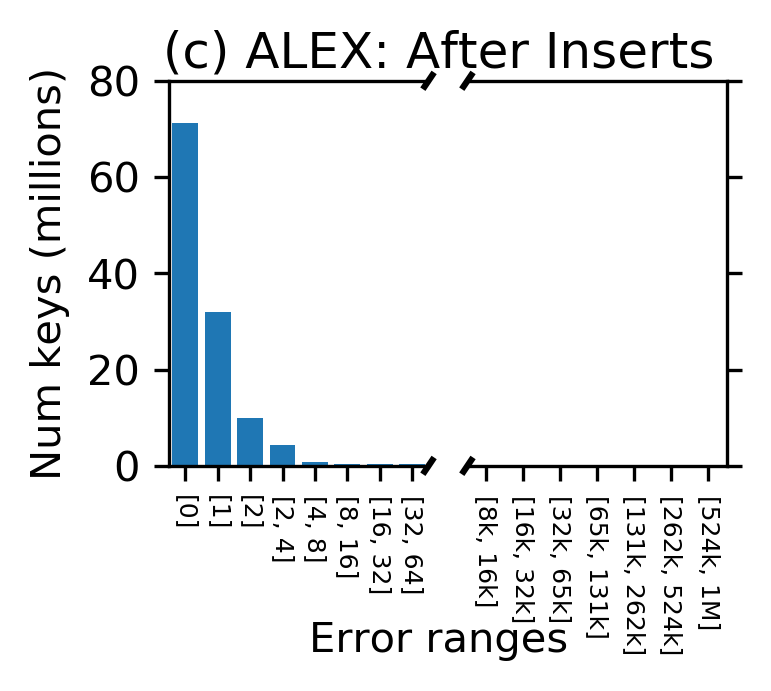}
        \label{fig:alex_errors_updates}
        }
    \caption{
        \atlex achieves smaller prediction error than the \kraskali.
    }
    \end{minipage}
    \vspace{-1em}
\end{figure*}

\subsection{Drilldown into \atlex Design Trade-offs}
\label{subsec:drilldown}
In this section, we delve deeper into how node layout and adaptive RMI help \atlex achieve its design goals.

Part of \atlex's advantage over \kraskali comes from using model-based insertion with Gapped Arrays in the data nodes, but most of \atlex's advantage for dynamic workloads comes from the adaptive RMI.
To demonstrate the effects of each contribution, \cref{fig:ablation} shows that taking a 2-layer \kraskali and replacing the single dense array of values with a Gapped Array per leaf (LI w/Gapped Array) already achieves significant speedup over \kraskali for the read-only workload.
However, a \kraskali with Gapped Arrays achieves poor performance on read-write workloads due to the presence of fully-packed regions which require shifting many keys for each insert.
\atlex's ability to adapt the RMI structure to the data is necessary for good insert performance.

During lookups, the majority of the time is spent doing local search
around the predicted position.
%Therefore, lookup time is mainly determined by the accuracy of the RMI models.
Smaller prediction errors directly contribute to decreased lookup time.
To analyze the prediction errors of the \kraskali and \atlex, we initialize an index with 100 million keys from the longitudes dataset, use the index to predict the position of each of the 100 million keys, and track the distance between the predicted position and the actual position. \cref{fig:simple_errors} shows that the \kraskali has
prediction error with mode around 8-32 positions, with a long tail
to the right.
On the other hand, \atlex achieves much lower prediction error by
using model-based inserts.  \cref{fig:alex_errors} shows that
after initializing, \atlex often has no prediction error, the
errors that do occur are often small, and the long tail of errors has disappeared. \cref{fig:alex_errors_updates} shows that even after 20 million inserts, \atlex
maintains low prediction errors.

\begin{table}[]
\centering
\caption{\thirdrev{Data Node Actions When Full (Write-Heavy)}}
\vspace{-1em}
\small
\label{tab:node_splits}
\begin{tabular}{@{}lllll@{}}
\toprule
       & \textbf{longitudes}     & \textbf{longlat}  & \textbf{lognormal} & \textbf{YCSB} \\
\midrule
\textbf{Expand + scale} & 26157 & 113801 & 2383 & 1022 \\
\textbf{Expand + retrain} & 219 & 2520 & 2 & 1026 \\
\textbf{Split sideways} & 79 & 2153 & 7 & 0 \\
\textbf{Split downwards} & 0 & 230 & 0 & 0 \\
\midrule
\textbf{Total times full} & 26455 & 118704 & 2392 & 2048 \\
\bottomrule
\end{tabular}
\end{table}

\thirdrev{Once a data node becomes full, one of four actions happens: if there is no cost deviation, then (1) the node is expanded and the model is scaled. Otherwise, the node is either (2) expanded and its model retrained, (3) split sideways, or (4) split downwards.
\cref{tab:node_splits} shows that in the vast majority of cases, the data node is simply expanded and the model scaled, which implies that models usually remain accurate even after inserts, assuming no radical distribution shift.
The number of occurrences of a data node becoming full is correlated with the number of data nodes (\cref{tab:node_size}).
On YCSB, expansion with model retraining is more common because the data nodes are large, so cost deviation often results simply from randomness.}

% \begin{figure}
%     \includegraphics[width=\columnwidth]{figures/error_simple_array.png}
%     \vspace{-2em}
%     \caption{The \kraskali has non-trivial model prediction error.}
%     \label{fig:simple_errors}
% \end{figure}

% \begin{figure}
%     \includegraphics[width=\columnwidth]{figures/error_alex.png}
%     \includegraphics[width=\columnwidth]{figures/error_alex_updates.png}
%     \vspace{-2em}
%     \caption{\atlex achieves smaller model prediction error, even after inserts.}
%     \label{fig:alex_errors}
% \end{figure}

% During inserts, time is spent on finding the insert position, which
% involves a lookup, as well as the number of shifts that occur during
% insertion.  \cref{fig:shifts} shows that the \kraskali, which uses a single Gap-less Array to store all the data, results in a high number of shifts, which is why it is not well suited
% for inserts.  In the static RMI layout, using PMA instead of the gapped array decreases the number of shifts per insert by 45$\times$, because the PMA avoids fully-packed regions.
% Alternatively, using adaptive RMI decreases the number of shifts
% that the gapped array performs by 37$\times$ because it limits the size of the leaf
% nodes, as discussed in \cref{subsec:adaptive_RMI}, which
% reduces the size and impact of fully-packed regions.
% Therefore, PMA and adaptive RMI are two ways of avoiding high shifts per insert resulting from fully-packed regions.

\begin{figure}
    % \begin{minipage}{0.25\linewidth}       
    %     % \centering
    %     % \includegraphics[width=\columnwidth]
    %     %         {figures/legend_shifts.png}
    %     % \\
    %     \includegraphics[width=\columnwidth,trim={8 0 8 8},clip]{figures/shifts.png}
    %     \vspace{-2em}
    %     \caption{Shifts per insert.}
    %     \label{fig:shifts}
    % \end{minipage}
    \begin{minipage}{0.48\linewidth} 
        % \centering
        % \includegraphics[width=\columnwidth]
        %         {figures/legend_latency.png}
        % \\
        \includegraphics[width=\columnwidth,trim={7 7 7 6},clip]{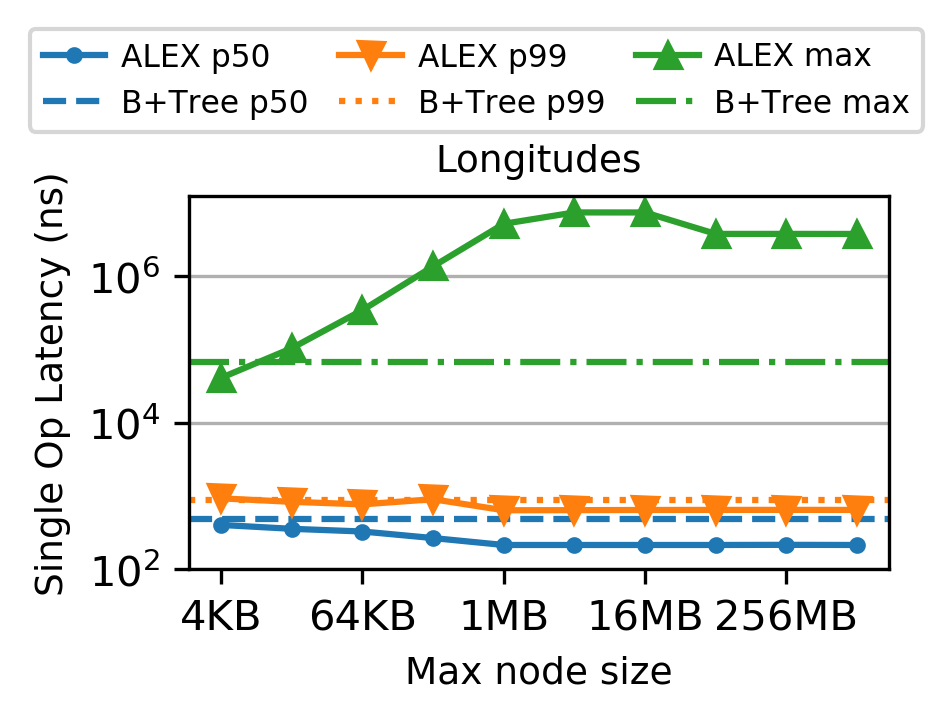}
        \vspace{-1em}
        \caption{Latency of a single operation.}
        \label{fig:latency}
    \end{minipage}
    % \begin{minipage}{0.25\linewidth} 
    %     % \centering
    %     % \includegraphics[width=\columnwidth]
    %     %         {figures/legend_space.png}
    %     % \\
    %     \includegraphics[width=\columnwidth,trim={7 7 7 7},clip]{figures/space_mixed.png}
    %     % \vspace{-2em}
    %     \caption{Data space usage.}
    %     \label{fig:space}            
    % \end{minipage} 
    \begin{minipage}{0.48\linewidth} 
        \centering
        \vspace{-1em}
        \includegraphics[width=\columnwidth,trim={7 10 7 7},clip]{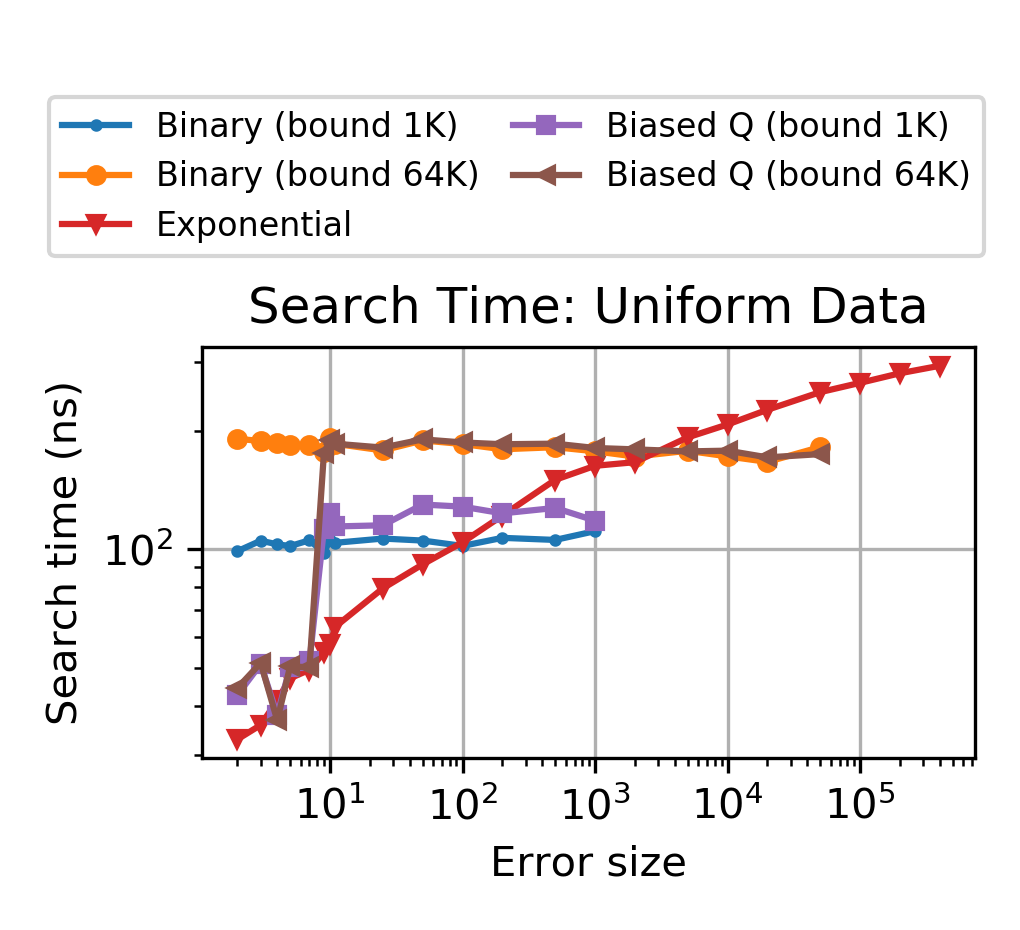}
        \vspace{-2em}
        \caption{\thirdrev{Exponential vs. other search methods.}}
        \label{fig:search_benchmark}      
    \end{minipage} 
    \vspace{-1em}
\end{figure}

Users can adjust the max node size to achieve target tail latencies, if desired.
In~\cref{fig:latency}, we run the write-heavy workload on the longitudes dataset, measuring the latency for every operation. As we increase the max node size, median and even p99 latency of \atlex decreases, because \atlex has more flexibility to build a better-performing RMI (e.g., ability to have higher internal node fanout). However, maximum latency increases, because an insert that triggers an expansion or split of a large node is slow.
If the user has strict latency requirements, they can decrease the max node size accordingly.
After increasing the max node size beyond 64MB, latencies do not change because \atlex never decides to use a node larger than 64MB.

\subsubsection{Search Method Comparison}
\label{subsec:search_comparison}
In order to show the trade-off between exponential search and other search methods, we perform a microbenchmark on synthetic data.
We create a dataset with 100 million perfectly uniformly distributed doubles.
We then perform searches for 10 million randomly selected values from this dataset.
\thirdrev{We use three search methods: binary search and biased quaternary search (proposed in~\cite{kraska2018case} to take advantage of accurate predictions), each evaluated with two different error bound sizes, as well as exponential search.}
For each lookup, the search method is given a predicted position that has some synthetic amount of error in the distance to the actual position value.
\cref{fig:search_benchmark} shows that the search time of exponential search increases proportionally with the logarithm of error size, whereas the binary search methods take a constant amount of time, regardless of error size. This is because binary search must always begin search within its error bounds, and cannot take advantage of cases when the error is small.
Therefore, exponential search should outperform binary search if the prediction error of the RMI models in \atlex is small.
As we showed in \cref{subsec:drilldown}, \atlex maintains low prediction errors through model-based inserts.
Therefore, \atlex is well suited to take advantage of exponential search.
\thirdrev{Biased quaternary search is competitive with exponential search when error is below $\sigma$ (we set $\sigma=8$ for this experiment; see~\cite{kraska2018case} for details) because search can be confined to a small range, but performs similarly to binary search when error exceeds $\sigma$ because the full error bound must be searched.
We prefer exponential search to biased quaternary search due to its smoother performance degradation and simplicity of implementation (e.g., no need to tune $\sigma$).}

% \begin{figure}

% \end{figure}

%% \subsubsection{Summary of \atlex Variants}
%% Throughout this evaluation, we have shown where each of the \atlex variants works best. \atlexgasrmi works best on read-only workloads because the gapped array and static RMI are optimized for lookups, but performs poorly for inserts because it does nothing to avoid fully-packed regions. \atlexgaarmi and \atlexpmasrmi avoid fully-packed regions by using adaptive RMI and PMA, respectively, and both variants perform well on read-write workloads. In general, \atlexgaarmi works better for datasets whose distributions are simpler to model, while for highly non-linear datasets, adaptive RMI is unable to completely avoid fully-packed regions and \atlexpmasrmi is the better choice. Lastly, \atlexpmaarmi works best on adversarial read-write workloads where both PMA and adaptive RMI are required to avoid fully-packed regions.

%% \subsection{\atlex Extension to Persistent Memory and SSD}
%% \label{subsec:persistence}

%% \todo{Sketch of an outline:
%% --extended ALEX to support secondary storage, with two variations (1) Pmem (2) ssd/hdd.
%% --Briefly describe design of (1) and (2)
%% --Show one throughput chart that shows ALEX's perf for: DRAM, PMEM, and SSD, on a single chart
%% --Say that preliminary results are promising, we leave the rest to future work
%% }

%% file: sections/5-related.tex
\section{Related Work}\label{sec:related}

% There is a rich body of related work on optimizing indexes which
% inspired us.

\textbf{Learned Index Structures:} The most relevant work is the
\kraskali~\cite{kraska2018case},
discussed in \cref{subsec:learned_indexes}. \kraskali has similarities to prior work that explored how to compute down a tree index.
Tries~\cite{knuth1997art} use key prefixes
instead of \bptree splitters. Masstree~\cite{mao2012cache} and Adaptive Radix
Tree~\cite{leis2013adaptive} combine the ideas of \bptree and trie to
reduce cache misses. Plop-hashing~\cite{kriegel1988plop} uses
piecewise linear order-preserving hashing to distribute keys more
evenly over pages. Digital B-tree~\cite{lomet1981digital} uses
bits of a key to compute down the tree more flexibly. ~\cite{Lomet:1987:PEF:12047.12049} proposes to partially expand the space instead of
always doubling when splitting in \bptree. ~\cite{graefe2006b} proposes the idea of
interpolation search within \bptree nodes;
% This is used in SQL
% Server's implementation of \bptree.
%instead of using
%binary search, we recursively use interpolation search to locate the
%position of a key, or first use interpolation search to find a
%position and then do a local search around that position, similar to
%how search works in a learned index.
this idea was
revisited in~\cite{neumann-blog}. The interpolation-based
search algorithms in~\cite{interpolsearch2019} can
complement ALEX's search strategy.
Hermit~\cite{wu2019hermit} creates a succinct tree structure for secondary indexes.

% Hermit~\cite{wu2019hermit} creates a succinct tree
% structure to learn the correlation between the target column and
% another indexed column. Given a query on the target column, Hermit
% predicts a corresponding range on the correlated column and navigates
% the query to that column which already has a full index.

% Other works propose replacing the leaf nodes of a \bptree with other
% data structures. In these works, the inner nodes of \bptree structure
% remain unchanged. Since leaf nodes take the most space in a \bptree,
% these works change the leaf nodes to compress the index, while
% maintaining search and update
% performance. FITing-tree~\cite{galakatos2018tree} uses linear models in its
% leaf nodes, while BF-tree~\cite{athanassoulis2014bf} uses bloom filters
% in its leaf nodes.

Other works propose replacing the leaf nodes of a \bptree with other data structures in order to compress the index, while
maintaining search and update
performance. FITing-tree~\cite{galakatos2018tree} uses linear models in its
leaf nodes, while BF-tree~\cite{athanassoulis2014bf} uses bloom filters
in its leaf nodes.

All these works share the idea that using extra computation or data structures can make search
faster by reducing the number of
binary search hops and corresponding cache misses, while allowing larger node sizes and hence a smaller index
size. However, \atlex is different in several ways:
(1) We use a model to split the key space, similar to a trie,
but no search is required until we reach the leaf level. (2) \atlex's
accurate linear models enable larger node sizes without
sacrificing search and update performance. (3) Model-based
insertion reduces the impact of model's misprediction. (4) \atlex's cost models automatically adjust the index structure to dynamic workloads.

\textbf{Memory Optimized Indexes:} There is a large body of work
on optimizing tree index structures for main memory by exploiting
hardware features such as CPU cache, multi-core, SIMD, and prefetching.
CSS-trees~\cite{RaoR99}
improve \bptree's cache behavior by matching index node size to CPU cache-line size and eliminating pointers in index nodes
by using arithmetic operations to find child nodes. CSB$^{+}$-tree~\cite{rao2000making} extends the static CSS-trees by supporting incremental updates
without sacrificing CPU cache performance. \cite{HankinsP03} evaluates the effect of node size on
the performance of CSB$^{+}$-tree analytically and empirically.
pB+-tree~\cite{ChenGM01} uses larger index nodes and relies on
prefetching instructions to bring index nodes into cache before nodes
are accessed. In addition to optimizing for cache performance,
FAST~\cite{kim2010fast} further optimizes searches within index nodes
by exploiting SIMD parallelism.

\textbf{ML in other DB components:}
%In the emerging trend of using machine learning to optimize
%databases, the most explored area is query optimization.
%Multiple
%studies~\cite{kipf2018learned,Dutt2019selectivity,Marcus2019Neo}
%propose ideas to use regression or reinforcement learning to predict cardinalities or to optimize join queries. Also,
%machine learning is used to predict arrival rate of
%queries~\cite{ma2018query}.
Machine learning has been used to improve cardinality estimation~\cite{kipf2018learned,Dutt2019selectivity}, query optimization~\cite{Marcus2019Neo}, workload forecasting~\cite{ma2018query}, multi-dimensional indexing~\cite{nathan2020flood}, and data partitioning~\cite{qdtree-sigmod}. 
SageDB~\cite{kraska2019sagedb} envisions a database system in which every component is replaced by a learned component.
% do entity matching~\cite{mudgal2018deep},
% and to predict memory access patterns for
% prefetching~\cite{hashemi2018learning}. 
These studies show that the use of machine learning enables
workload-specific optimizations, which also inspired our work.

%% file: sections/7-conclusion.tex
\section{Conclusion}
\label{sec:conclusion}

We build on the excitement of learned indexes by proposing \atlex, a
new updatable learned index that effectively
combines the core insights from the \kraskali with proven storage and indexing techniques.
Specifically, we propose a Gapped Array node layout that uses model-based inserts and exponential search, combined with an adaptive RMI structure driven by simple cost models, to achieve high performance and low memory footprint on dynamic workloads.
%Specifically, we show that a careful space-time trade-off leads to an updatable data structure and
%also significantly improves search performance. \atlex proposes a
%Gapped Array node layout and techniques to dynamically and
%automatically adapt the RMI structure to updates.
% Given this flexibility, \atlex can adapt to
% different datasets and workloads.
Our in-depth experimental results show that \atlex not only consistently beats \bptree across the read-write workload spectrum, it even beats the existing \kraskali, on all datasets,
by up to 2.2$\times$ with read-only workloads.
% ~\todo{UPDATE THE RESULTS.} Our
% in-depth experimental results show that \atlex not only beats a
% \bptree on three of the four datasets, using read-only and read-write
% workloads, it even beats the existing learned index, on all datasets,
% by up to 2.7$\times$ with read-only workloads.
% \atlex is an important
% case study in this new and exciting space.

We believe this paper presents important learnings to our community
and opens avenues for future research in this area. We intend to
pursue \firstrev{open theoretical problems about \atlex performance}, supporting secondary storage for larger than memory datasets,
and new concurrency control techniques tailored to the \atlex design.

%% file: sections/appendix-sigmod.tex
% \section{Supplemental Appendices}
% \label{sec:appendix}
% \addcontentsline{toc}{section}{Appendices}
\renewcommand{\thesubsection}{\Alph{subsection}}

\appendix
\subsection{Extended Bulk Loading Evaluation}
\label{sec:bulk_loading_eval}
In this section, we provide an extended version of \cref{subsubsec:main_bulk_loading_eval}, which evaluates the speed of \atlex's bulk loading mechanism against other indexes.
%For \atlex, we introduce two optimizations that decrease the time for bulk loading: (1) approximate model computation (AMC), which uses sampling to compute the linear regression models for each node; and (2) approximate cost computation (ACC), which uses sampling to compute the intra-node cost for each data node (\cref{subsubsec:cost_models}).
%These optimizations can be used when faster bulk loading time is desired, but they are by no means necessary for \atlex to achieve high performance on workloads.
\atlex uses two optimizations for bulk loading---approximate model computation (AMC) and approximate cost computation (ACC)---which we explain in more detail below.

\begin{figure}
    \includegraphics[width=\columnwidth,trim={8 8 7 8},clip]{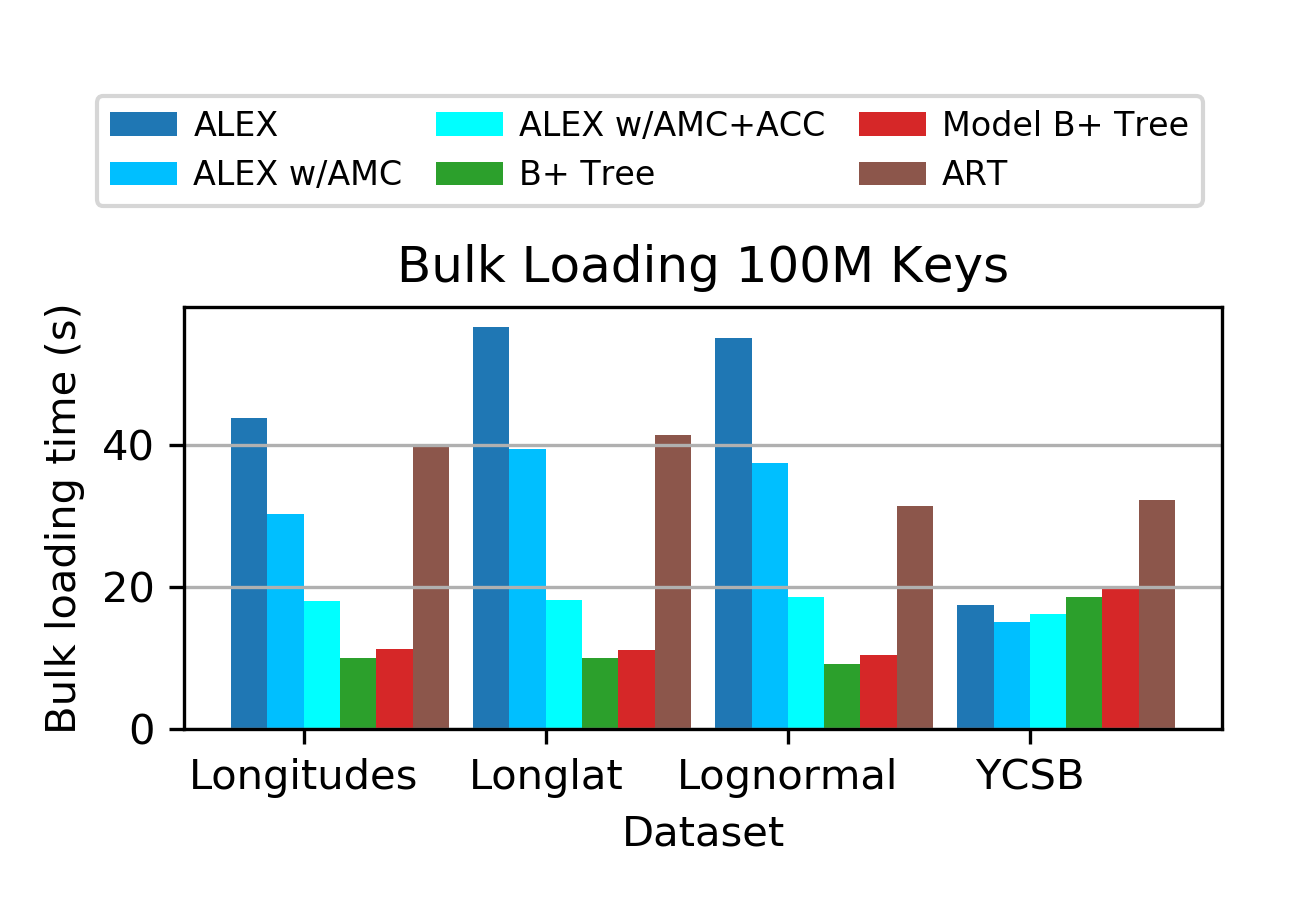}
    \caption{
        With both optimizations (AMC and ACC), \atlex only takes 50\% more than time than \bptree to bulk load when averaged across four datasets.
    }
    \label{fig:bulk_loading_time}
\end{figure}

For each index, we bulk load 100 million keys from each of the four datasets from \cref{sec:exp}.
This includes the time used to sort the 100 million keys.
\cref{fig:bulk_loading_time}, which is a more detailed version of \cref{fig:simple_bulk_loading_time}, shows that on average, \atlex with no optimizations takes 3.6$\times$ more time to bulk load than \bptree, which is the fastest index to bulk load.
However, with the AMC optimization, \atlex takes 2.6$\times$ more time on average than \bptree.
With both optimizations, \atlex only takes 50\% more than time on average than \bptree, and in the worst case is only 2$\times$ slower than \bptree.
The results for \atlex in \cref{subsubsec:main_bulk_loading_eval} use both optimizations.
On the YCSB dataset, \atlex's structure is very simple (\cref{tab:node_size}), and therefore is very efficient to bulk load even when unoptimized; in fact, \atlex's large node sizes allow \atlex to bulk load faster than other indexes due to the benefits of locality.
%\modelbtree is slightly slower to bulk load than \bptree due to the overhead of training linear regression models for each node.
%\art is slower to bulk load than \bptree, \modelbtree, and \atlex with both optimizations.

%Note that \atlex can quickly make up for its slower bulk loading time than \bptree by having higher throughput performance than other indexes.
%For example, on the longitudes dataset, \atlex with the AMC optimization spends 20.4 more seconds than \bptree on bulk loading.
%\cref{fig:punchline_mixed} shows that on the read-heavy workload, \atlex performs 6.4 million more operations per second than \bptree.
%Therefore, after only around 6.5 million inserts, \atlex will have used less overall time than \bptree.

\begin{figure}
	\centering
	\includegraphics[width=\columnwidth]
	{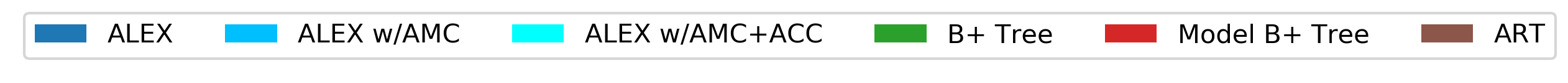}
	\vspace{-2em}
	\\
    \subfloat{
        \includegraphics[width=0.48\columnwidth,trim={5 8 7 8},clip]{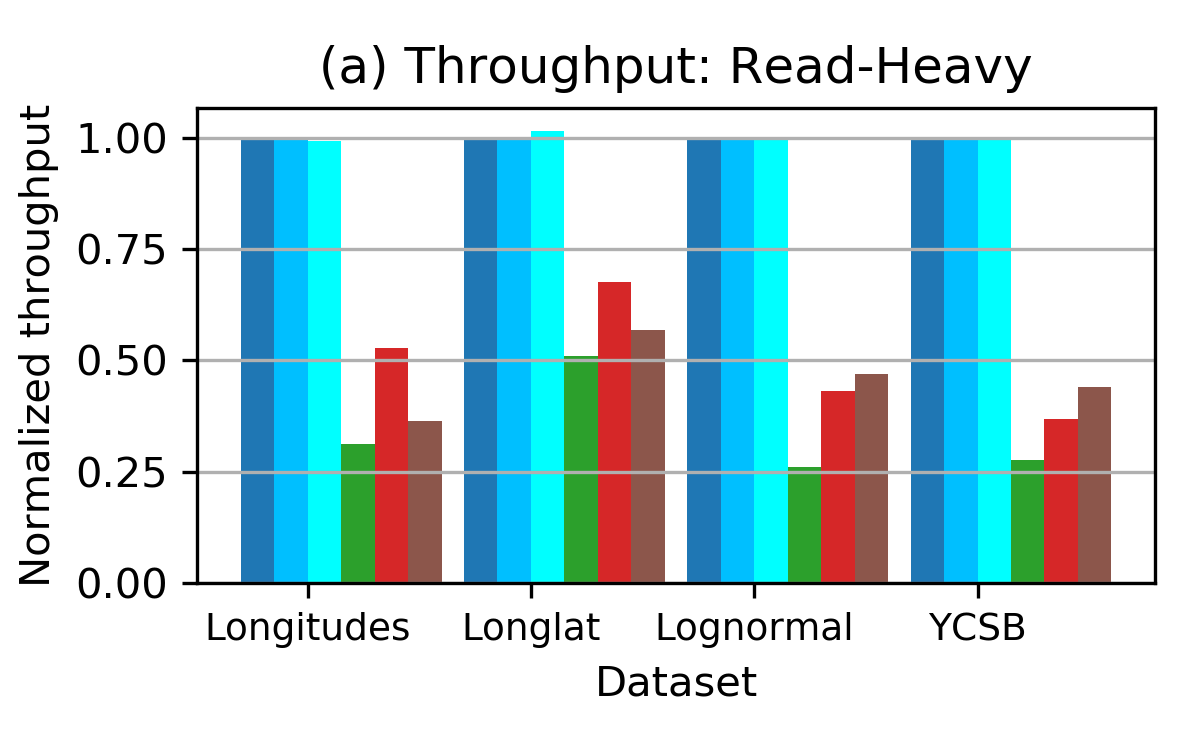}
        }
    ~
    \subfloat{
        \includegraphics[width=0.48\columnwidth,trim={5 8 7 8},clip]{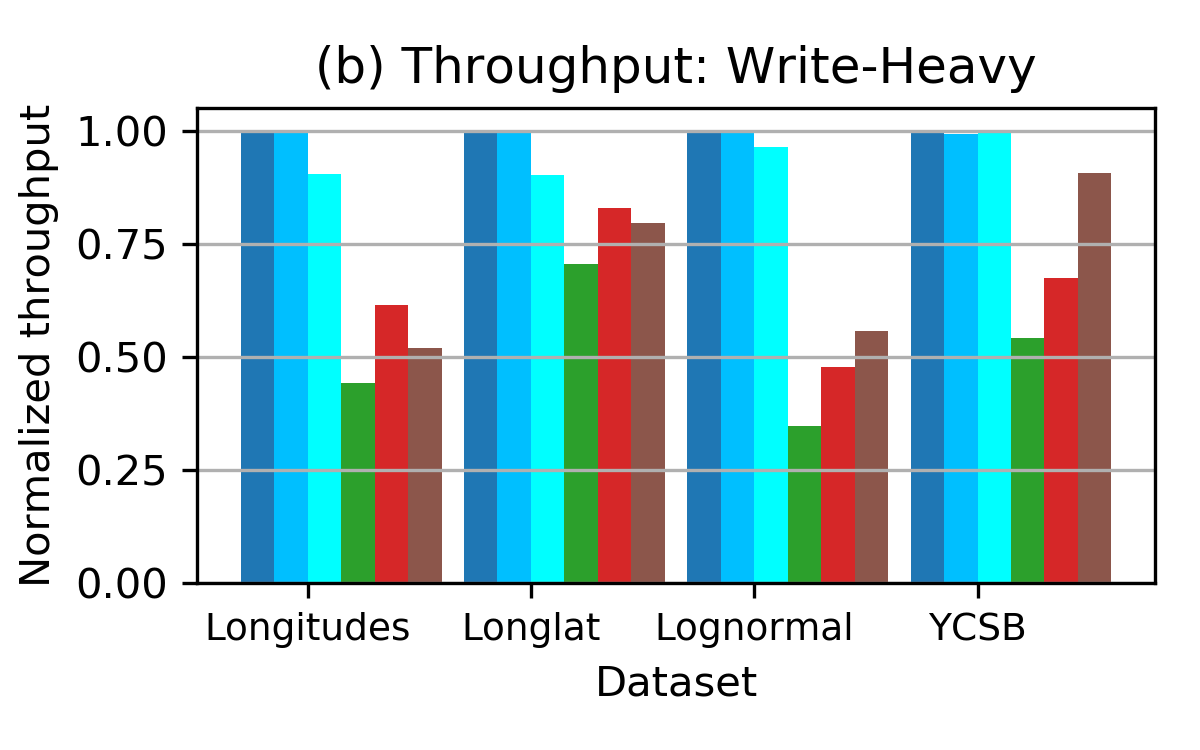}
        }
    \caption{
        Using the AMC optimization when bulk loading does not cause any noticeable change in \atlex performance, but ACC can cause a slight decrease in throughput for write-heavy workloads.
    }
    \label{fig:bulk_load_throughput}
\end{figure}

\cref{fig:bulk_load_throughput} shows the impact of the two optimizations on the throughput of running a read-heavy or write-heavy workload on \atlex after bulk loading.
The AMC optimization has negligible impact on throughput for all datasets and for both workloads.
Adding the ACC optimization has negligible impact on the read-heavy workload, but decreases throughput by up to 9.6\% on the write-heavy workload; we provide explanation for this behavior below.

Based on these results, we conclude that the AMC optimization should always be used to improve bulk loading performance, whereas the ACC optimization might cause a slight decrease in throughput performance and therefore should be used only if faster bulk loading is required.
We now explain the two optimizations in more detail.

\subsubsection{Approximate Model Computation}
\label{sec:amc}
We perform approximate model computation (AMC) efficiently while achieving accuracy through \emph{progressive systematic sampling}.
Given a data node of sorted keys, we perform \emph{systematic sampling} (i.e., sampling every $n$th key) to obtain a small sample of keys, and compute a linear regression model using that sample.
We then repeatedly double the sample size and recompute the linear model using the larger sample.
When the relative change in the model parameters (i.e., the slope and intercept) both change by less than 1\% from one sample to the next, we terminate the process.
In our experience, this 1\% threshold strikes a balance between achieving accuracy and using small sample sizes, and did not need to be tuned.

Note that by using systematic sampling, all keys in the sample used to compute the current model will also appear in all subsequent samples.
Therefore, each linear model can be computed \emph{progressively} starting from the existing model computed from the previous sample, instead of from scratch.
No redundant work is done, and even in the worst case, AMC will take no more time than computing one linear model from all keys (if we ignore minor overheads and the effects of locality).
%By default we set the size of the initial sample to be smallest number larger than 10 that would guarantee an $n$ that is a power of two.

\subsubsection{Approximate Cost Computation}
\label{sec:acc}
We also perform approximate intra-node cost computation (ACC) for a data node of sorted keys through progressive systematic sampling.
However, ACC differs from AMC in two ways.
First, the cost for a data node must be computed from scratch for each sample; the cost depends on where keys are placed within a Gapped Array, which itself depends on which keys are present in the sample.
Second, and more importantly, the cost of a data node naturally increases with the number of keys in the node.
This makes ACC an extrapolation problem (i.e., use the cost of a small sample to predict the cost of the entire data node), whereas AMC is an estimation problem (i.e., use a small sample to directly estimate the model parameters).

ACC repeatedly doubles its sample size.
Let the latest three samples be $s_1$, which is half the size of $s_2$, which is half the size of $s_3$.
Let the costs computed from these samples be $c_1$, $c_2$, and $c_3$, respectively.
We use the $c_1$ and $c_2$ to perform a linear extrapolation to predict $c_3$.
If this prediction is accurate (i.e., if relative error with the true $c_3$ is within 20\%), then we use $c_2$ and $c_3$ to perform a linear extrapolation to predict the cost of the entire data node, and we terminate the process.\footnote{In reality, we do not predict the cost directly, but rather each component of the cost (search iterations per lookup and shifts per insert) independently. This is because the expected extrapolation behavior differs: iterations per lookup grows logarithmically with sample size, whereas shifts per insert grows linearly with sample size.}
Otherwise, we continue doubling the sample size.
The intuition behind this process is that we want to verify the accuracy of extrapolation using small samples before extrapolating to the entire data node.
We allow a higher relative error than for AMC because the extrapolation process is inherently imprecise, since it is impossible to accurately predict the cost using a sample without a priori knowledge of the data distribution.

We can now explain why \cref{fig:bulk_load_throughput} shows that adding the ACC optimization decreases throughput on the write-heavy workload by up to 9.6\%.
It is because the average number of shifts per insert, which is one component of the intra-node cost, is difficult to estimate accurately.
Therefore, if ACC underestimates the component of cost related to shifts, the bulk loaded \atlex structure may be inefficient for inserts (e.g., an insert that requires more shifts than expected can be very slow).
The intra-node cost is more difficult to approximate accurately for the longitudes and longlat datasets, which is why the decrease in throughput is most noticeable for those two datasets.
However, note that over time, the dynamic nature of \atlex will eventually correct for incorrectly estimated costs, so throughput performance in the long run will be independent of the bulk loading mechanism.

\subsection{Extreme Distribution Shift Evaluation}
\label{sec:extreme_distshift_eval}
In order to evaluate the performance of \atlex under a radically changing key distribution, we combine the four datasets from \cref{sec:exp} into one dataset by randomly selecting 50 million keys from each of the four datasets in order to create one combined dataset with 200 million keys.
We scaled keys from each dataset to fit in the same domain.
Note that we would not typically expect a single table to contain keys from four independent distributions.
Therefore, this complex combined dataset is an extreme stress test for the adaptibility of \atlex.

We run a write-heavy workload (50\% point lookups and 50\% inserts) over the combined dataset, but we vary the order in which keys are bulk loaded and inserted.
For all variants, we bulk load using 50 million keys and run the write-heavy workload until the remaining 150 million keys are all inserted.
We create four variants that represent distribution shift; each variant bulk loads using the 50 million keys selected from one of the four original datasets, then gradually inserts keys from the other three original datasets, in order.
For example, the ``L-LN-LL-Y'' variant bulk loads using the 50 million keys selected from the longitudes (L) dataset, then runs the write-heavy workload by inserting the 50 million keys from the lognormal (LN) dataset, then the longlat (LL) dataset, and finally the YCSB (Y) dataset.
For reference, we also include a variant in which all 200 million keys are shuffled, so that no key distribution shift is observed.

\begin{figure}
    \includegraphics[width=\columnwidth,trim={8 8 7 8},clip]{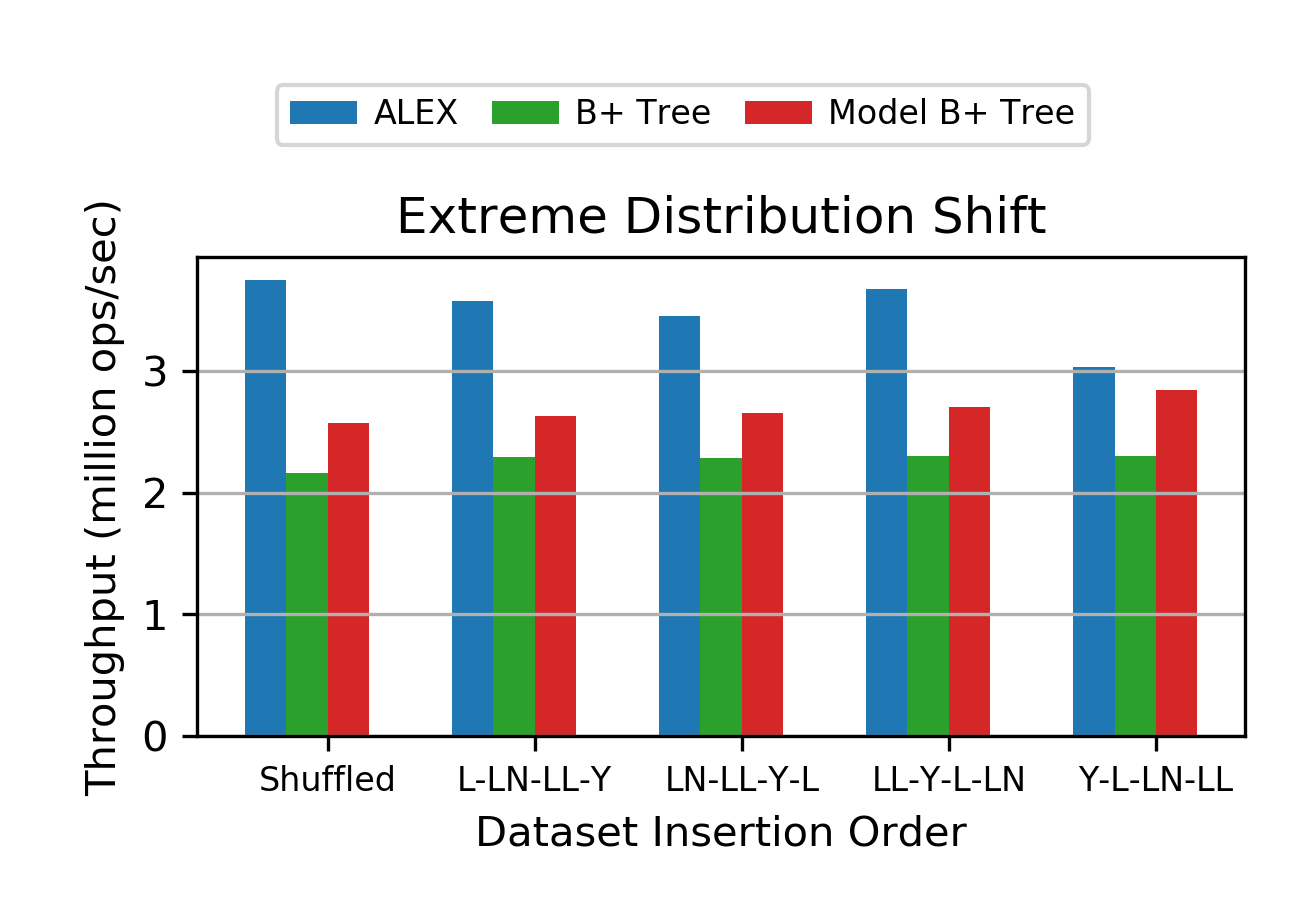}
    \caption{
        \atlex maintains high performance under radically changing key distribution, although performance does differ slightly depending on the distribution used for bulk loading.
    }
    \label{fig:combined_distshift}
\end{figure}

\cref{fig:combined_distshift} provides three insights.
First, on the workload that represents no distribution shift (``Shuffled''), \atlex continues to outperform other indexes.
It is interesting to note that the throughput of \atlex on the combined dataset is between the throughputs achieved on each dataset individually (\cref{fig:punchline_write}): higher than for longlat, and lower than for the other three datasets.
Second, \atlex achieves lower throughput in the four variants that represent distribution shift than without distribution shift, but still outperforms other indexes.
This result aligns with the intuition that \atlex must spend extra time restructuring itself to adapt to the changing key distribution.
Third, the throughput differs based on which dataset's keys are used to bulk load \atlex.
When bulk loading using keys from a complex key distribution, such as longlat, \atlex achieves throughput similar to the variant with no distribution shift; on the other hand, when bulk loading using keys from a simple key distribution, such as YCSB, \atlex throughput suffers.
This is because when bulk loading with a simple key distribution, the bulk loaded structure of \atlex will be shallow, with few nodes (\cref{tab:node_size}). When the subsequently inserted keys come from a much more complex key distribution, \atlex must quickly adapt its structure to be deeper and have more nodes, which can incur significant overhead.
On the other hand, when bulk loading with a complex key distribution, the bulk loaded structure is already deep, with many nodes, and so can more readily adapt to changes in the key distribution without too much overhead.

To allow \atlex to more quickly adapt the RMI structure to radically changing key distributions: (1) we check data nodes periodically for cost deviation instead of only when the data node is full, and (2) if the number of shifts per insert in a data node is extremely high, we force the data node to split (as opposed to expanding and retraining the model). When no distribution shift occurs, these two checks  have negligible impact on performance, because checking for cost deviation has minimal overhead and cost deviation occurs infrequently (\cref{tab:node_splits}). By default, we check for cost deviation for every 64 inserts into that data node, and over 100 shifts per insert is considered extremely high.

\subsection{Extended Range Query Evaluation}
\label{sec:range_scan_appendix}
\begin{figure*}[th!]
	\centering
	\includegraphics[width=\columnwidth]
	{figures/legend_selectivity.png}
	\vspace{-1em}
	\\
	\subfloat{
		\includegraphics[width=0.24\textwidth,trim={10 10 10 8},clip]{figures/selectivity_throughput_longitudes.png}
		\label{fig:selectivity_throughput_longitudes}
	}
	~
	\subfloat{
		\includegraphics[width=0.24\textwidth,trim={9 10 10 8},clip]{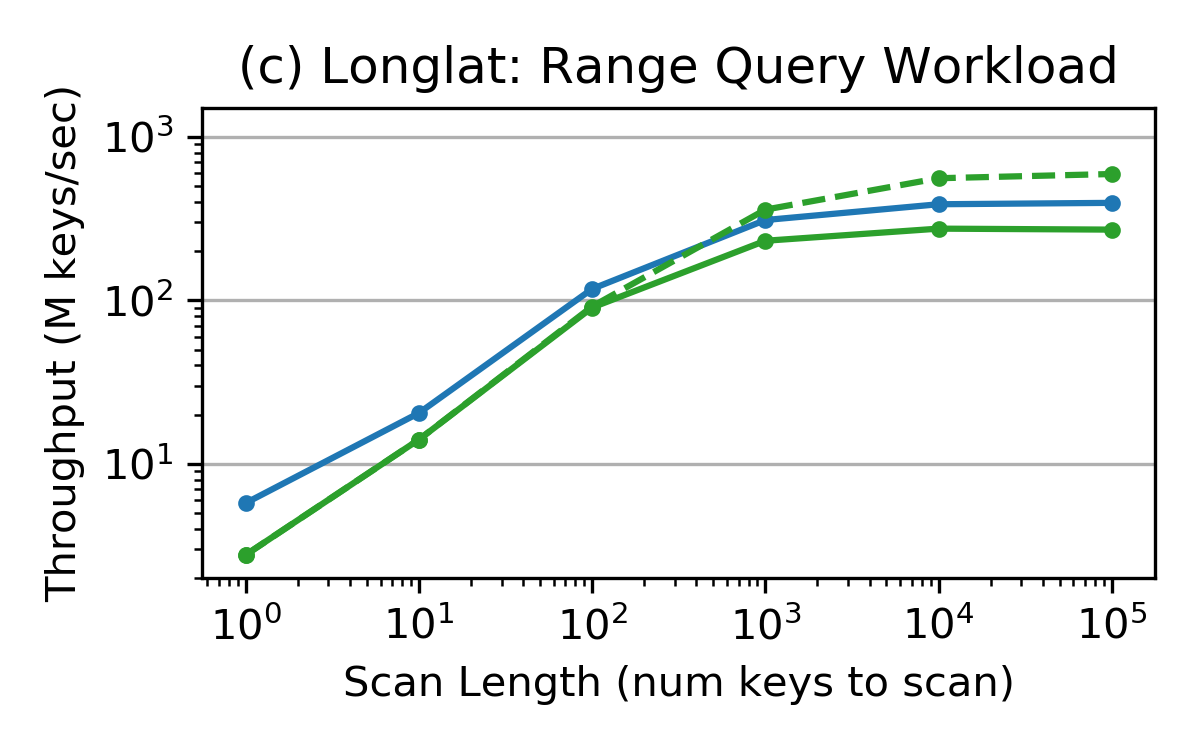}
		\label{fig:selectivity_throughput_longlat}
	}
	~
	\subfloat{
		\includegraphics[width=0.24\textwidth,trim={9 10 10 8},clip]{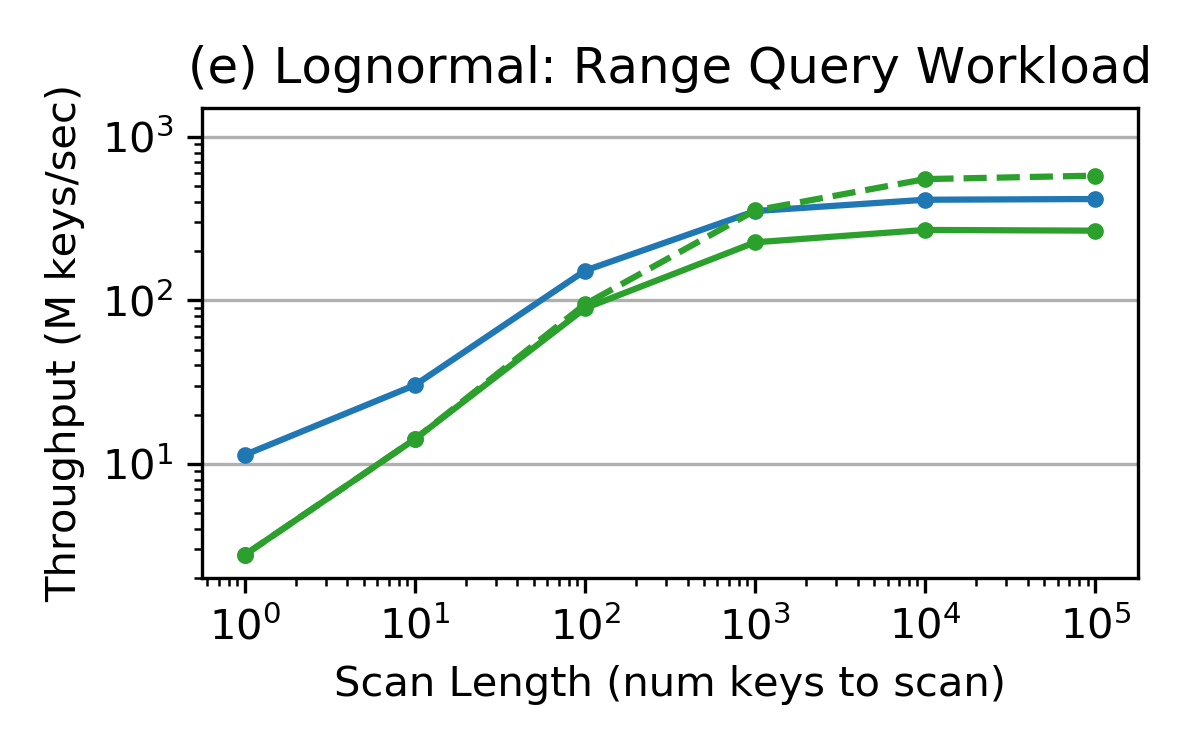}
		\label{fig:selectivity_throughput_lognormal}
	}
	~
	\subfloat{
		\includegraphics[width=0.24\textwidth,trim={9 10 10 8},clip]{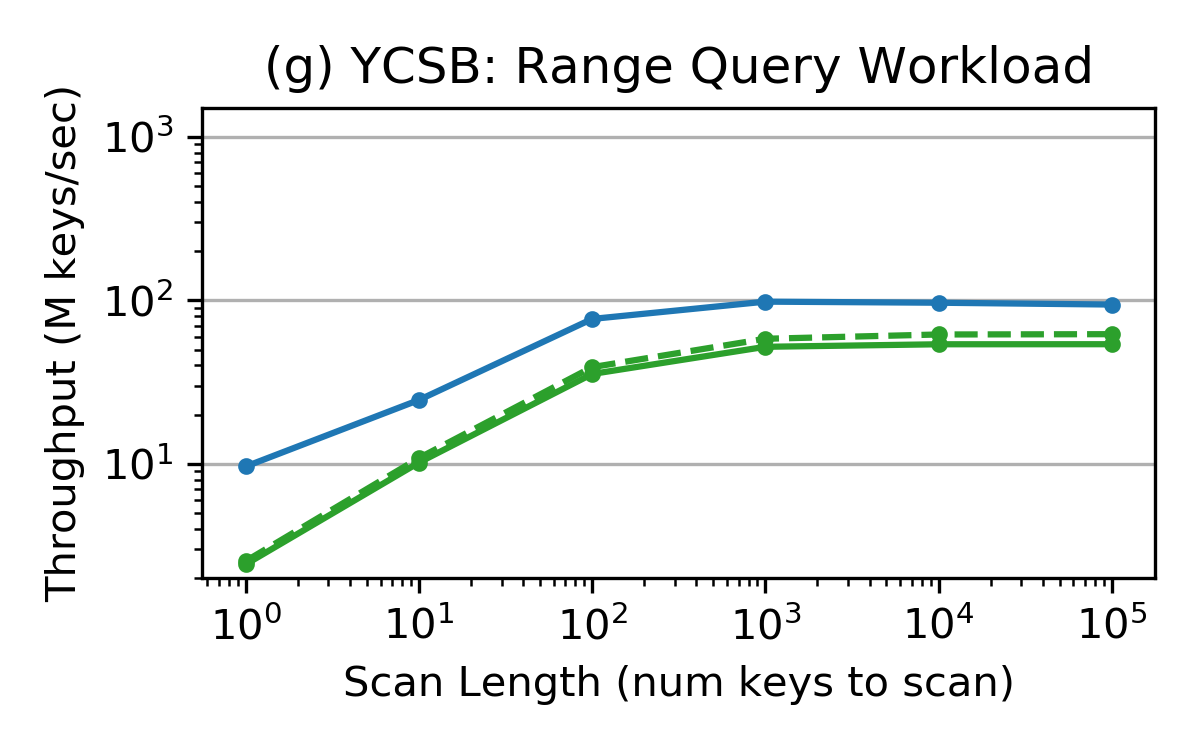}
		\label{fig:selectivity_throughput_ycsb}
	}
	\vspace{-0.5em}
	\\
	\subfloat{
		\includegraphics[width=0.24\textwidth,trim={10 10 10 8},clip]{figures/selectivity_counterpoint_longitudes.png}
		\label{fig:counterpoint_longitudes}
	}
	~
	\subfloat{
		\includegraphics[width=0.24\textwidth,trim={9 10 10 8},clip]{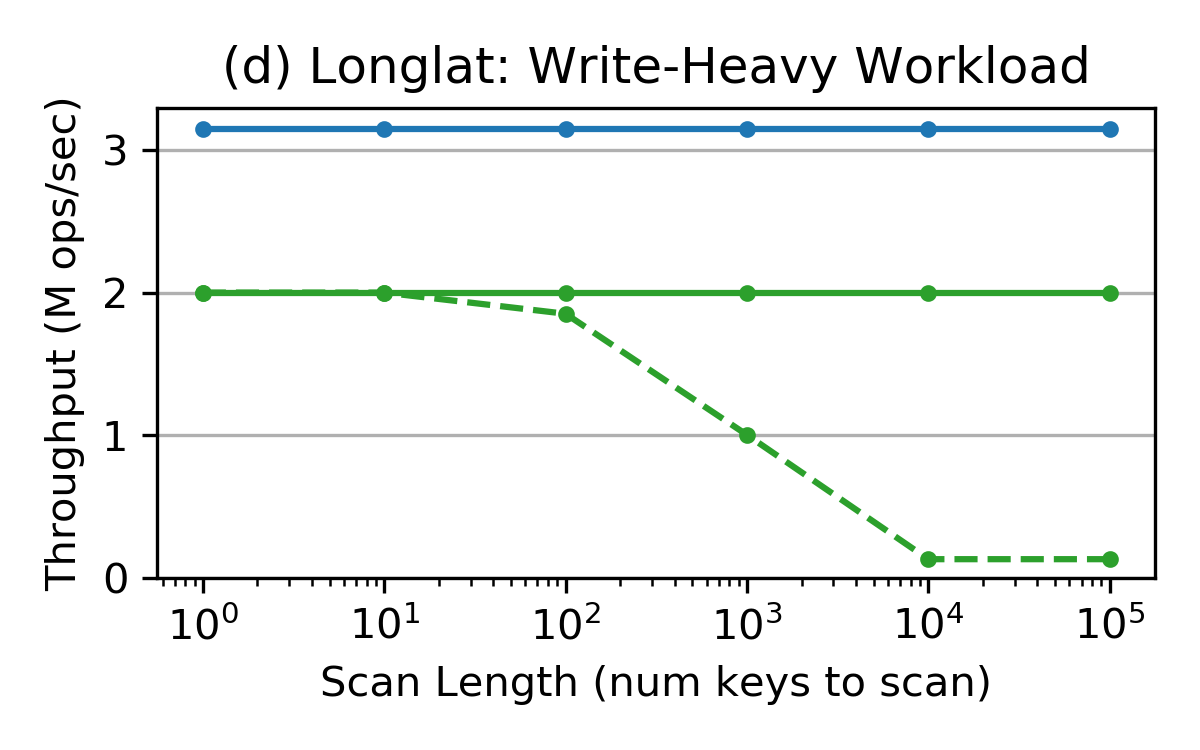}
		\label{fig:counterpoint_longlat}
	}
	~
	\subfloat{
		\includegraphics[width=0.24\textwidth,trim={9 10 10 8},clip]{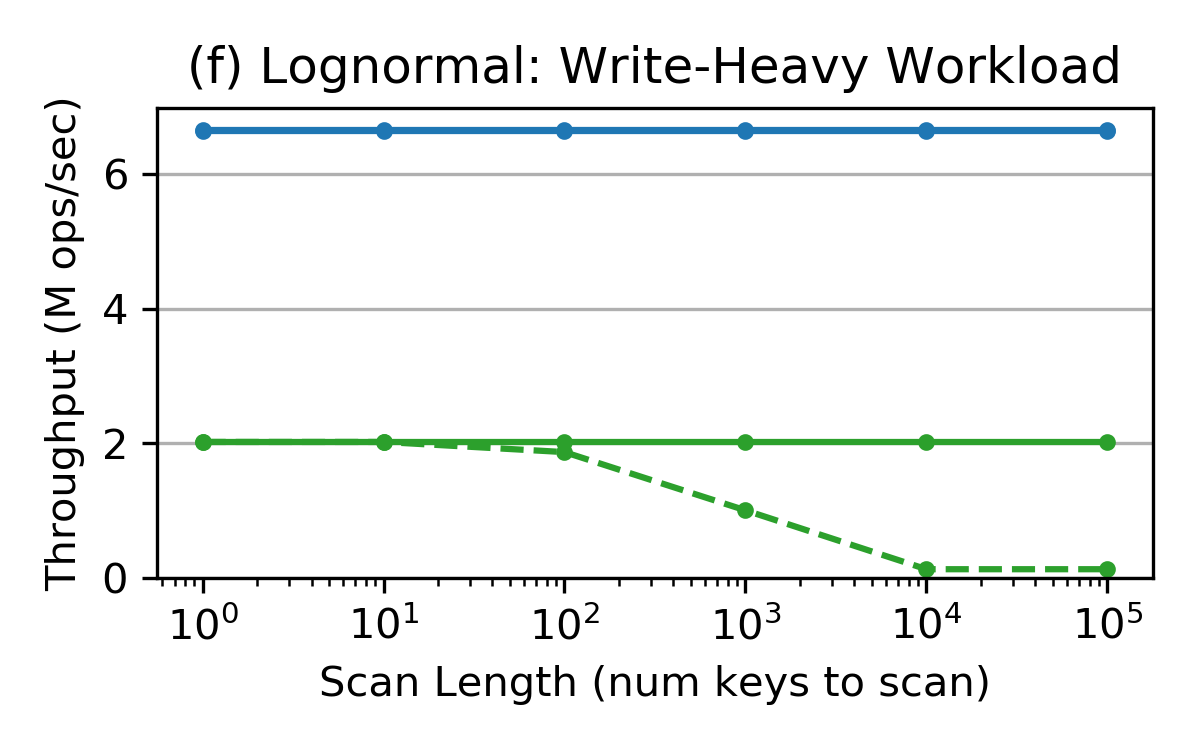}
		\label{fig:counterpoint_lognormal}
	}
	~
	\subfloat{
		\includegraphics[width=0.24\textwidth,trim={9 10 10 8},clip]{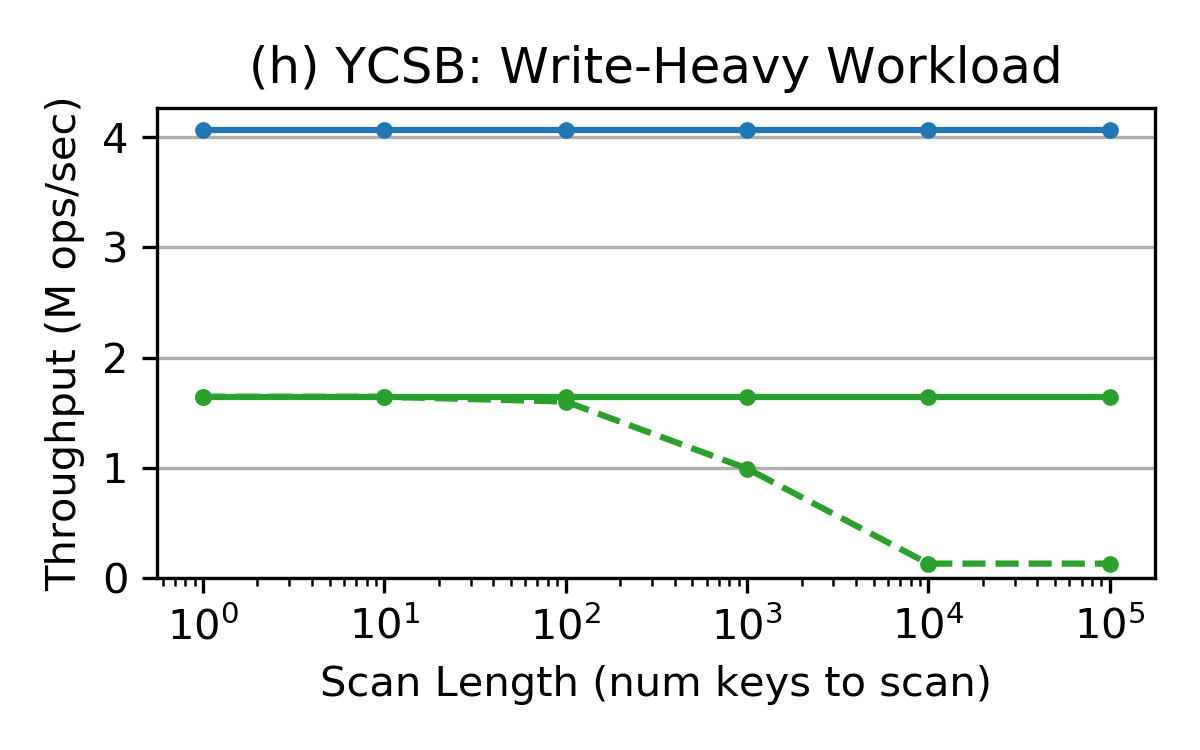}
		\label{fig:counterpoint_ycsb}
	}
	\caption{
		Across all datasets, \atlex maintains an advantage over fixed-page-size \bptree even for longer range scans.
	}
	\label{fig:selectivity_extended}
	\vspace{-1em}
\end{figure*}
\subsubsection{Varying Range Query Scan Length}
\label{subsec:range_scan_throughput}
We extend the experiment from \cref{fig:selectivity} to all four datasets.
\cref{fig:selectivity_extended} shows that across all datasets, \atlex maintains its advantage over fixed-page-size \bptree, and re-tuning the \bptree page size can lead to better range query performance but will decrease performance on point lookups and inserts.
For both \atlex and \bptree, performance on YCSB is slower than for the other three datasets because YCSB has a larger payload size, which worsens scan locality.

\subsubsection{Mixed Workload Evaluation}
\label{subsec:mixed_workload_eval}
\begin{figure}
    \centering
    \includegraphics[width=0.6\columnwidth]
    {figures/legend_scalability.png}
    \vspace{-1em}
    \\
    \subfloat{
    \includegraphics[width=\columnwidth,trim={8 8 7 8},clip]{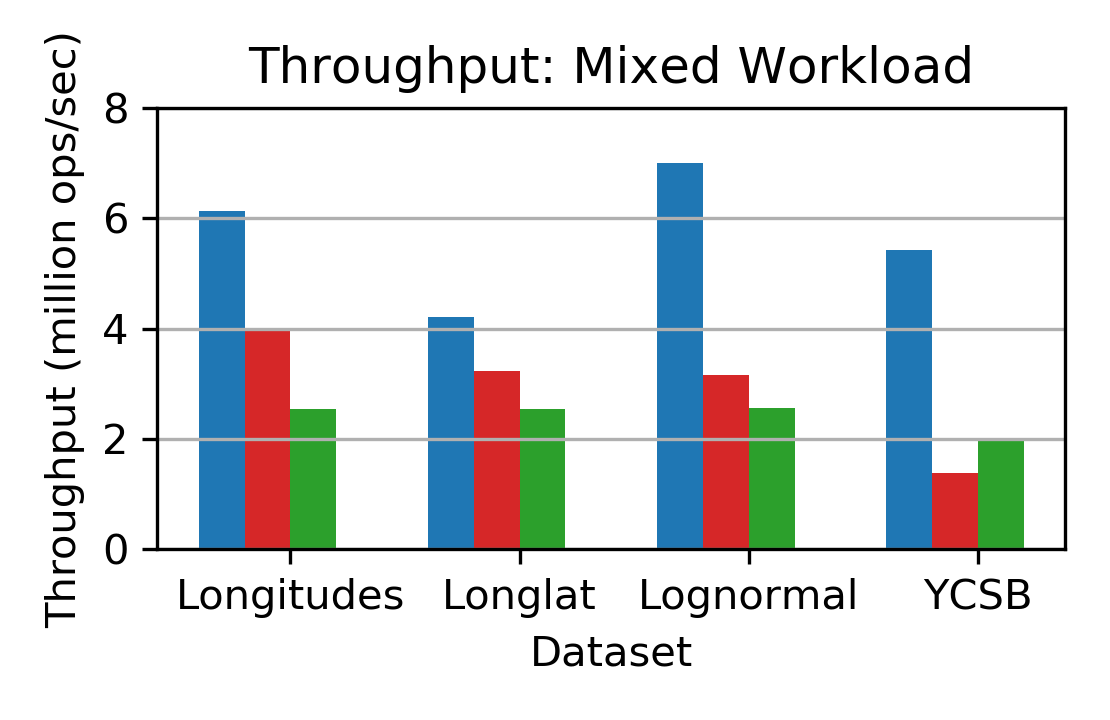}
    }
    \caption{
        \atlex maintains high performance under a mixed workload with 5\% inserts, 85\% point lookups, and 10\% short range queries.
    }
    \label{fig:mixed_workload}
\end{figure}
We evaluate a mixed workload with 5\% inserts, 85\% point lookups, and 10\% range queries with a maximum scan length of 100.
The remainder of the experimental setup is the same as in \cref{sec:exp_setup}.
\cref{fig:mixed_workload} shows that \atlex maintains its performance advantage over other indexes.
The \art implementation from~\cite{artimpl} does not support range queries.

\subsection{Drilldown into Cost Computation}
\label{sec:cost_drilldown}
In this section, we first provide more details about the cost model introduced in \cref{subsubsec:cost_models}.
We then evaluate the performance of computing costs using cost models.

\begin{table}[]
\centering
\caption{Terms used to describe the cost model}
\vspace{-1em}
\small
\label{tab:cost_model}
\begin{tabular}{@{}lll@{}}
\toprule
    \textbf{Term}     & \textbf{Description}   \\
\midrule
    $\mathcal{A}$ & An instantiation of \atlex \\
    $\mathcal{N}$ & Set of all nodes in $\mathcal{A}$ \\
    $\mathcal{D}$ & Set of data nodes in $\mathcal{A}$. This means that $\mathcal{D}\subseteq \mathcal{N}$ \\
    $S(N)$ & \makecell[l]{Average number of exponential search iterations \\ for a lookup in $N\in \mathcal{D}$} \\
    $I(N)$ & Average number of shifts for an insert into $N\in \mathcal{D}$ \\
    $K(N)$ & Number of keys in $N\in \mathcal{D}$ \\
    $F(N)$ & \makecell[l]{Fraction of operations that are inserts (as opposed \\ to lookups) in $N\in \mathcal{D}$} \\
    $C_I(N)$ & Intra-node cost of $N\in \mathcal{D}$ \\
    $C_T(N)$ & TraverseToLeaf cost of $N\in \mathcal{D}$ \\
    $D(N)$ & Depth of $N\in \mathcal{N}$ (root node has depth 0) \\
    $B(\mathcal{A})$ & Total size in bytes of all nodes in $\mathcal{A}$ \\
    $w_s,w_i,w_d,w_b$ & Fixed pre-defined weight parameters \\
\bottomrule
\end{tabular}
\end{table}

\subsubsection{Cost Model Details}
\label{sec:cost_model_details}
We formally define the cost model using the terms in \cref{tab:cost_model}.
At a high level, the intra-node cost of a data node represents the average time to perform an operation (i.e., a point lookup or insert) on that data node, and the TraverseToLeaf cost of a data node represents the time for traversing from the root node to the data node.

For a given data node $N\in \mathcal{D}$, the intra-node cost $C_I(N)$ is defined as
\begin{equation}
C_I(N) = w_sS(N) + w_iI(N)F(N)
\label{eq:intranode_cost}
\end{equation}
Both lookups and inserts must perform an exponential search, whereas only inserts must perform shifts.
This is why $I(N)$ is weighted by $F(N)$.

For a given data node $N\in \mathcal{D}$, the TraverseToLeaf cost $C_T(N)$ of traversing from the root node to $N$ is defined as
\begin{equation}
C_T(N) = w_dD(N) + w_bB(\mathcal{A})
\label{eq:traversetoleaf_cost}
\end{equation}
The depth of $N$ is the number of pointer chases needed to reach the data node.
In our cost model, every traverse to leaf has a fixed cost that is caused by the total size of the \atlex RMI, because larger RMI causes worse cache locality.

For an instantiation of \atlex $\mathcal{A}$, the cost of $\mathcal{A}$ represents the average time to perform a query (i.e., a point lookup or insert) starting from the root node, and is defined as
\begin{equation}
C(\mathcal{A}) = \frac{\sum_{N\in \mathcal{D}}(C_I(N) + C_T(N))K(N)}{\sum_{N\in \mathcal{D}}K(N)}
\label{eq:alex_cost}
\end{equation}
In other words, the cost of $\mathcal{A}$ is the sum of the intra-node cost and TraverseToLeaf cost of each data node, normalized by how many keys are contained in the data node.
We normalize because each data node does not contribute equally to average query time.
For example, a data node that has high intra-node cost but is rarely queried might not have as much impact on average query time as a data node with lower intra-node cost that is frequently queried.
We use the number of keys in each data node as a proxy for its impact on the average query time.
An alternative is to normalize using the true query access frequency of each data node.

The weight parameters $w_s,w_i,w_d,w_b$ do not need to be tuned for each dataset or workload, because they represent fixed quantities.
For our evaluation, we set $w_s=10$,$w_i=1$,$w_d=10$, and $w_b=10^{-6}$.
In terms of impact on throughput performance, these weights intuitively mean that each exponential search iteration takes 10 ns, each shift takes 1 ns, each pointer chase to traverse down one level of the RMI takes 10 ns, and each MB of total size contributes a slowdown of 1 ns due to worse cache locality.
As a side effect, $w_b$ acts as a regularizer to prevent the RMI from growing unnecessarily large.
We found that our simple cost model performed well throughout our evaluation.
However, it may still be beneficial to formulate a more complex cost model that more accurately reflects true runtime; this is left as future work.

\subsubsection{Cost Computation Performance}
\label{sec:cost_performance}
The cost of the entire RMI, $C(\mathcal{A})$, is never explicitly computed.
Instead, all decisions based on cost are made locally.
This is possible due to the linearity of the cost model.
For example, when deciding between expanding a data node and splitting the data node in two, we compare the \emph{incremental} impact on $C(\mathcal{A})$ between the two options.
This only involves computing the intra-node cost of the expanded data node and each of the two split data nodes; the intra-node costs of all other data nodes in the RMI remain the same.

Cost computation occurs at two points during \atlex operation (\cref{subsubsec:insertion}): (1) when a data node becomes full, the expected intra-node cost is compared to the empirical intra-node cost to check for cost deviation. This comparison has very low performance overhead because the empirical values of $S(N)$ and $I(N)$ are maintained by the data node, so computing the empirical intra-node cost merely involves three multiplications and an addition. (2) If cost deviation is detected, \atlex must make a cost-based decision about how to adjust the RMI structure.
This involves computing the expected intra-node cost of candidate data nodes which may be created as a result of adjusting the RMI structure.
Since the candidate data nodes do not yet exist, we must compute the expected $S(N)$ and $I(N)$, which involves implicitly building the candidate data node.
The majority of time spent on cost-based decision making is spent on computing expected $S(N)$ and $S(I)$.

\begin{table}[]
\centering
\caption{Fraction of time spent on cost computation}
\vspace{-1em}
\small
\label{tab:cost_computation}
\begin{tabular}{@{}lllll@{}}
\toprule
       & \textbf{longitudes}     & \textbf{longlat}  & \textbf{lognormal} & \textbf{YCSB} \\
\midrule
\textbf{Read-Only} & 0 & 0 & 0 & 0 \\
\textbf{Read-Heavy} & 0.000271 & 0.000214 & 0.000617 & 0 \\
\textbf{Write-Heavy} & 0.00142 & 0.00901 & 0.00452 & 0.116 \\
\textbf{Write-Only} & 0.0270 & 0.0732 & 0.0237 & 0.149 \\
\bottomrule
\end{tabular}
\end{table}

\cref{tab:cost_computation} shows the fraction of overall workload time spent on computing costs and making cost-based decisions.
On the read-only workload, no time is spent on cost-based decision because nodes never become full.
As the fraction of writes increases, an increasing fraction of time is spent on cost computation because nodes become full more frequently.
However, even on the write-only workload, cost computation takes up a small fraction of overall time spent on the workload.
YCSB sees the highest fraction of time spent on cost computation, due to two factors: data nodes are larger, so computing $S(N)$ and $I(N)$ for larger candidate nodes takes more time, and lookups and inserts on YCSB are efficient, so data nodes become full more quickly.
Longlat sees the next highest fraction of time spent on cost computation, which is due to the high frequency with which data nodes become full (\cref{tab:node_splits}).

\subsection{Comparison of Gapped Array and PMA}
\label{sec:pma}
The Gapped Array structure introduced in \cref{subsec:node_layout} has some similarities to an existing data structure known as the Packed Memory Array (PMA)~\cite{bender2006adaptive}.
In this section, we first describe the PMA, and then we describe why we choose to not use the PMA within \atlex.

Like the Gapped Array, PMA is an array with gaps.
Unlike the Gapped Array, PMA is
designed to uniformly space its gaps between elements and
to maintain this property as new elements are inserted. The
PMA achieves this goal by rebalancing local portions of the
array when the gaps are no longer uniformly spaced. Under
random inserts from a static distribution, the PMA can insert
elements in $O(\log{n})$ time, which is the same as the Gapped
Array. However, when inserts do not come from a static
distribution, the PMA can guarantee worst-case insertion in
$O(\log^2{n})$ time, which is better than the worst case of the
Gappd Array, which is $O(n)$ time.

We now describe the PMA more concretely; more details
can be found in~\cite{bender2006adaptive}. The PMA is an array whose size is
always a power of 2. The PMA divides itself into equally spaced segments, and the number of segments is also a
power of 2. The PMA builds an implicit binary tree on
top of the array, where each segment is a leaf node, each
inner node represents the region of the array covered by
its two children, and the root node represents the entire
array. The PMA places density bounds on each node of this
implicit binary tree, where the density bound determines
the maximum ratio of elements to positions in the region
of the array represented by the node. The nodes nearer
the leaves will have higher density bounds, and the nodes
nearer the root will have lower density bounds. The density
bounds guarantee that no region of the array will become
too packed. If an insertion into a segment will violate the
segment's density bounds, then we can find some local region
of the array and uniformly redistribute all elements within
this region, such that after the redistribution, none of the
density bounds are violated. As the array becomes more full,
ultimately no local redistribution can avoid violating density
bounds. At this point, the PMA expands by doubling in size
and inserting all elements uniformly spaced in the expanded
array.

We do not use the PMA as the underlying storage structure for \atlex data nodes because the PMA negates the benefits of model-based inserts, which is critical for search performance.
For example, when rebalancing a local portion of the array, the PMA spreads
the keys in the local region over more space, which worsens search performance because the keys are moved
further away from their predicted location.
Furthermore, the main benefit of PMA---efficient inserts for non-static or complex key distributions---is already achieved by \atlex through the adaptive RMI structure.
In our evaluation, we found that \atlex using data nodes built on Gapped Arrays consistently outperformed data nodes built on PMA.

\subsection{Analysis of Model-based Search}
\label{sec:model_based_search_analysis}
Model-based inserts try to place keys in Gapped Array in their predicted positions. We analyze the trade-off between Gapped Array space usage and search performance in terms of $\expfactor$, the ratio of Gapped Array slots to number of actual keys.
Assume the keys in the data node are $x_1<x_2<\cdots<x_n$, and the linear model before rounding is $y=ax+b$ when $\expfactor=1$, i.e., when no extra space is allocated. Define $\delta_i=x_{i+1}-x_i,\Delta_i=x_{i+2}-x_i$. We first present a condition under which all the keys in that data node are placed in the predicted location, i.e., search for all keys are direct hits.

\begin{theorem}\label{th:perfect}
    When $\expfactor\ge\frac{1}{a\min_{i=1}^{n-1}\delta_i}$, every key in the data node is placed in the predicted location exactly.
\end{theorem}
\begin{proof}
	Consider two keys in the leaf node $x_i$ and $x_j,i\neq j$. The predicted locations before rounding are $y_i$ and $y_j$, respectively. 
	When $|y_i-y_j|\ge 1$, we know that the rounded locations $\lfloor y_i\rfloor$ and $\lfloor y_j \rfloor$ cannot be equal. Under the linear model $y=\expfactor(ax+b)$, we can write the condition as:
	\begin{align}\label{eq:cond_perfect}
		|y_i-y_j|=|\expfactor a(x_i-x_j)|\ge 1
	\end{align}
	If this condition is true for all the pairs $(i,j),i\neq j$, then all the keys will have a unique predicted location. For the condition Eq.~\eqref{eq:cond_perfect} to be true for all $i\neq j$, it suffices to have:
	\begin{align}
		\min_{i=1}^{n-1} \expfactor a(x_{i+1}-x_i)\ge 1
	\end{align}   
	which is equivalent to 
	$\expfactor\ge\frac{1}{a\min_{i=1}^{n-1}\delta_i}$.
\end{proof}
%Due to space limit, we leave proofs to our extended technical report (not cited due to double-blind).
% This result suggests that once the expansion factor is larger than $\frac{1}{a\min_{i=1}^{n-1}\delta_i}$, the search performance for the keys in this data node will not further improve. 
% In other words, 
We now understand that $c=1$ corresponds to the optimal space, and $c\ge \frac{1}{a\min_{i=1}^{n-1}\delta_i} = c_{max}$ corresponds to the optimal search time (ignoring the effect of cache misses). 
% What happens when $1\le c<\frac{1}{a\min_{i=1}^{n-1}\delta_i}$?
We now bound the number of keys with direct hits when $\expfactor<c_{max}$.
\begin{theorem}\label{th:upper}
    The number of keys placed in the predicted location is no larger than $2+\left|\{1\le i\le n-2| \Delta_i> \frac{1}{ca} \}\right|  $, where $\left|\{1\le i\le n-2| \Delta_i> \frac{1}{ca} \}\right|$ is the number of $\Delta_i$'s larger than $\frac{1}{ca}$.
\end{theorem}
\begin{proof}
	We define a mapping $f:[n-2]\rightarrow [n]$, where $f(i)$ is defined recursively according to the following cases:
	
	Case (1): $y_{i+2}-y_{i}>1$. Let $f(i)=1$. 
	Case (2): $y_{i+2}-y_{i}\le 1,\lfloor y_{i+1} \rfloor = \lfloor y_{i}\rfloor, f(i-1)\le i$ or $i=1$. Let $f(i)=i+1$.
	Case (3): Neither case (1) or (2) is true. Let $f(i)=i+2$.
	
	We prove that $\forall 1\le i<j\le n-2$, if $f(i)>1,f(j)>1$, then $i+1\le f(i)\le i+2, j+1\le f(j)\le j+2$, and $f(i)<f(j)$. 
	
	First, when $f(i)>1, f(j)>1$, we know that case (1) is false for both $i$ and $j$. So $f(i)$ is either $i+1$ or $i+2$, and $f(j)$ is either $j+1$ or $j+2$.
	
	Second, if $i+1<j$, then $f(i)\le i+2<j+1\le f(j)$. So we only need to prove $f(i)<f(j)$ when $i+1=j$. 
	Now consider the only two possible values for $f(j)$, $j+1$ and $j+2$, when $i+1=j$.     
	If $f(j)=j+1=i+2$, by definition we know that case (2) is true for $f(j)$. That means $f(j-1)=j$ or 1. But we already know $f(j-1)=f(i)>1$. So $f(i)=f(j-1)=j=i+1<i+2=f(j)$. If $f(j)=j+2$, then $f(i)\le i+2<j+2=f(j)$.
	
	So far, we have proved that $f(i)$ is unique when $f(i)>1$. Now we prove that the key $x_{f(i)}$ is not placed in $\lfloor y_{f(i)}\rfloor$ when $f(i)>1$, i.e., either case (2) or case (3) is true for $f(i)$. In both cases, $y_{i+2}-y_{i}\le 1$, and the rounded integers $\lfloor y_{i+2}\rfloor$ and $\lfloor y_{i}\rfloor$ must be either equal or adjacent: 
	$\lfloor y_{i+2}\rfloor - \lfloor y_{i}\rfloor\le 1 $. 
	That means $\lfloor y_{i+1} \rfloor$ must be equal to either $\lfloor y_{i+2}\rfloor $ or $\lfloor y_{i}\rfloor$. 
	
	We prove by mathematical induction. For the minimal $i$ s.t. $f(i)>1$, if case (2) is true, $\lfloor y_{i+1} \rfloor = \lfloor y_{i}\rfloor$. That means $x_{i+1}$ cannot be placed at $\lfloor y_{i+1} \rfloor$ because that location is already occupied before $x_{i+1}$ is inserted. And $f(i)=i+1$ by definition. If case (2) is false, since we already know $y_{i+2}-y_{i}\le 1$, $f(i-1)=1$ or $i=1$, it follows that $\lfloor y_{i+1} \rfloor > \lfloor y_{i}\rfloor$. That implies $\lfloor y_{i+1} \rfloor = \lfloor y_{i+2}\rfloor$. So $x_{i+2}$ cannot be placed at $\lfloor y_{i+2} \rfloor$. And $f(i)=i+2$ because case (3) happens.
	
	Given that the key $x_{f(i-1)}$ is not placed at $\lfloor y_{f(i-1)}\rfloor$ when $f(i-1)>1$, we now prove it is also true for $i$. The proof for case (2) is the same as above. If case (2) is false, and $\lfloor y_{i+1} \rfloor > \lfloor y_{i}\rfloor$, the proof is also the same as above. The remaining possibility of case (3) is that $\lfloor y_{i+1} \rfloor = \lfloor y_{i}\rfloor$, and $f(i-1)=i+1$. The inductive hypothesis states that $x_{i+1}$ is not placed at $\lfloor y_{i+1} \rfloor$. That means $x_{i+1}$ is placed at a location equal or larger than $\lfloor y_{i+1} \rfloor+1 = \lfloor y_{i}\rfloor+1$. But we also know that $\lfloor y_{i+2}\rfloor \le \lfloor y_{i}\rfloor+ 1 $. So $x_{i+2}$ cannot be placed at $\lfloor y_{i+2}\rfloor$ which is not on the right of $x_{i+1}$'s location. Since case (3) is false, $f(i)=i+2$.
	
	By induction, we show that when $f(i)>1$, the key $x_{f(i)}$ cannot be placed at $\lfloor y_{f(i)}\rfloor$. That means when we look up $x_{f(i)}$, we cannot directly hit it from the model prediction. Since we also proved that $f(i)$ has a unique value when $f(i)>1$, the number of misses from the model prediction is at least the size of $S=\{i\in [n-2]|f(i)>1\}$. By the definition of $f(i)$,  $S=\{i\in [n-2]|y_{i+2}-y_{i}\le 1\}$. Therefore, the number of direct hits by the model is at most $n-|S|=2+\left|\{i\in [n-2]|y_{i+2}-y_{i}> 1 \} \right|=2+\left|\{1\le i\le n-2| \Delta_i\ge \frac{1}{\expfactor a} \}\right| $.
	% Consider three adjacent keys $x_i,x_{i+1},x_{i+2}, 1\le i\le n-2$, and their predicted locations $y_i,y_{i+1},y_{i+2}$. 

\end{proof}

% The proof can be found in Appendix~\ref{sec:appendix_proof}.
This result presents an upper bound on the number of direct hits from the model, which is positively correlated with $\expfactor$.
% The smaller the $\expfactor$, the fewer the direct hits, requiring more comparisons.
% We can use this result to determine a lower bound on $\expfactor$ such that the upper bound of direct hits is above a threshold.
This upper bound also applies to the \kraskali, which has $\expfactor=1$. This explains why the Gapped Array has the potential to dramatically decrease the search time. Similarly, we can lower bound the number of direct hits.
\begin{theorem}
    The number of keys placed in the predicted location is no smaller than $l+1$, where $l$ is the largest integer such that $\forall 1\le i\le l, \delta_i\ge \frac{1}{\expfactor a} $, \i.e., the number of consecutive $\delta_i$'s from the beginning equal or larger than $\frac{1}{\expfactor a}$.
\end{theorem} 
The proof is not hard based on the ideas from the previous two proofs. 
% The result is an improvement of Theorem~\ref{th:perfect}, and can be used to reduce $\expfactor$ when perfect lookup is not required.
% This bound is approximate, and it is difficult to derive an exact lower bound beyond Theorem~\ref{th:lower}, because 